\documentclass{article}
\usepackage{kr}
\usepackage{times}
\usepackage{soul}
\usepackage{url}
\usepackage[hidelinks]{hyperref}
\usepackage[utf8]{inputenc}
\usepackage[small]{caption}
\usepackage{booktabs}
\usepackage{enumerate}

\usepackage{graphicx} 
\usepackage{amsmath,amssymb,amsthm}
\usepackage{xspace}
\usepackage[textsize=small]{todonotes}
\usepackage{ifsym}
\usepackage{tabularx, environ}
\usepackage{lineno}
\usepackage{colortbl}

\usepackage{stmaryrd}

\usepackage{amsthm}

\newcommand{\eval}[2]{\llbracket #1 \rrbracket^{#2}}
\newcommand{\atleast}[3]{{\geq_{#1}\,}{#2}{.#3}}

\newcommand{\ivar}{\mathsf{vars}}

\newcommand{\idom}{\mathsf{dom}}
\newcommand{\rsem}{\ensuremath{\rrbracket}}
\newcommand{\lsem}{\ensuremath{\llbracket}}
\newcommand{\OPT}{\ensuremath{\mathbin{\text{OPT}}}}

\newcommand{\sem}[1]{\llbracket #1 \rrbracket}
\newcommand{\semg}[2]{\sem{#1}_{#2}}
\newcommand{\semG}[1]{\semg{#1}{G}}

\newcommand{\problemdef}[3]{
\begin{center}
\setlength{\tabcolsep}{2pt}
\begin{tabular}{r p{19em} }
Problem: & #1 \\ &\\[\dimexpr-\normalbaselineskip +5pt]
Input:     & #2.\\
Question:  & #3?\\
\end{tabular}
\end{center}
}

\AtBeginEnvironment{align}{\setcounter{equation}{0}}

\makeatletter
\newcommand{\problemtitle}[1]{\gdef\@problemtitle{#1}}
\newcommand{\probleminstance}[1]{\gdef\@probleminstance{#1}}
\newcommand{\problemquestion}[1]{\gdef\@problemquestion{#1}}
\NewEnviron{problem}{
  \problemtitle{}\probleminstance{}\problemquestion{}
  \BODY
  \par\addvspace{.5\baselineskip}
  \noindent
  \begin{tabularx}{\textwidth}{@{\hspace{\parindent}} r X c}
    \multicolumn{2}{@{\hspace{\parindent}}c}{\@problemtitle} \\
    \textbf{Instance:} & \@probleminstance \\
    \textbf{Question:} & \@problemquestion
  \end{tabularx}
  \par\addvspace{.5\baselineskip}
}
\makeatother

\newtheorem{theorem}{Theorem}

\newtheorem{lemma}{Lemma}

\newtheorem{theoremAppendix}{Theorem}

\newtheorem{definition}{Definition}
\newtheorem{example}{Example}

\renewcommand{\phi}{\varphi}
\newcommand{\texttti}[1]{\texttt{\textit{#1}}}

\newcommand{\ptime}{\ensuremath{\mathsf{P}}\xspace}
\newcommand{\np}{\ensuremath{\mathsf{NP}}\xspace}
\newcommand{\conp}{\ensuremath{\mathsf{coNP}}\xspace}
\newcommand{\DP}{\ensuremath{\mathsf{DP}}\xspace}
\newcommand{\DPtwo}{\ensuremath{\mathsf{DP}_2}\xspace}
\newcommand{\sigmatwo}{\ensuremath{\mathsf{\Sigma}_2\mathsf{P}}\xspace}
\newcommand{\pitwo}{\ensuremath{\mathsf{\Pi}_2\mathsf{P}}\xspace}

\newcommand{\thetatwo}{\ensuremath{\mathsf{\Theta}_2\mathsf{P}}\xspace}
\newcommand{\thetak}{\ensuremath{\mathsf{\Theta}_{k+1}\mathsf{P}}\xspace}
\newcommand{\thetathree}{\ensuremath{\mathsf{\Theta}_3\mathsf{P}}\xspace}

\newcommand{\sigmak}{\ensuremath{\mathsf{\Sigma}_k\mathsf{P}}\xspace}
\newcommand{\pithree}{\ensuremath{\mathsf{\Pi}_3\mathsf{P}}\xspace}
\newcommand{\bh}{\ensuremath{\mathsf{BH}}\xspace}
\newcommand{\bhone}{\ensuremath{\mathsf{BH}_1}\xspace}
\newcommand{\bhtwo}{\ensuremath{\mathsf{BH}_2}\xspace}
\newcommand{\bhk}{\ensuremath{\mathsf{BH}_k}\xspace}

\newcommand{\bhtwok}{\ensuremath{\mathsf{BH}_{2k}}\xspace}

\newcommand{\bhtwokplus}{\ensuremath{\mathsf{BH}_{2k+1}}\xspace}

\newcommand{\nop}[1]{}

\newcommand{\cqa}{\textsc{CQA}\xspace}
\newcommand{\bgpEqu}{\textsc{CQA(Bgp$,{=},\forall$)}\xspace}
\newcommand{\bgpSubs}{\textsc{CQA(Bgp$,\subseteq,\forall$)}\xspace}
\newcommand{\bgpCard}{\textsc{CQA(Bgp$,\leq,\forall$)}\xspace}

\newcommand{\bgpProjEqu}{\textsc{CQA($\pi$-Bgp$,{=},\forall$)}\xspace}
\newcommand{\bgpProjSubs}{\textsc{CQA($\pi$-Bgp$,\subseteq,\forall$)}\xspace}
\newcommand{\bgpProjCard}{\textsc{CQA($\pi$-Bgp$,\leq,\forall$)}\xspace}

\newcommand{\wdptEqu}{\textsc{CQA(wdQ$,=,\forall$)}\xspace}
\newcommand{\wdptSubs}{\textsc{CQA(wdQ$,\subseteq,\forall$)}\xspace}
\newcommand{\wdptCard}{\textsc{CQA(wdQ$,\leq,\forall$)}\xspace}

\newcommand{\wdptProjEqu}{\textsc{CQA($\pi$-wdQ,$=,\forall$)}\xspace}
\newcommand{\wdptProjSubs}{\textsc{CQA($\pi$-wdQ,$\subseteq,\forall$)}\xspace}
\newcommand{\wdptProjCard}{\textsc{CQA($\pi$-wdQ,$\leq,\forall$)}\xspace}

\newcommand{\bgpEquEx}{\textsc{CQA(Bgp$,=, \exists$)}\xspace}
\newcommand{\bgpSubsEx}{\textsc{CQA(Bgp$,\subseteq, \exists$)}\xspace}
\newcommand{\bgpCardEx}{\textsc{CQA(Bgp$,\leq, \exists$)}\xspace}

\newcommand{\bgpProjEquEx}{\textsc{CQA($\pi$-Bgp$,=, \exists$)}\xspace}

\newcommand{\bgpProjCardEx}{\textsc{CQA($\pi$-Bgp$,\leq, \exists$)}\xspace}

\newcommand{\wdptEquEx}{\textsc{CQA(wdQ$,=, \exists$)}\xspace}

\newcommand{\wdptCardEx}{\textsc{CQA(wdQ$,\leq, \exists$)}\xspace}

\newcommand{\wdptProjEquEx}{\textsc{CQA($\pi$-wdQ$,=, \exists$)}\xspace}
\newcommand{\wdptProjSubsEx}{\textsc{CQA($\pi$-wdQ$,\subseteq, \exists$)}\xspace}
\newcommand{\wdptProjCardEx}{\textsc{CQA($\pi$-wdQ$,\leq, \exists$)}\xspace}

\newcommand{\bgpEquIn}{\textsc{CQA(Bgp$,=, \cap$)}\xspace}

\newcommand{\bgpProjEquIn}{\textsc{CQA($\pi$-Bgp$,=, \cap$)}\xspace}
\newcommand{\bgpProjSubsIn}{\textsc{CQA($\pi$-Bgp$,\subseteq, \cap$)}\xspace}
\newcommand{\bgpProjCardIn}{\textsc{CQA($\pi$-Bgp$,\leq, \cap$)}\xspace}

\newcommand{\wdptEquIn}{\textsc{CQA(wdQ$,=, \cap$)}\xspace}
\newcommand{\wdptSubsIn}{\textsc{CQA(wdQ$,\subseteq, \cap$)}\xspace}
\newcommand{\wdptCardIn}{\textsc{CQA(wdQ$,\leq, \cap$)}\xspace}

\newcommand{\wdptProjEquIn}{\textsc{CQA($\pi$-wdQ$,=, \cap$)}\xspace}
\newcommand{\wdptProjSubsIn}{\textsc{CQA($\pi$-wdQ$,\subseteq, \cap$)}\xspace}
\newcommand{\wdptProjCardIn}{\textsc{CQA($\pi$-wdQ$,\leq, \cap$)}\xspace}

\newcommand{\bgpPrecIn}{\textsc{CQA(Bgp$,\preceq, \cap$)}\xspace}

\newcommand{\bgpProjPrecIn}{\textsc{CQA($\pi$-Bgp$,\preceq, \cap$)}\xspace}

\newcommand{\wdptPrecIn}{\textsc{CQA(wdQ$,\preceq, \cap$)}\xspace}

\newcommand{\wdptProjPrec}{\textsc{CQA($\pi$-wdQ$,\preceq, \forall$)}\xspace}

\newcommand{\wdptProjPrecEx}{\textsc{CQA(\wdptProj$,\preceq, \exists$)}\xspace}

\newcommand{\LEqu}{\textsc{CQA($\mathcal{L}, =,\forall$)}\xspace}
\newcommand{\LSubs}{\textsc{CQA($\mathcal{L}, \subseteq,\forall$)}\xspace}
\newcommand{\LCard}{\textsc{CQA($\mathcal{L}, \leq,\forall$)}\xspace}

\newcommand{\LEquEx}{\textsc{CQA($\mathcal{L}, =, \exists $)}\xspace}
\newcommand{\LSubsEx}{\textsc{CQA($\mathcal{L}, \subseteq,\exists$)}\xspace}
\newcommand{\LCardEx}{\textsc{CQA($\mathcal{L}, \leq,\exists$)}\xspace}

\newcommand{\LEquIn}{\textsc{CQA($\mathcal{L}, =,\cap$)}\xspace}
\newcommand{\LSubsIn}{\textsc{CQA($\mathcal{L}, \subseteq,\cap$)}\xspace}
\newcommand{\LCardIn}{\textsc{CQA($\mathcal{L}, \leq,\cap$)}\xspace}

\newcommand{\BEqu}{\textsc{CQA($\bgp, =,\forall$)}\xspace}
\newcommand{\BEquEx}{\textsc{CQA($\bgp, =, \exists $)}\xspace}
\newcommand{\LPrecEx}{\textsc{CQA($\mathcal{L}, \preceq , \exists $)}\xspace}

\newcommand{\BPrec}{\textsc{CQA($\bgp ,\preceq, \forall$)}\xspace}

\newcommand{\PBPrec}{\textsc{CQA($\bgpProj ,\preceq, \forall$)}\xspace}
\newcommand{\PWPrec}{\textsc{CQA($\wdptProj ,\preceq, \forall$)}\xspace}

\newcommand{\PBEquIn}{\textsc{CQA($\bgpProj, =,\cap$)}\xspace}
\newcommand{\PBCardIn}{\textsc{CQA($\bgpProj, \leq,\cap$)}\xspace}
\newcommand{\WPrec}{\textsc{CQA($\wdpt ,\preceq, \forall$)}\xspace}
\newcommand{\PWEqu}{\textsc{CQA($\wdptProj, =, \forall$)}\xspace}

\newcommand{\PBSubset}{\textsc{CQA($\bgpProj ,\subseteq, \forall$)}\xspace}

\newcommand{\mcqa}{\textsc{mCQA}\xspace}

\newcommand{\mwdptEqu}{\textsc{mCQA(wdQ$,=,\forall$)}\xspace}

\newcommand{\mwdptProjEqu}{\textsc{mCQA($\pi$-wdQ$,=,\forall$)}\xspace}

\newcommand{\mbgpProjEquEx}{\textsc{mCQA($\pi$-$\textsc{Bgp}, =,\exists$)}\xspace}

\newcommand{\mwdptProjEquEx}{\textsc{mCQA($\pi$-wdQ$,=, \exists$)}\xspace}

\newcommand{\mbgpEquIn}{\textsc{mCQA(Bgp$,=, \cap$)}\xspace}

\newcommand{\mbgpProjEquIn}{\textsc{mCQA($\pi$-Bgp$,=, \cap$)}\xspace}

\newcommand{\mwdptEquIn}{\textsc{mCQA(wdQ$,=, \cap$)}\xspace}

\newcommand{\mbgpEquQ}{\textsc{mCQA($\bgp,=, \mathcal{S}$)}\xspace}
\newcommand{\mLPrecQ}{\textsc{mCQA($\mathcal{L}, {\preceq},\mathcal{S}$)}\xspace}
\newcommand{\LPrecQ}{\textsc{CQA($\mathcal{L}, \preceq, \mathcal{S}$)}\xspace}

\newcommand{\bgp}{\textsc{Bgp}\xspace}
\newcommand{\bgps}{\textsc{Bgp}s\xspace}
\newcommand{\bgpProj}{\textsc{$\pi$-Bgp}\xspace}
\newcommand{\bgpProjs}{\textsc{$\pi$-Bgp}s\xspace}
\newcommand{\wdpt}{\textsc{wdQ}\xspace}
\newcommand{\wdpts}{\textsc{wdQ}s\xspace}
\newcommand{\wdptProj}{\textsc{$\pi$-wdQ}\xspace}
\newcommand{\wdptProjs}{\textsc{$\pi$-wdQ}s\xspace}

\newcommand{\lowerupperb}[1]{{#1} & $\textcolor{purple}{\blacktriangle\blacktriangledown}$}
\newcommand{\lowerupperg}[1]{{#1} & $\textcolor{olive}{\blacktriangle\blacktriangledown}$}
\newcommand{\lowerupperbl}[1]{{#1} & $\textcolor{black}{\blacktriangle\blacktriangledown}$}
\newcommand{\loweruppert}[1]{{#1} & $\textcolor{teal}{\blacktriangle\blacktriangledown}$}
\newcommand{\lowerr}[1]{{#1} & $\textcolor{violet}{\blacktriangle\phantom{\blacktriangledown}}$}
\newcommand{\lowerb}[1]{{#1} & $\textcolor{purple}{\blacktriangle\phantom{\blacktriangledown}}$}
\newcommand{\lowerg}[1]{{#1} & $\textcolor{olive}{\blacktriangle\phantom{\blacktriangledown}}$}
\newcommand{\lowerbl}[1]{{#1} & $\textcolor{black}{\blacktriangle\phantom{\blacktriangledown}}$}
\newcommand{\lowert}[1]{{#1} & $\textcolor{teal}{\blacktriangle\phantom{\blacktriangledown}}$}
\newcommand{\upperr}[1]{{#1} & $\phantom{\blacktriangle}\textcolor{violet}{\blacktriangledown}$}
\newcommand{\upperb}[1]{{#1} & $\phantom{\blacktriangle}\textcolor{purple}{\blacktriangledown}$}
\newcommand{\upperg}[1]{{#1} & $\phantom{\blacktriangle}\textcolor{olive}{\blacktriangledown}$}
\newcommand{\upperbl}[1]{{#1} & $\phantom{\blacktriangle}\textcolor{black}{\blacktriangledown}$}
\newcommand{\uppert}[1]{{#1} & $\phantom{\blacktriangle}\textcolor{teal}{\blacktriangledown}$}

\newcommand{\tableOfResults}{
\begin{table*}
\centering
\setlength{\tabcolsep}{3.5pt}
\begin{tabular}{r|rl|rl|rl|rl|rl|rl|rl|rl|rl}
$\mathcal{L}\,$ \textbackslash $\, \preceq,\mathcal{S}$ & \multicolumn{2}{c|}{$=,\exists$} & \multicolumn{2}{c|}{$\leq,\exists$} & \multicolumn{2}{c|}{$\subseteq,\exists$} & \multicolumn{2}{c|}{$=,\forall$} & \multicolumn{2}{c|}{$\leq,\forall$} & \multicolumn{2}{c|}{$\subseteq,\forall$} & \multicolumn{2}{c|}{$=,\cap$} & \multicolumn{2}{c|}{$\leq,\cap$} & \multicolumn{2}{c}{$\subseteq,\cap$} \\
\hline &&&&&&&&&&&&&&&&&& \\[\dimexpr-\normalbaselineskip+2pt]
\textsc{Bgp} (DC)        & \lowerr{\np}      & \lowerr{\thetatwo}     & \lowerr{\sigmatwo}     & \lowerr{\conp}       & \lowerr{\thetatwo}      & \lowerr{\pitwo}       & \lowerg{\conp}            & \lowerg{\thetatwo}            & \lowerg{\pitwo}  \\
$\pi$-\textsc{Bgp} (DC)  & \np &              & \thetatwo &            & \sigmatwo &            & \conp &               & \thetatwo &            & \pitwo &              & \upperg{\conp}               & \thetatwo &                 & \pitwo &             \\
\textsc{wdQ} (DC)       & \np &              & \thetatwo &            & \sigmatwo &            & \conp &               & \thetatwo &            & \pitwo &              & \lowerupperbl{\DP}            & \thetatwo &                 & \lowerbl{\DPtwo}      \\
$\pi$-\textsc{wdQ} (DC) & \upperr{\np}        & \upperr{\thetatwo}      & \sigmatwo &            & \upperr{\conp}         & \upperr{\thetatwo}      & \upperr{\pitwo}        & \loweruppert{\thetatwo}      & \upperg{\thetatwo}           & \lowert{\thetathree} \\
\hline &&&&&&&&&&&&&&&&&& \\[\dimexpr-\normalbaselineskip+2pt]
\textsc{Bgp} (CC)        & \np &              & \thetatwo &            & \sigmatwo &            & \conp &               & \thetatwo &            & \pitwo &              & \upperg{\conp}               & \thetatwo &                 & \pitwo &             \\
$\pi$-\textsc{Bgp} (CC)  & \upperr{\np}        & \upperr{\thetatwo}      & \sigmatwo &            & \lowerupperb{\pitwo}   & \lowerupperb{\pitwo}    & \upperr{\pitwo}        & \lowerupperg{\thetatwo}      & \upperg{\thetatwo}           & \upperg{\pitwo}       \\
\textsc{wdQ} (CC)       & \lowerb{\sigmatwo} & \lowerb{\sigmatwo}     & \sigmatwo &            & \upperr{\conp}         & \upperr{\thetatwo}      & \upperr{\pitwo}        & \lowerupperbl{\thetatwo}      & \upperg{\thetatwo}           & \upperbl{\DPtwo}       \\
$\pi$-\textsc{wdQ} (CC) & \upperb{\sigmatwo}  & \upperb{\sigmatwo}      & \upperr{\sigmatwo}      & \lowerupperb{\pithree} & \lowerupperb{\pithree}  & \lowerupperb{\pithree} & \loweruppert{\sigmatwo} & \loweruppert{\sigmatwo}     & \uppert{\thetathree}  \\
\end{tabular}
\caption{Complexity Results of the \LPrecQ problem for 
query language $\mathcal{L} \in \{\bgp, \bgpProj, \wdpt, \wdptProj\}$, 
preference order ${\preceq} \in \{$=,\,$\leq$,\,$\subseteq$$\}$, and 
semantics $\mathcal{S}\in \{\exists, \forall, \cap\}$
as abbreviations for brave, AR and IAR semantics, respectively.
Moreover \emph{DC} and \emph{CC} stand for data and combined complexity.
The up- and down-triangles ($\blacktriangle \blacktriangledown$) indicate
hardness and membership proofs, respectively, in this work. 
All complexity classifications are completeness results but not all have to 
be proved separately, since hardness results carry over from more special to more general cases and 
membership results carry over in the opposite direction.   
Colors indicate the theorem where a (hardness or membership) result is proved, 
namely Theorem \ref{thm:ARbgp} 
($\textcolor{violet}{\blacktriangle\blacktriangledown}$),
Theorem \ref{thm:ARrest} 
($\textcolor{purple}{\blacktriangle\blacktriangledown}$),
Theorem \ref{thm:easyIAR}
($\textcolor{olive}{\blacktriangle\blacktriangledown}$),
Theorem \ref{thm:wdptIAR}
($\textcolor{black}{\blacktriangle\blacktriangledown}$),
and
Theorem \ref{prop:wdptProj} 
($\textcolor{teal}{\blacktriangle\blacktriangledown}$).
}
\label{tab:results}
\end{table*}
}

\newcommand{\cremove}[1]{
}

\allowdisplaybreaks

\newcommand{\myparagraph}[1]{\medskip\noindent\textit{#1}~}

\title{Consistent Query Answering over SHACL Constraints}
\author{Shqiponja \& Camillo \&  Reinhard}
\date{\today}

\author{%
Shqiponja Ahmetaj$^1$\and
Timo Camillo Merkl$^1$\and
Reinhard Pichler$^{1}$\\
\affiliations
$^1$TU Wien\\
\emails
\{shqiponja.ahmetaj, timo.merkl, reinhard.pichler\}@tuwien.ac.at
}

\begin{document}

\maketitle

\begin{abstract}
The Shapes Constraint Language (SHACL) was standardized by the World Wide Web as a constraint language to describe and validate RDF data graphs. SHACL uses the notion of shapes graph to describe a set of \emph{shape} constraints paired with \emph{targets}, that specify which nodes of the RDF graph should satisfy which shapes. 
An important question in practice is how to handle data graphs that do not validate the shapes graph. 
A solution is to \emph{tolerate} the non-validation and find ways to obtain meaningful and correct answers to queries despite the non-validation. 
This is known as \emph{consistent query answering} (CQA) and there is extensive literature on CQA in both the database and the KR setting. 
We study CQA in the context of SHACL for a fundamental fragment of the Semantic Web query language SPARQL. 
The goal of our work is a detailed complexity analysis of CQA for various semantics and possible restrictions on the acceptable repairs. 
It turns out that all considered variants of the problem are intractable, with complexities ranging between the first and third level of the polynomial hierarchy.
\end{abstract}

\section{Introduction}
\label{sect:introduction}
The SHACL Shapes Constraint Language is the World Wide Web recommendation language for expressing and validating constraints on RDF data graphs~\cite{shacl}. The normative standard emerged as a need for a \emph{prescriptive} language to provide guarantees on the structure and content of RDF data
\cite{W3CValidationWorkshopReport}. SHACL uses the notion of \emph{shapes graph} to describe a set of \emph{shape} constraints paired with \emph{targets}, that specify which nodes of the RDF graph should satisfy which shapes. The main computational problem in SHACL is to check whether an RDF graph \emph{validates} a shapes graph. However, the W3C specification does not clearly define, and sometimes leaves undefined, certain aspects of validation, such as the semantics of validation for \emph{recursive} constraints, which involve cyclic dependencies. A formal logic-based syntax and semantics for (recursive) SHACL 
was 
proposed in~\cite{DBLP:conf/semweb/CormanRS18}, and 
has served as bases for most subsequent works
~\cite{DBLP:conf/www/AndreselCORSS20,DBLP:conf/ecai/Ahmetaj0S23,DBLP:conf/lpnmr/BogaertsJB22}. 

SHACL is particularly close to expressive Description Logics (DLs), the logics underlying OWL, as already observed in several 
works~\cite{DBLP:conf/lpnmr/BogaertsJB22,%
DBLP:conf/semweb/LeinbergerSRLS20,%
DBLP:conf/wollic/Ortiz23,shqi-etal-2021kr}. 
Specifically, since SHACL adopts the closed world assumption, it is closely related to an extension of the DL $\mathcal{ALCOIQ}$ with regular role expressions and equalities, where all roles and some concept names are viewed as \emph{closed predicates} (see e.g.~\cite{DBLP:conf/ijcai/LutzSW13}).

It is commonly recognized that real-world data, and in particular, graph-structured data and large RDF triple stores, which are subject to frequent change, may be incomplete or contain faulty facts. It thus seems inevitable to expect a data graph to not validate a SHACL shapes graph, e.g., because it is missing some facts to validate a target or it has conflicting or contradictory facts. The question of how to handle such data graphs is very relevant in practice. A possible solution is to fix the data graph before reasoning. In the style of database \emph{repairs}, \cite{shqi-etal-2021kr,DBLP:conf/semweb/AhmetajDPS22} propose to fix the data graph through (minimal) additions or removals of facts such that the resulting data graph validates the shapes graph. 
However, this may not always be desirable in practice as there may be a large number of possible repairs and selecting one may keep wrong facts, or remove true facts.

Another alternative is to \emph{tolerate} the non-validation and find ways to leverage the consistent part of the data and obtain meaningful and correct answers to queries despite the non-validation. This view is known as \emph{consistent query answering} (CQA),
and it has gained a lot of attention since the seminal paper \cite{DBLP:conf/pods/ArenasBC99}.
The idea is to accept as answers to a query those that are true over all (minimal) repairs of the input database. This is also known as the \emph{AR} semantics~\cite{DBLP:conf/pods/ArenasBC99,DBLP:series/synthesis/2011Bertossi,DBLP:conf/rr/LemboLRRS10,DBLP:conf/icdt/ArmingPS16}. 
Several other inconsistency-tolerant semantics have also been studied such as \emph{brave}~\cite{DBLP:conf/ijcai/BienvenuR13} and \emph{IAR} \cite{DBLP:conf/rr/LemboLRRS10} semantics. The former accepts answers that are true in \emph{some} repair, 
and 
the latter accepts the most reliable answers, that is those that are true in the \emph{intersection of all} repairs. CQA has been extensively studied 
in various database and knowledge representation settings;
we refer to \cite{DBLP:series/synthesis/2011Bertossi,DBLP:journals/sigmod/Wijsen19,DBLP:conf/rweb/BienvenuB16} for nice surveys.

In this work, we focus on CQA in the presence of (recursive) SHACL shapes, which to our knowledge, has not yet been explored.
As a query language we consider a fundamental fragment of SPARQL, which is the standardized language to query RDF data. Specifically, we  focus on \emph{basic graph patterns} (BGPs), which are essentially conjunctive queries (CQs),
and the well-behaved extension with the OPTIONAL operator -- the so-called \emph{well-designed} fragment of SPARQL)~\cite{DBLP:journals/tods/PerezAG09},
referred to as {\em well-designed queries}
(\wdpts, for short) in this paper.

Our main goal is a detailed complexity analysis of the CQA problem.
We thus build on~\cite{shqi-etal-2021kr}, which analyzes the complexity of the main reasoning problems for repairs w.r.t.\  SHACL constraints, 
such as checking the existence of a repair and deciding 
if a particular fact is added or deleted in at least one or in all repairs. These problems 
were further refined by restricting to repairs that are minimal 
w.r.t.\ cardinality or set inclusion.  Of course, these restrictions are
also highly relevant for the CQA problem. In total, we will thus study 
numerous variants of the CQA problem by considering 4 query languages
(\bgps and \wdpts, with or without projection),  
under 3 semantics (brave, AR, IAR), with or without (cardinality or subset inclusion) minimality-restrictions. Moreover, we distinguish data complexity (where the SHACL constraints and the query are considered as fixed and only 
the data is allowed to vary) and combined complexity.
We refer to Table~\ref{tab:results} for an overview of our main results. 
Formal definitions of all terms appearing in this table are given 
in Sections~\ref{sect:preliminaries}
and
\ref{sect:queryingNonValidData}.

We note that the settings considered for CQA in the literature crucially differ from ours in several respects. 
The typical constraint languages studied for databases are (fragments of) 
tuple generating dependencies (tgds, i.e., rules with conjunctive queries (CQs) in both the body and in the head) 
and equality generating dependencies (egds, i.e., rules with a CQ in the body and an equality in the head), 
see e.g., 
\cite{DBLP:conf/icdt/AfratiK09,DBLP:conf/icdt/ArmingPS16,DBLP:journals/mst/CateFK15}. 
SHACL also has implications as the crucial building blocks. However, in contrast with these dependencies, SHACL allows (among other features not present in tgds and egds) {\em (unrestricted) negation} in the rule body, which has significant semantical implications.  There is also a large body of works on CQA in the context of Description Logic knowledge bases, see e.g.,~\cite{DBLP:conf/aaai/Bienvenu12,DBLP:journals/jair/BienvenuBG19,DBLP:journals/ws/LemboLRRS15,DBLP:journals/kais/DuQS13}. However, to our knowledge, there are no works studying CQA for (expressive) DLs with closed predicates. Finally, the query languages considered in the literature focus on CQs with a few works considering extensions such as tree/path queries~\cite{DBLP:conf/pods/KoutrisOW21,DBLP:journals/pacmmod/KoutrisOW24}, datalog \cite{DBLP:journals/mst/KoutrisW21}, and counting \cite{DBLP:conf/icdt/KhalfiouiW23}. 
Apart from CQs (i.e., BGPs), we also study {\em non-monotonic queries} in the form of wdQs. To the best of our knowledge, 
wdQs have not been considered in the context of CQA. 

In this paper, we proceed as follows: 
We initially investigate the above mentioned problems, considering the scenario that a repair always exists. Indeed, it seems plausible to assume that a SHACL shapes graph is carefully designed so that no conflicting constraints are introduced. 
If this is not guaranteed, the existence of a repair 
can be tested with \np-power as shown in \cite{shqi-etal-2021kr}.
In case of a negative outcome of this test, it may still be possible to provide a repair that validates a subset of the targets. To address this, in the spirit of inconsistency 
tolerance,~\cite{DBLP:conf/semweb/AhmetajDPS22} proposes a relaxed notion of repairs, which 
aims at validating a maximal subset of the targets. We also study the complexity of CQA over {\em maximal repairs}.

Our main contributions are summarized as follows:

\begin{itemize}
    \item We first study the complexity of CQA in settings where we can validate all targets.
    It turns out that brave and AR semantics behave very similarly in terms of algorithms (to establish membership results) and in terms of methods for proving hardness results. We therefore study these cases simultaneously in 
    Section~\ref{sect:brave}. 
    \item In Section~\ref{sect:iar}, we study the complexity for IAR semantics. 
    It turns out that the influence of choosing different query languages and/or minimality conditions on the repairs gives a yet more colorful picture than with brave and AR semantics.
    \item Finally, in Section \ref{sect:max}, we extend our complexity analysis to the settings where the data cannot be fully repaired.
    Thus, we resort back to validating as many targets as possible. 
    It turns out that this increases the complexity beyond the first level of the polynomial hierarchy. 
    However, in all cases in Table~\ref{tab:results} with complexity of \thetatwo or above, the complexity classification remains the same.
\end{itemize}
In all cases, we provide a complete complexity classification in the form of matching upper and lower bounds. 
We note that all results hold for both non-recursive and recursive constraints.
Due to lack of space, proof details have to be omitted.  
Full proofs of all results presented here are given 
in the appendix.

\section{Preliminaries}
\label{sect:preliminaries}
In this section, we introduce RDF graphs, SHACL, 
\emph{validation} against RDF graphs, and (well-designed) SPARQL queries.  We follow the abstract syntax and semantics for the fragment of SHACL core studied in~\cite{shqi-etal-2021kr}; for more details on the W3C specification of SHACL core we refer to~\cite{shacl}. 

\myparagraph{RDF Graphs.}
We let 
$N_N$, $N_C$, $N_P$ denote countably infinite, mutually disjoint sets of \emph{nodes} (constants), \emph{class names}, and \emph{property names}, respectively. 
An RDF (data) graph $G$ is a finite set of (ground) \emph{atoms} of the form $B(\texttti{c})$ and $p(\texttti{c},\texttti{d})$, where  $B \in N_C$, $p \in N_P$, and
$\texttti{c},\texttti{d} \in N_N$. The set of nodes appearing in $G$ is denoted with
$V(G)$. 

 \myparagraph{SHACL Validation.} We assume a countably infinite set $N_S$
of \emph{shape names}, disjoint from $N_N\cup N_C \cup N_P$. A \emph{shape atom} is an expression of the form $\mathsf{s}(\texttti{a})$, where $\mathsf{s}\in N_S$ and $\texttti{a} \in N_N$. A \emph{path expression} $E$ is a
regular expression built using the usual operators $*$, $\cdot$,
$\cup$, property names $p \in N_P$ and \emph{inverse properties} $p^-$, where $p \in N_P$. A
\emph{(complex) shape} is an expression $\phi$ obeying~the~syntax:
\begin{align*}
 \phi,\phi'::= \top \mid \mathsf{s}\mid B \mid \texttti{c} \mid \phi\land \phi'  \mid \neg \phi \mid  \geq_n E.\phi \mid E = E',
\end{align*}
where $\mathsf{s} \in N_S$, $p \in N_P$, $B \in N_C$, $\texttti{c} \in N_N$, $n$ is a positive integer, and $E$, $E'$ are path expressions. 
In what follows, we write
$\phi\lor \phi'$ instead of $\neg (\neg \phi\land \neg \phi')$;
$\geq_n E$ instead of $\geq_n E.\top$;
$\exists E.\phi$ instead of $\geq_1 E.\phi $; 
$\forall E.\phi$ instead of $\neg \exists E.\neg
\phi$;  ${=_n}E.\phi$ instead of $\leq_n E.\phi \land \geq_n E.\phi$.  

A \emph{(shape) constraint} is an expression
$\mathsf{s} \leftrightarrow \phi$ where $\mathsf{s} \in N_S$ and $\phi$ is a complex
shape. 
W.l.o.g., we view \emph{targets} as shape atoms of the form $\mathsf{s}(\texttti{a})$, where $\mathsf{s} \in N_S$ and $\texttti{a} \in N_N$, which asks to check whether the shape name~$\mathsf{s}$ is validated at node~$\texttti{a}$ of the input data graph. The SHACL specification allows for a richer specification of targets 
but these do not affect the results in this paper. 
 A \emph{shapes graph} is a pair $( \mathcal{C,T})$, where $\mathcal{C}$ is a set of constraints and $\mathcal{T}$ is a set of targets. We assume that each shape name appearing in $\mathcal{C}$ occurs exactly once on the left-hand side of a constraint. A set of constraints $\mathcal{C}$ is \emph{recursive}, if there is a shape name in $\mathcal{C}$ that directly or indirectly refers to itself. 

The evaluation of shape expressions is given by assigning nodes of the
data graph to (possibly multiple) shape names. More formally, a
\emph{(shape) assignment} for a data graph $G$ is a set $I = G \cup L$, where $L$ is a set of shape atoms such that $\texttti{a} \in V(G)$ for each $\mathsf{s}(\texttti{a}) \in L$. The evaluation of a complex shape w.r.t.\,an assignment $I$ is given in terms of a function $\llbracket \cdot \rrbracket^I$ that maps a shape expression $\phi$ to a set of nodes, and a
path expression $E$ to a set of pairs of nodes~(see Table~\ref{tab:evaluation}). 

\begin{table}
  \renewcommand{\arraystretch}{1}
  \begin{tabular}{l}
    $\eval{\top}{I}=  V(I)$ \hspace{.3cm} $\eval{\texttti{c}}{I} =  \{\texttti{c}\}$  \hspace{.3cm} $\eval{B}{I} =  \{\texttti{c} \mid B(\texttti{c}) \in I\}$ \\
      $\eval{ \mathsf{s} }{I} =  \{\texttti{c}\mid \mathsf{s}(\texttti{c})\in I\}$ \quad  $\eval{ p  }{I}=   \{ (\texttti{a},\texttti{b})\mid p(\texttti{a},\texttti{b})\in I \}$ \\
     $\eval{ p^-  }{I}=  \{ (\texttti{a},\texttti{b})\mid p(\texttti{b},\texttti{a})\in I \}$   \\
    $\eval{ E\cup E'  }{I}=   \eval{ E  }{I} \cup  \eval{  E'  }{I}$ \quad
    $\eval{ E\cdot E'  }{I}=   \eval{ E  }{I} \circ  \eval{  E'  }{I}$ \\
    $\eval{E^{*}  }{I}=  \{(\texttti{a},\texttti{a}) \mid \texttti{a} \in V(I)\} \cup \eval{ E}{I} \cup \eval{E\cdot E}{I} \cup \cdots$ \\
    $ \eval{ \neg \phi}{I}=  V(I)\setminus \eval{  \phi}{I}$   \quad \quad  $\eval{ \phi_1\land \phi_2 }{I} =  \eval{ \phi_1 }{I}\cap \eval{ \phi_2 }{I}$ \\
    $\eval{ \atleast{n}{E}{\phi} }{I}=  \{ \texttti{c} \mid |\{ (\texttti{c}, \texttti{d})  \in \eval{E}{I} \text{ and } \texttti{d} \in \eval{ \phi }{I} \}| \geq n \}$  \\
    $\eval{ E = E' }{I}=  \{ \texttti{c} \mid \forall \texttti{d}:(\texttti{c},\texttti{d}) \in \eval{E}{I} \mbox{ iff } (\texttti{c}, \texttti{d}) \in \eval{E'}{I} \}$   
  \end{tabular}
  \caption{Evaluation of complex shapes}
  \label{tab:evaluation}
\end{table}

There are several validation semantics for SHACL with recursion
\cite{DBLP:conf/semweb/CormanRS18,DBLP:conf/www/AndreselCORSS20,DBLP:conf/datalog/ChmurovicS22}, which coincide on non-recursive SHACL. Here, we follow~\cite{shqi-etal-2021kr} and consider the supported model semantics from~\cite{DBLP:conf/semweb/CormanRS18}. Assume a SHACL shapes graph $( \mathcal{C,T} )$ and a data graph $G$ such that each node that appears in $\mathcal{C}$ or $\mathcal{T}$ also appears in $G$.  Then, an assignment $I$ for $G$ is a \emph{(supported) model}
of $\mathcal{C}$ if $\llbracket \phi \rrbracket^I = \mathsf{s}^ I$ for all
$\mathsf{s} \leftrightarrow \phi \in \mathcal{C}$. The data graph $G$ \emph{validates}
$( \mathcal{C,T} )$ if there exists an assignment $I = G \cup L$
for $G$ such that (i) $I$ is a model of $\mathcal{C}$, and (ii) $\mathcal{T} \subseteq L$.

\begin{example}\label{ex:exp1} 
Consider  $G$ and the  shapes graph $(\mathcal{C},\mathcal{T})$:
{\small
\begin{align*}
  G = & \{ \mathit{Prof}(\texttti{Ann}), \mathit{worksWith}(\texttti{Lea}, \texttti{Ann}),\mathit{Student}(\texttti{Ben}),  \\
   & \text{      } 
    \mathit{id}(\texttti{Ben}, \texttti{ID1}), \mathit{id}(\texttti{Ben}, \texttti{ID2}), \mathit{enrolledIn}(\texttti{Ben}, \texttti{c}),
   \\
   & \text{      }  \mathit{id}(\texttti{John}, \texttti{ID3}), \mathit{Student}(\texttti{John}) 
    \} \\
  \mathcal{C} =  &\{ \mathsf{Profshape}   \leftrightarrow 
  \mathit{Prof} \lor
  \exists \mathit{worksWith}.\mathsf{Profshape},  \\
  &~~ \mathsf{Studshape} \leftrightarrow \mathit{Student} \land {=_1} \mathit{id}  \land \exists \mathit{enrolledIn} 
  \} \\[1ex] 
  \mathcal{T} = & \{\mathsf{Studshape}(\texttti{Ben}),  
  \mathsf{Studshape}(\texttti{John}) 
  \} 
\end{align*}
}

\noindent
The first constraint is recursive and intuitively, it states that nodes validating the shape name $\mathsf{Profshape}$ must either belong to the class $\mathit{Prof}$ or have a $worksWith$ connection with some node validating $\mathsf{Profshape}$.
The second constraint states that nodes validating $\mathsf{Studshape}$ must belong to the class $\mathit{Students}$, have exactly one $\mathit{id}$, and belong to some $\mathit{enrolledIn}$ fact. The targets ask to check whether $\texttti{Ben}$ and $\texttti{John}$ satisfy the constraint for $\mathsf{Studshape}$. 
The data graph $G$ does not validate the shapes graph. Intuitively, the reason is that $G$ contains more than one $\mathit{id}$ for $Ben$ and it is missing an $\mathit{enrolledIn}$ fact for $\mathit{John}$.
In other words, data graph $G$ is inconsistent w.r.t.\ shapes graph 
$(\mathcal{C},\mathcal{T})$. Conversely, $G$ validates $(\mathcal{C}, \mathcal{T'})$ with $\mathcal{T'} = \{\mathsf{Profshape}(\texttti{Lea})\}$ as witnessed by the assignment $I = G \cup \{\mathsf{Profshape}(\texttti{Lea}), \mathsf{Profshape}(\texttti{Ann})\}$.
\end{example}

\myparagraph{Well-Designed SPARQL.}
Let $N_V$ be an infinite set of variables, disjoint from $N_N\cup N_C \cup N_P \cup N_S$.
A \textit{basic graph pattern (\bgp)\/} is a conjunction of atoms $\psi_1 \land \cdots \land \psi_n$, where $n\geq 0$ and each $\psi_i$ is of the form $B(t)$ or $p(t_1, t_2)$ with $B \in N_C$, $p \in N_P$, and $t, t_1, t_2 \in N_V \cup N_N$. 
We denote the empty conjunction with $\top$.
Intuitively, a \bgp is a conjunctive query built from atoms over class and property names over variables and constants, where all the variables are output variables. We focus on SPARQL queries built from \bgps and the \emph{OPTIONAL} (or OPT) operator. 
We thus may assume the so-called ``OPT-normal form'' 
(which disallows OPT-operators in the scope of a $\wedge$-operator) 
introduced by~\cite{DBLP:journals/tods/PerezAG09}. 

A \emph{(SPARQL) mapping} is any
partial function $\mu$ from $N_V$ to $N_N$. Given a unary or binary atom $\psi(\Vec{t})$ and a mapping $\mu$, we use $\mu(\psi(\Vec{t}))$ to denote the ground atom obtained from $\psi(\Vec{t})$ by replacing every variable $x$ in $\Vec{t}$ by $\mu(x)$.  We write
$\idom(\mu)$ to denote the domain of $\mu$ and
$\ivar(Q)$ for the set of variables in 
$Q$. 
Mappings $\mu_1$ and $\mu_2$ are \emph{compatible} (written $\mu_1\sim \mu_2$)
 if $\mu_1(x) = \mu_2(x)$ for all $x \in \idom(\mu_1)\cap \idom(\mu_2)$. 

The evaluation of a SPARQL query $Q$ over an RDF graph $G$ is defined as follows: 
\begin{enumerate}
   \item $\semG{Q} = \{\mu \mid \idom(\mu) = \ivar(Q), $ and $ \mu(\psi_i) \in G $ for $ i=1,\dots,n \}$, where $Q$ is a BGP $\psi_1 \land \cdots \land \psi_n$.
    \item $\lsem Q_1 \OPT\ Q_2 \rsem_G = \{\mu_1 \cup \mu_2 \mid \mu_1 \in \lsem Q_1\rsem_G, \mu_2 \in \lsem Q_2\rsem_G, $ and $ \mu_1 \sim \mu_2\} \cup \{\mu_1 \in \lsem Q_1 \rsem_G \mid \forall\mu_2 \in \lsem Q_2 \rsem_G: \mu_1 \nsim \mu_2\}$ 
    \end{enumerate}
As in ~\cite{DBLP:journals/tods/PerezAG09}, we assume set semantics.
A SPARQL query $Q$ is \emph{well-designed} (\wdpts for short), if there is no subquery $Q' = (P_1 \,\text{OPT}\ P_2)$ of $Q$ and a variable $x$, such that $x$ occurs in $P_2$, outside of $Q'$, but not in $P_1$. 
It was shown in~\cite{DBLP:journals/tods/PerezAG09} that 
the complexity of query evaluation with unrestricted OPT is $\mathsf{PSPACE}$-complete, while for \wdpts is $\mathsf{coNP}$-complete. 

Projection in SPARQL is realized via the SELECT result modifier on top of queries. For a mapping $\mu$ and a set $X$ of variables, we let $\mu|_X$ denote the mapping $\mu'$
that restricts $\mu$ to the variables in $X$, that is $\idom(\mu') = X \cap \idom(\mu)$ and $\mu'(x) = \mu(x)$ for all $x \in \idom(\mu')$. The result of evaluating a query $Q$ with projection to the variables in $X$ over a graph $G$ is defined as $\lsem\pi_X Q \rsem_G = \{\mu|_X \mid \mu \in \lsem Q\rsem_G \}$. 
We refer to queries $\pi_X Q$ as a ``projected \wdpt'' 
(or \wdptProj), and as a ``projected \bgp'' (or \bgpProj) if $Q$ is just a \bgp.
Note that there is no gain in expressiveness when we allow projections to also appear inside $Q$.
It was shown in~\cite{DBLP:journals/tods/Letelier0PS13} that checking if some mapping $\mu$ is an answer to a \wdptProj $\pi_X Q$ over a graph $G$ is $\sigmatwo$-complete. 
By inspecting the $\sigmatwo$-membership proof, it turns out that, w.l.o.g., we may assume for an arbitrary \wdptProj 
$\pi_X Q$, that $Q$ is of the form $(( \dots ((P\OPT P_1)\OPT P_2) \dots) \OPT P_k) $, such that $\ivar(P) \cap X = \idom(\mu)$ and, for every $i$, $\ivar(P_i) \subseteq X$ and $P_i$ contains at least one variable from $X \setminus \ivar(P)$. 
Then  $\mu$ is an answer to $\pi_X Q$, iff \emph{there exists} an extension $\nu$ of $\mu$ to the variables in $Y = \ivar(P) \setminus X$ s.t.\ \emph{there does not exist} an extension of $\nu$ to an answer of one of the queries $P_i$.

\section{Querying Non-valid Data Graphs}
\label{sect:queryingNonValidData}
In this section, we recall the notion of repairs for a data graph in the presence of a SHACL shapes graph proposed in~\cite{shqi-etal-2021kr} and then introduce the three inconsistency-tolerant semantics in this setting.

\tableOfResults

\myparagraph{Repairing Non-Validation.} \label{sec:expl}
We can explain non-validation of a SHACL shapes graph in the style of database repairs. Hence, a repair is provided as a set $A$ of facts to be added and a set $D$ of 
facts to be deleted,  so that the resulting data graph validates the shapes graph. We recall 
the definition below, but instead 
of ``{\em explanations}'' 
we speak of ``\emph{repairs}''.  

\begin{definition}[\cite{shqi-etal-2021kr}]\label{def:expl} Let $G$ be a data graph, let $( \mathcal{C,T} )$  be a SHACL shapes graph, 
and let the set of \emph{hypotheses} $H$ be a data graph disjoint from $G$. 
A \emph{repair} for $(G, \mathcal{C,T},H)$ is a pair $(A,D)$, such that $D \subseteq G$, $A \subseteq H$, and $(G \setminus D) \cup A $ validates~$( \mathcal{C,T} )$; we call $(G \setminus D) \cup A $ the \emph{repaired} (data) graph $G_R$ of $G$ w.r.t. $R = (A,D)$. 
\end{definition}

As usual in databases, instead of considering all possible repairs, 
we consider \emph{preference relations} given by a pre-order $\preceq$ (a reflexive and transitive relation)
on the set of repairs. Following ~\cite{shqi-etal-2021kr}, we study
\emph{subset-minimal} $(\subseteq)$, and \emph{cardinality-minimal} $(\leq)$ repairs. For two repairs $(A, D),(A', D')$, we write $(A, D) \subseteq (A', D')$ if $A \subseteq A'$ and $D\subseteq D'$, and $(A, D) \leq (A', D')$ if  $|A| + |D| \leq |A'| + |D'|$.   A preferred  repair under the pre-order $\preceq$, called $\preceq$-repair, is a repair $R$ such that there is no repair $R'$ for $\Psi$ with $R' \preceq R$ and $R \not\preceq R'$.
In that case, we also call $G_R$ a $\preceq$-repaired graph.
Clearly, every $\leq$-repair is also a $\subseteq$-repair,
but not vice versa. We denote with $=$ when there is no preference order, 
and we use $\preceq$ as a placeholder for $\subseteq$, $\leq$, and $=$.
We illustrate the notion of repairs by revisiting Example~\ref{ex:exp1}.

\begin{example}\label{ex:exp2}  
Consider $G$ and $(\mathcal{C},\mathcal{T})$ from Example~\ref{ex:exp1} and 
$H$ defined as follows:
\begin{align*}
  H = & \{ \mathit{enrolledIn}(\texttti{John},\texttti{c1}), \mathit{enrolledIn}(\texttti{Ben},\texttti{c2})\}
\end{align*}
Recall from Example~\ref{ex:exp1} that $G$ does not validate $(\mathcal{C,T})$.
Validation can be obtained by repairing $G$ with the subset- and cardinality-minimal repairs $R_1 = (A,D_1)$ and $ R_2 = (A,D_2)$,
where $A = \{ \mathit{enrolledIn}(\texttti{John},\texttti{c1})\}$, and each $D_j$ includes the fact $\mathit{id}(\texttti{Ben}, \texttti{ID}j)$. 
There are more repairs, e.g., $R_3=(A',D_1)$ and $R_4 = (A',D_2)$ with $A'=H$, but they are neither $\subseteq$-minimal nor $\leq$-minimal.
\end{example}

Constructs such as existential restrictions in constraints may sometimes enforce repairs to add atoms over fresh nodes. This is supported by the set of hypothesis $H$. Observe that $H$ can only introduce a limited number of fresh nodes. Indeed, it was shown in~\cite{shqi-etal-2021kr} that if $H$ is left unrestricted, most problems related to repairs, 
such as checking the existence of a repair, become undecidable. 

\myparagraph{Query Answering Semantics under Repairs.} 
We now define the three inconsistency-tolerant semantics {\em brave}, 
{\em AR}, and {\em IAR semantics}, which we will refer to by the symbols $\exists$, $\forall$, and $\cap$, respectively.
Consider a query $Q$, a mapping $\mu$, a data graph $G$, a shapes graph $(\mathcal{C,T})$, and hypotheses $H$.
Then, $\mu$ is an answer of $Q$ over $\Psi=(G,\mathcal{C,T},H)$
and preference order $\preceq$: 

\begin{itemize}
    \item under  \emph{brave semantics}, if there \emph{exists} a $\preceq$-repair $R$ for $\Psi$, such that $\mu \in \lsem Q \rsem_{G_R}$,
    \item under \emph{AR semantics}, if $\mu \in \lsem Q \rsem_{G_R}$ \emph{for all} $\preceq$-repairs $R$ for $\Psi$,
    \item under \emph{IAR semantics}, if $\mu \in \lsem Q \rsem_{G_\cap}$, where $G_\cap= \bigcap\{G_R \mid R $ is a $ \preceq$-repair for $ \Psi \}$, i.e., the \emph{intersection of all} $\preceq$-repaired graphs $G_R$.
\end{itemize} 

We illustrate the semantics by continuing Example~\ref{ex:exp2}.
\begin{example}\label{ex:exp3} Consider again $\Psi = (G,\mathcal{C,T},H)$ from Example~\ref{ex:exp2} together with the \bgp $Q = \mathit{Student}(x) \land id(x,y)$. Clearly, the mapping $\mu_1 = \{x \rightarrow \texttti{John}, y \rightarrow \texttti{ID3}\}$ is an answer of $Q$ over $\Psi$ under brave, AR, and IAR semantics. The mappings $\mu_2 = \{x \rightarrow \texttti{Ben}, y \rightarrow \texttti{ID1}\}$ and 
$\mu_3 = \{x \rightarrow \texttti{Ben}, y \rightarrow \texttti{ID2}\}$ are answers of $Q$ over $G_{R_1}$ and $G_{R_2}$, respectively, and hence under brave semantics, but not under AR and IAR semantics. 
Now consider the \wdpt $Q2 = \mathit{Student}(x) \,\text{ OPT }\, \text{id}(x,y)$. In this case, $\mu_1$ and $\mu_4 = \{x\rightarrow \texttti{Ben}\}$ are solutions under IAR semantics.
Clearly,  $\mu_1$, $\mu_2$, and $\mu_3$ are still answers to $Q2$ under brave semantics and $\mu_1$ under AR semantics. Note that the above statements hold for each preference order $\preceq$.
\end{example}

Observe that for \bgps (with and without projection), if $\mu$ is an answer under IAR semantics, then $\mu$ is an answer under AR and brave semantics. This is due to the monotonicity property of \bgps. For well-designed queries, this may not be the case. For instance, $\mu_4$ in Example~\ref{ex:exp3} is an answer under IAR semantics, but not under AR and brave semantics. However, $\mu_4$ 
{\em can be extended} to a solution over every repaired graph $G_{R_i}$ of $G$. In fact, it is known that well-designed queries have some weak form of monotonicity, in the sense that a solution is not lost, but it may be extended if new facts are added to the data.

\myparagraph{Decision Problems.}
For a given SPARQL query language $\mathcal{L} \in \{\bgp, \bgpProj, \wdpt, \wdptProj\}$, preference order ${\preceq}\in \{=, {\leq}, \subseteq\}$, and 
inconsistency-tolerant semantics $\mathcal{S}\in \{\exists, \forall, \cap\}$, we define the \cqa problem 
as follows:

\problemdef{\LPrecQ}{A query $Q \in \mathcal{L}$, $\Psi= (G,\mathcal{C,T},H)$, 
 and a mapping $\mu$}{Is $\mu$ an answer of $Q$ over $\Psi$ and preference order $\preceq$ under $\mathcal{S}$-semantics}
For all settings, we analyze both the \emph{data} and \emph{combined} complexity. 
We refer to Table~\ref{tab:results} for the complete picture of our main results, which will be discussed in detail in Sections \ref{sect:brave} and \ref{sect:iar}. 
In Section \ref{sect:max}, we will then study the ``max-variants'' of these problems, i.e., settings where the existence of a repair is not guaranteed and we may have to settle for validating a maximal subset of the targets.
There may be several reasons for the non-existence of a repair --  including conflicting constraints with target shape atoms, unsatisfiable constraints, or insufficient hypothesis set. 
For instance, consider constraints $\mathsf{s1} \leftrightarrow B$ and $\mathsf{s2} \leftrightarrow \neg B$ and targets $\mathsf{s1}(\texttti{a})$ and $\mathsf{s2}(\texttti{a})$; in this case adding $B(\texttti{a})$ violates the second constraint and not adding it violates the first constraint. Hence, there exists no repair for any input data graph.
Note that we never use recursive constraints in our lower bounds, and the combined complexity lower bounds hold even for fixed constraints and hypotheses. The results in Sections \ref{sect:brave} and \ref{sect:iar} hold even for fixed targets.

\section{Brave and AR Semantics}
\label{sect:brave}

We start our complexity analysis of CQA with two of the most basic cases, which yield the lowest complexity classifications in Table~\ref{tab:results}, namely the data complexity of the \BEquEx and \BEqu problems. 
The $\np$-membership of the former and the \conp-membership of the latter are immediate:
Given a graph $G$, shapes graph $(\mathcal{C,T})$, hypotheses $H$, \bgp $Q$, and mapping $\mu$, do the following: 
(1) guess a repaired graph  $G_R$ together with a supported model $I$ and 
(2) check that $\mu \in \lsem Q \rsem_{G_R}$ holds (for \BEquEx) or $\mu \not\in \lsem Q \rsem_{G_R}$ holds (for \BEqu), respectively; moreover, check that $I$ is indeed a supported model.
That is, yes-instances in the case of brave semantics and no-instances in the case of AR semantics are identified by essentially the same procedure. 
We, therefore, study the two semantics simultaneously in this section.
Note that, switching to \wdptProj as the most expressive query language considered here, does not increase the {\em data complexity}, since the check $\mu \in \lsem Q \rsem_{G_R}$ or $\mu \not\in \lsem Q \rsem_{G_R}$ is still feasible in \ptime for \wdptProjs (this is even true for arbitrary SPARQL queries).

Now consider the \LCardEx and \LCard problems for $\mathcal{L}\in \{\bgp, \bgpProj, \wdpt, \wdptProj\}$. 
To arrive at a cardinality minimal repair $R$, one first has to compute the minimal cardinality $k$ of a repair.
This can be done by asking \np-questions of the form: ``Does there exist a repair of size $\leq c$?''.
With binary search, only a logarithmic number of $\np$-oracle calls are required for this task. 
After that, we can check with another oracle call if there exists a repair $R = (A,D)$ {\em of size $|A| + |D| = k$} with $\mu \in \lsem Q \rsem_{G_R}$ or $\mu \not\in \lsem Q \rsem_{G_R}$, respectively. 
In total, we thus end up in \thetatwo.

Finally, for the \LSubsEx and \LSubs problems, we would still start by (1) guessing a repaired graph  $G_R$ plus supported model $I$. But then, in step (2), we have to additionally check that $R$ is $\subseteq$-minimal, which requires \conp-power (i.e., there {\em does not exist a $\subseteq$-smaller repair}). 
In total, this gives us a $\sigmatwo$-procedure for \LSubsEx and a $\pitwo$-procedure for \LSubs.

As can be seen in Table \ref{tab:results}, the {\em combined complexity} gives a much more varied picture. 
Compared with data complexity, we note that the complexity may possibly increase by up to 2 levels in the polynomial hierarchy. 
Moreover, the increase of complexity now also differs between brave and AR semantics. 
However, in some cases, the combined complexity remains the same as the data complexity.
In particular, this applies to all cases of \LPrecEx with  $\mathcal{L}\in \{\bgp, \bgpProj\}$ and $\preceq$\,$\in$\,$\{=, \leq, \subseteq\}$.
The reason for this is that, for $\bgps$, no additional complexity arises anyway (this also holds for \BPrec).
For $\bgpProjs$, checking $\mu \in \lsem Q \rsem_{G_R}$ requires an additional non-deterministic guess for the extension of $\mu$ to the bound variables. 
However, this additional guess does not push the complexity out of the complexity classes $\np$, $\thetatwo$, $\sigmatwo$, and, likewise, out of $\pitwo$ in the case of \PBSubset.
On the other hand, checking $\mu \in \lsem Q \rsem_{G_R}$ for a \wdpt $Q$ requires an additional \conp-check (namely, that $\mu$ cannot be extended to an answer for one of the OPT parts). 
Hence, only in cases where a \conp-check is already needed for data complexity is there no increase of complexity for combined complexity. 
Most notably, this applies to $\wdptProjSubsEx$ and $\WPrec$.

Theorem \ref{thm:ARbgp} below states that all membership results sket\-ched here are indeed tight.

\begin{theorem}
\label{thm:ARbgp}
The following statements are true for data complexity:
\begin{itemize}
    \item \LEquEx is \np-c for $\mathcal{L}\in \{\bgp, \bgpProj, \wdpt, $ $\wdptProj\}$.
    \item \LCardEx and \LCard are \thetatwo-c for $\mathcal{L}\in \{\bgp, \bgpProj, \wdpt, \wdptProj\}$.
    \item \LSubsEx is \sigmatwo-c for $\mathcal{L}\in \{\bgp, \bgpProj, \wdpt, $ $\wdptProj\}$.
    \item \LEqu is \conp-c for $\mathcal{L}\in \{\bgp, \bgpProj, \wdpt, $ $\wdptProj\}$.
    \item \LSubs is \pitwo-c for $\mathcal{L}\in \{\bgp, \bgpProj, \wdpt, $ $\wdptProj\}$.
\end{itemize}
The following statements are true for combined complexity:
\begin{itemize}
    \item \LEquEx is \np-c for $\mathcal{L}\in \{\bgp, \bgpProj\}$.
    \item \LCardEx is \thetatwo-c for $\mathcal{L}\in \{\bgp, \bgpProj\}$.
    \item \LSubsEx is \sigmatwo-c for $\mathcal{L}\in \{\bgp, \bgpProj, \wdpt, $ $\wdptProj\}$.
    \item \LEqu is \conp-c for $\mathcal{L}\in \{\bgp, \wdpt\}$.
    \item \LCard is \thetatwo-c for $\mathcal{L}\in \{\bgp, \wdpt\}$.
    \item \LSubs is \pitwo-c for $\mathcal{L}\in \{\bgp, \bgpProj, \wdpt\}$.
\end{itemize}
\end{theorem}

\begin{proof}[Proof sketch.] 
We only discuss the hardness part. 
Our reductions are from prototypical complete problems on propositional formulas for the complexity classes \np, \conp, \thetatwo, \sigmatwo, \pitwo.
We thus encode the semantics of propositional logic into the constraints $\mathcal{C}$ and the concrete propositional formula $\phi$ (in 3-CNF) into the data graph $G$.
This is done in a way such that repairs correspond to truth assignments under which $\phi$ is evaluated. 
As an example, consider the following constraints $\mathcal{C}$ used for the case \bgpEquEx:
{\small
    \begin{align*}
        \mathsf{lit} \leftrightarrow {} & Lit \land ((T \land \lnot F \land \exists dual.F) 
        \lor (F  \land \lnot T \land \exists dual.T))  \\
        \mathsf{cl} \leftrightarrow {} & Cl \land {=_3}or^- \land {=_1}and \land (\\
        & (F \land \lnot T \land \forall or^-.F) 
        \lor (T \land \lnot F \land \exists or^-.T) ) \\
        \mathsf{phi} \leftrightarrow {} & Phi \land ((\forall and^-.T \land T \land \lnot F) 
        \lor (\exists and^-.F \land F \land \lnot T))\\
        \mathsf{val} \leftrightarrow {} & \forall next^*.(\texttti{e} \lor \mathsf{lit} \lor \mathsf{cl} \lor \mathsf{phi}) 
        \land \exists next^*.\texttti{e}
    \end{align*}
}    
    The target is $\mathcal{T}=\{\mathsf{val}(\texttti{s})\}$, $H=\emptyset$, $\mu = \{x\mapsto \texttti{s}\}$, and the query is $Q=T(x)$.
    The nodes of $G$ are the literals and the clauses of $\phi$ and two auxiliary nodes $\texttti{s}, \texttti{e}$, where $\texttti{s}$ represents $\phi$ itself.
    The first three constraints of $\mathcal{C}$ are respectively meant for the literals (put into the class $Lit$), clauses (put into the class $Cl$), and $\phi$ itself (put into the class $Phi$), ensuring that all of these are either true or false in a repaired graph $G_R$.
    I.e., they are all both in the classes $T$ and $F$ in $G$ but can only remain in one of them in a repaired graph $G_R$.  These constraints ensure that truth values are correctly propagated through Boolean expressions.
    Concretely, nodes of $G$ are connected via properties $dual, and, or$ such that dual literals have dual truth values, clauses are true iff one of its literals is true, and $\phi$ is true iff all of its clauses are true.
        Lastly, the property $next$ encodes the immediate successors and predecessors of an arbitrary linear order $\preceq_{next}$ that starts in node $\texttti{s}$, goes through all other nodes of $G$, and ends in $\texttti{e}$.
    This property is used in the last constraint to enforce the first three constraints on literals, clauses, and $\phi$ itself ($\texttti{s}$ represents $\phi$) without having to explicitly name them in the shapes graph. 
    Concretely, in $G_R$ the shape name $\mathsf{val}$ has to be validated at $\texttti{s}$ and, thus, the last constraint ensures that every node $\texttti{n}$ of $G$ that appears after $\texttti{s}$ in the order $\preceq_{next}$, i.e.,  every node of $G$ has to either be $\texttti{e}$ or validate $\mathsf{lit}, \mathsf{cl}$, or $\mathsf{phi}$.
    The remaining parts of the constraints ensure, together with the target, that repairs do not alter the property names $dual, or, and, next$, and the class names $Lit,Cl,Phi$ (remove atoms over these names). 
    The satisfiability of $\phi$ is then checked by asking whether there is a repaired graph $G_R$ s.t.\ $\phi$ is in class $T$ in $G_R$.
    Thus, we prove \np-hardness.
    
    When we are only interested in $\leq$-repairs, we can have repairs be ``penalized'' for setting variables to true.
    Therefore, with a \bgp query, we can check whether a concrete variable $\texttti{x}$ is true in some \textit{minimal} model of $\phi$, establishing a reduction from the \thetatwo-complete problem \textsc{CardMinSAT} \cite{DBLP:journals/lmcs/CreignouPW18}.

    For $\subseteq$-repairs, we proceed one step higher up the polynomial hierarchy.
    We do this by splitting the variables of $\phi$ into $X$ variables and $Y$ variables.
    Now, the question is if there is a truth assignment of the $X$ variables that can \textit{not} be extended to a model of $\phi$, i.e., whether $\exists X \forall Y \lnot \phi$.
    We thus construct $(G,\mathcal{C,T},H)$ s.t., intuitively, repairs have the choice between instantiating only the $X$ variables or both the $X$ and the $Y$ variables.
    Crucially, when the instantiation is the same on the $X$ variables, a repair of the second kind is then a subset of the repair of the first kind.
    But a repair of the second kind has to represent a model of $\phi$.
    Thus, we can answer $\exists X \forall Y \lnot \phi$ by asking whether there is a $\subseteq$-repair that instantiates only the $X$ variables.
\end{proof}

Next, we consider the remaining cases for brave and AR semantics. 
We have already seen above that when considering combined complexity, two additional sources of complexity may arise: 
(1) another non-deterministic guess if the query involves projection (namely, to see if an extension of $\mu$ to the bound variables exists) and 
(2) another \conp-check in case of \wdpts (namely, to check that $\mu$ cannot be extended to one of the OPT parts). 
Hence, for \LEquEx and \LCardEx with $\mathcal{L}\in \{\wdpt, \wdptProj\}$, the additional \conp-check increases the combined complexity to $\sigmatwo$. 
For \PBPrec with $\preceq$\,$\in$\,$\{=,\leq\}$, the additional non-deterministic guess increases the complexity to $\pitwo$. 
The most dramatic increase of complexity (namely from \conp to \pithree) happens for \PWEqu where the additional guess and \conp-check introduce orthogonal new sources of complexity.
Of course, also for \PWPrec with $\preceq$\,$\in$\,$\{\leq, \subseteq\}$, the complexity rises to \pithree.

In Theorem \ref{thm:ARrest} below, we again state that the above-sketched membership results are actually tight.

\begin{theorem}
\label{thm:ARrest}
The following statements are true for combined complexity:
\begin{itemize}
    \item \LEquEx and \LCardEx are \sigmatwo-c for $\mathcal{L}\in \{\wdpt, \wdptProj\}$.
    \item \bgpProjEqu and \bgpProjCard are \pitwo-c.
    \item \wdptProjPrec is \pithree-c for ${\preceq}{}\in{} \{=,\leq,\subseteq\}$.
\end{itemize}
\end{theorem}

\begin{proof}[Proof sketch]
    For the hardness proofs, we use a \sigmak-complete $k$-round coloring game played on a graph $G$ \cite{DBLP:journals/jcss/AjtaiFS00}.
    In the game, Player 1 starts Round 1 by coloring the degree 1 vertices.
    Then Player 2 proceeds to color the degree 2 vertices.
    Then Player 1 again colors the degree 3 vertices, and so on \dots.
    In the last Round $k$, the Turn Player colors all remaining vertices and wins if the final coloring is a valid 3-coloring.
    The question is whether Player 1 has a winning strategy. 

We show the reduction for \wdptCardEx and \wdptEquEx from $2$-\textsc{Rounds-3-Colorability}. Let $G_{col}$ be an arbitrary instance of the problem with $n$ nodes. 
We construct a data graph $G$ whose nodes are the vertices of $G_{col}$ and 5 extra nodes $\texttti{r}, \texttti{g}, \texttti{b}, \texttti{s}, \texttti{e}$. For every $\texttti{v} \in V(G_{col})$, we add $col(\texttti{v},\texttti{r})$, $col(\texttti{v},\texttti{g})$, $col(\texttti{v},\texttti{b})$ to indicate that every vertex can be colored by every color. We add $neq(\texttti{c}, \texttti{c}')$ for $\texttti{c}, \texttti{c}' \in \{\texttti{r}, \texttti{g}, \texttti{b}\}, \texttti{c}\neq \texttti{c}'$ to distinguish colors; $L(\texttti{u})$ to $G$ for every leaf (degree 1 node) $\texttti{u}$ of $G_{col}$; $I(\texttti{v})$ for every non-leaf node $\texttti{v}$ of $G_{col}$. We again encode an arbitrary linear order $\preceq_{next}$ on all the nodes of $G$ that starts in $\texttti{s}$ and ends in $\texttti{e}$.
We construct $(\mathcal{C,T})$, where $\mathcal{C}$ is:
    \begin{align*}
     &   \mathsf{leaf}  \leftrightarrow {} L \land  {=_1} col \quad 
        \mathsf{inner}  \leftrightarrow{}  I \land  {=_3} col \quad \mathsf{valC}  \leftrightarrow {} {=_2} neq\\
   &     \mathsf{valV}  \leftrightarrow {} \forall next^*.(\texttti{s}\lor \texttti{e} \lor \mathsf{leaf} \lor \mathsf{inner})
         \land \exists next^*.\texttti{e}
    \end{align*}
and $\mathcal{T} = \{\mathsf{valV}(\texttti{s}), \mathsf{valC}(\texttti{r})$, $ \mathsf{valC}(\texttti{g})$, $ \mathsf{valC}(\texttti{b})\}$. Moreover, $H = \emptyset$, $\mu = \{\}$, and the query is $    Q = \top  \OPT P$, where $P = \bigwedge_{i = 1}^n col(\texttti{v}_i,x_i) \land \bigwedge_{(\texttti{v}_i,\texttti{v}_j)\in G_{col}} neq(x_i,x_j)$. Then, there exists a $\preceq$-repair $R$ for $(G,\mathcal{C,T},H)$ such that $\mu \in \lsem Q\rsem_{G_R}$ and ${\preceq} \in \{=, \leq\}$ iff there is a coloring of the leaf nodes that cannot be extended to a coloring of the whole $G_{col}$, i.e., Player 1 has a winning strategy.
Intuitively, the constraints with the targets ensure that every repaired graph $G_R$ provides a possible coloring of the leaves (constraint for $\mathsf{leaf}$), but leaves the three colors for the inner nodes (constraint for $\mathsf{inner}$). 
The query $Q$ together with $\mu$ then asks whether there is no valid coloring of the rest of the nodes. 
 
For AR semantics and \bgpProjs, we can use the query $\pi_\emptyset P$. Then, $\mu$ is an answer to $\pi_\emptyset P$ over every repaired graph $G_R$ iff there exists a coloring of the whole graph given any coloring of the leaves, i.e., Player 1 has no winning strategy. The idea is similar for AR semantics and \wdptProjs.
\end{proof}

\section{IAR Semantics}
\label{sect:iar}

We now turn our attention to CQA under IAR semantics. 
For \bgp, the AR- and IAR-semantics coincide as an atom $\alpha$ appears in every $\preceq$-repaired graph iff it appears in the intersection of all $\preceq$-repaired graphs.
Thus, we start by looking at the case of \bgpProj in data complexity. 
The natural idea seems to consist in modifying the basic guess-and-check algorithm from Section \ref{sect:brave} by guessing the intersection $G_\cap$ of all repaired graphs in step (1) and 
extending step (2) by a check that $G_\cap$ is indeed the desired intersection.
However, this approach introduces an additional source of complexity since we apparently need  \conp-power to check that the atoms in  $G_\cap$ are indeed contained in every repaired graph.

In this section, we show that we can, in fact, do significantly better. 
The key idea is to {\em guess a superset $G'_\cap$} of the intersection $G_\cap$ of all repaired graphs. 
Given graph $G$ and hypotheses $H$, we know that every repair (and, hence, also $G_\cap$) must be a subset of $G \cup H$. 
Now the crux is to guess {\em witnesses} for atoms that are definitely {\em not} in the intersection $G_\cap$.
That is, for each atom $\alpha \in (G \cup H) \setminus G'_\cap$  (there are at most linearly many) guess a repair $R_\alpha$ with $\alpha \not\in G_{R_\alpha}$.
To sum up, for \bgpProjEquIn, we identify no-instances as follows:
(1) guess a subset $G'_\cap \subseteq (G \cup H)$ together with repaired graphs $G_{R_\alpha}$ (and their supported models $I_\alpha$) for each $\alpha \in (G\cup H)\setminus G'_\cap$ and 
(2) check that $\mu \not \in \lsem Q \rsem_{G'_\cap}$ holds; moreover, for every $\alpha$, check that $\alpha\not\in G_{R_\alpha}$ (and that $I_\alpha$ is indeed a supported model for repair $R_\alpha$).
By the monotonicity of \bgpProjs, it does not matter if $G'_\cap$ is a strict superset of $G_\cap$. 
This algorithm, therefore, establishes the \conp-membership of \bgpProjEquIn.

For $\subseteq$-repairs, the check in step (2) has to be extended by checking that all of the guessed repairs $R_\alpha$ are  $\subseteq$-minimal. 
This only requires \conp-checks and we remain in $\pitwo$ for \bgpProjSubsIn for data and combined complexity.

Let us now consider the case of $\leq$-repairs. 
Recall from Section \ref{sect:brave} that, with \thetatwo-power, we can compute the minimal cardinality $k$ of any repair. 
We can then, again with \thetatwo-power, compute the exact cardinality $K$ of the intersection $G_\cap$ of all $\leq$-repaired graphs. 
This can be achieved by logarithmically many oracle calls of the form: ``is the size of the intersection of all $\leq$-repaired graphs less than $c$?'' (or, equivalently, are there at least $|G \cup H|$ - $c$ atoms in $G \cup H$ not contained in some $\leq$-repaired graph). 
We can then modify the guess in step (1) to guessing $G'_\cap \subseteq (G \cup H)$ with $|G'_\cap| = K$. 
This means that we get $G_\cap = G'_\cap$. 
Hence, our guess-and-check algorithm now also works for \wdpts since we no longer rely on monotonicity of the query language. 
Moreover, for data complexity, projection does no harm and we get $\thetatwo$-membership of \LCardIn for all query languages considered here. 
For the combined complexity of \PBCardIn or \wdptCardIn, we also end up in $\thetatwo$, since checking $\mu \in \lsem Q \rsem_{G'_\cap}$ just requires yet another oracle call.
By analogous considerations, we establish \thetatwo-membership also for \PBEquIn.

Again, we can show that all membership results sket\-ched here are indeed tight: 

\begin{theorem}
\label{thm:easyIAR}
The following statements are true for data complexity:
\begin{itemize}
    \item \LEquIn is \conp-c for $\mathcal{L}\in \{\bgp, $ $\bgpProj\}$.
    \item \LCardIn is \thetatwo-c for $\mathcal{L}\in \{\bgp, $ $\bgpProj, $ $\wdpt, $ $\wdptProj\}$.
    \item \LSubsIn is \pitwo-c for $\mathcal{L}\in \{\bgp, \bgpProj\}$.
\end{itemize}
The following statements are true for combined complexity:
\begin{itemize}
    \item \bgpEquIn is \conp-c.
    \item \bgpProjEquIn is \thetatwo-c.
    \item \LCardIn is \thetatwo-c for $\mathcal{L}\in \{\bgp, \bgpProj, \wdpt\}$.
    \item \LSubsIn is \pitwo-c for $\mathcal{L}\in \{\bgp, \bgpProj\}$.
\end{itemize}
\end{theorem}

\begin{proof}[Proof sketch]

    The only remaining case is the \thetatwo-hardness of \bgpProjEquIn.
    To that end, we reduce from the problem \textsc{CardMin-Precoloring} that asks whether for a graph $G'$ and 
    {\em pre-coloring} (i.e., an assignment of admissible colors for each vertex) 
    $c:V(G') \rightarrow 2^{\{\texttti{r}, \texttti{g}, \texttti{b}\}}$, a given vertex $\texttti{v}_1\in V(G')=\{\texttti{v}_1, \dots, \texttti{v}_n\}$ is colored $\texttti{g}$ in some 3-coloring of $G'$ that minimizes the use of the color $\texttti{g}$.

    The bulk of the reduction, i.e., the definition of the data graph and the shapes graph is very technical and 
    therefore omitted here.
    But the intuition behind it is quite simple.
    The idea is to split the ``computation'' into two steps.
    The first step is to compute the minimum number of vertices that need to be colored $\texttti{g}$ and this is handled by the (intersection of the) repaired graphs.
    The second step is to try to compute such a minimal coloring while assigning color $\texttti{g}$ to vertex $\texttti{v}_1$.
    This is handled by the (boolean) query $\pi_\emptyset Q$
    where $Q$ is
      {\small
    \begin{align*}
      & \bigwedge_{(\texttti{v}_i,\texttti{v}_j)\in G'} neqCol(y_i,y_j)\land  Pre_{\{\texttti{g}\}}(y_1) \land \bigwedge_{i} Pre_{c(\texttti{v}_i)}(y_i) \\
       &   \land \bigwedge_{i} allowed(x_i,y_i) \land  \bigwedge_{i\neq j} neqCnt (x_i,x_j)  .
    \end{align*} 
    }
    In this query, the variables $y_i$ represent colors of the vertices of the graph $G'$ (concretely, $y_i$ corresponds to $\texttti{v}_i$) with the classes $Pre_C$ ensuring that the coloring adheres to the given pre-coloring and, additionally, $\texttti{v}_1$ is colored $\texttti{g}$.
    The property $neqCol$ in the query then ensures that adjacent vertices are assigned different colors.
    Thus, in total, the first line of the query ensures that we are dealing with a valid 3-coloring of $G'$ and $\texttti{v}_1$ is colored $\texttti{g}$.
    This part is more or less unimpacted by repairs.
    However, the repairs impact the $allowed$ property in an important way, thus restricting the number of vertices that are ``allowed'' to be colored $\texttti{g}$ (the colors $\texttti{r}, \texttti{b}$ are always ``allowed'').
    Concretely, there are $n$ possible values for the variables $x_i$ and all have to be different due to $neqCnt$.
    These values are the counters $\texttti{i}_1,\dots,\texttti{i}_n$ and $allowed(\texttti{i}_j, \texttti{g})$ is in the intersection of all repaired graphs iff $\texttti{g}$ has to be used $\geq j$-times.
    Let $k$ be such that at least $k$ vertices have to be colored $\texttti{g}$.
    Then, to satisfy the last line of the query, the vertices colored $\texttti{g}$ by the variables $y_i$ have to be matched to $\texttti{i}_1,\dots,\texttti{i}_k$, ensuring that there are at most $k$ vertices colored $\texttti{g}$.
    Thus, looking at the whole query, a valid instantiation of the $X$ and $Y$ variables constitutes a minimal 3-coloring that assigns color $\texttti{g}$ to vertex $\texttti{v}_1$.
\end{proof}

Next, we take a look at the remaining cases for \wdpts.
Recall from Section~\ref{sect:preliminaries} that we may assume that a \wdpt $Q$ is of the form $(( \dots ((P\OPT P_1)\OPT P_2) \dots)\OPT P_k)$ with $\ivar(P) = \idom(\mu)$ (we are considering \wdpts without projection here) and each of the subqueries $P_i$ contains at least one free variable outside $\ivar(P)$.

Then $\mu$ is an answer to $Q$ if and only if $\mu$ is an answer to $P$ but, for all $i$, $\mu$ \textit{cannot} be extended to an answer of $P_i$.
Checking if $\mu$ is an answer to $P$ corresponds to identifying yes-instances of \bgpEquIn; checking that $\mu$ cannot be extended to an answer of any $P_i$ corresponds to identifying no-instances of \PBEquIn (note that the variables in $\ivar(P_i) \setminus \ivar(P)$ behave like bound variables in this case, since we are not interested in a particular extension of $\mu$ to these variables but 
in {\em any} extension).
The problem \wdptPrecIn can, therefore, be seen as the intersection of \bgpPrecIn and (multiple) \bgpProjPrecIn.
This proves the \DP- and \thetatwo-membership in case of $=$, respectively for data and combined complexity, and the \DPtwo-membership in case of $\subseteq$ for both settings.
We thus get the following complexity classification for \wdpts.

\begin{theorem}
\label{thm:wdptIAR}
The following statements are true for data complexity:
\begin{itemize}
    \item \wdptEquIn is \DP-c.
    \item \wdptSubsIn is \DPtwo-c.
\end{itemize}
The following statements are true for combined complexity:
\begin{itemize}
    \item \wdptEquIn is \thetatwo-c.
    \item \wdptSubsIn is \DPtwo-c.
\end{itemize}
\end{theorem}

\begin{proof}[Proof sketch]    
    The hardness part follows a similar idea as the membership: We reduce from instance pairs $(\mathcal{I},\mathcal{J})$, where $\mathcal{I}$ comes from \bgpPrecIn and $\mathcal{J}$ comes from \bgpProjPrecIn.
    For the reduction, we simply merge (take the union of) the data graphs, the hypotheses, the constraints, and the targets.
    Then, the intersection $G_\cap$ of all $\preceq$-repaired graphs (of the new instance) is exactly the union of the intersections $G_{\mathcal{I},\cap}$ and $G_{\mathcal{J},\cap}$ (of the old instances $\mathcal{I}$ and $\mathcal{J}$).
    Furthermore, the two mappings $\mu_\mathcal{I}, \mu_\mathcal{J}$ are merged to 
    $\mu$, while the queries $Q_\mathcal{I} = P_1$ and $Q_\mathcal{J} = \pi_Y  P_2$ are combined into $Q=(P_1 \land \top(Y)) \OPT P_2$. 
    Here, 
    $\top(Y)$ is a short-hand
    for a conjunction that allows instantiating the variables in $Y$ to any node of $\mathcal{J}$.
    Then, for $\mu$ to be an answer to $Q$ over $G_\cap$, the part $\mu_\mathcal{I}$ has to be an answer of $P_1$ over the part $G_{\mathcal{I},\cap}$ while $\mu_\mathcal{J}$ cannot be extended to an answer of $P_2$ over the part $G_{\mathcal{J},\cap}$.
\end{proof}

We now consider the last remaining cases for \wdptProj.
As we have already discussed above, \thetatwo-power suffices to compute the size $K$ of the intersection of all ($\leq$-)repaired graphs.
For $\subseteq$-minimality, \thetathree-power is needed to compute $K$ by analogous arguments and the fact that checking $\subseteq$-minimality only requires access to an \np-oracle.
For data complexity, given $K$, we can then simply ``guess'' the intersection $G_\cap$ and answer the queries in polynomial time. 
Hence, we end up in \thetatwo for \wdptProjEquIn and in \thetathree for \wdptProjSubsIn. 
For combined complexity, we have to take the \sigmatwo-completeness of evaluating \wdptProjs into account. 
This fits into the \thetathree-bound in the case of \wdptProjSubsIn but leads to \sigmatwo-completeness of \wdptProjEquIn.

\begin{theorem}
\label{prop:wdptProj}
The following statements are true for data complexity:
\begin{itemize}
    \item \wdptProjEquIn is \thetatwo-c.
    \item \wdptProjSubsIn is \thetathree-c.
\end{itemize}
The following statements are true for combined complexity:
\begin{itemize}
    \item \wdptProjEquIn, \wdptProjCardIn are \sigmatwo-c.
    \item \wdptProjSubsIn is \thetathree-c.
\end{itemize}
\end{theorem}

\begin{proof}[Proof sketch]
    The cases \wdptProjEquIn and \wdptProjCardIn immediately follow from the hardness of answering \wdptProj in combined complexity.

    For \wdptProjEquIn, as in the proof of Theorem \ref{thm:ARbgp}, we encode the semantics of propositional logic into the constraints $\mathcal{C}$ and the concrete propositional formulas $\phi$  into the data graph $G$.
    However, this time, we do not simply encode a single propositional formula into our instance but a list $L = (\phi_1, \psi_1),\dots, (\phi_n, \psi_n)$ of pairs of propositional formulas.
    $L$ is a yes-instance if there is a pair $(\phi_i,\psi_i)$ such that $\phi_i$ is unsatisfiable while $\psi_i$ is satisfiable.
    This constitutes a natural \thetatwo-complete problem.

   Intuitively, the constructed instance is such that the repairs have the (independent) choice of what truth values to assign to each variable appearing in any formula.
    The truth values of the formulas themselves are then functionally determined by this choice.
    Thus, $F(\phi_i)$ (resp. $F(\psi_i)$) is in the repaired graph if and only if $\phi_i$ (resp. $\psi_i$) evaluates to false under the given choice of the truth values for the variables in $\phi_i$ (resp. $\psi_i$).
    Since each truth value assignment constitutes a repair, $F(\phi_i)$ (resp. $F(\psi_i)$) appears in the intersection of all repaired graphs if and only if $\phi_i$ (resp. $\psi_i$) is unsatisfiable.
    Furthermore, formulas $\phi_i$ (resp. $\psi_i$) are connected to their indices via $is_\phi(\texttti{i}_i, \phi_i)$ (resp. $is_\psi(\texttti{i}_i, \psi_i)$).
    Thus, the empty mapping $\mu=\{ \, \}$ is an answer to the query $\pi_zQ$ with
    {\small
    \begin{align*}
        Q =  is_{\phi}(x,y)\land  F(y) \OPT is_{\psi}(x,z)\land  F(z)
    \end{align*}
    }
    if and only if there is an $\texttti{i}_i$ (variable $x$) such that $\phi_i$ (variable $y$) is unsatisfiable but $\psi_i$ (variable $z$) is not unsatisfiable, i.e., $\psi_i$ is satisfiable.

    For \wdptProjSubsIn, we have to combine the ideas of the previous case \wdptProjEquIn with the ideas of the case \bgpSubsEx.
    That is, we split the variables of each formula $\phi$ into $X$ variables and $Y$ variables and then interpret the list as consisting of QBFs $\forall X\exists Y \phi$.
    Then, similar to \bgpSubsEx, $\subseteq$-repairs ``know'' whether a concrete truth assignment of $X$ can be extended to a model of $\phi$ and, the intersection ``knows'' whether this is the case for all truth assignments of $X$.
    Thus, a similar query as before is enough to identify if there is a $(\phi_i,\psi_i)$ such that $\forall X_i\exists Y_i \phi_i$ and $\lnot \forall X'_i\exists Y'_i \psi_i$ hold.
\end{proof}

\section{Max-Variants of CQA}
\label{sect:max}

We now consider CQA  for settings where the existence of a repair is not guaranteed.
Following~\cite{DBLP:conf/semweb/AhmetajDPS22},
we relax the notion of repairs and aim at satisfying the maximum number of targets. 
Formally, given a graph $G$, a shapes graph $(\mathcal{C},\mathcal{T})$ and a set of hypotheses $H$, we define {\em max-repairs} for $(G,\mathcal{C},\mathcal{T}, H)$ to be pairs $(A,D)$ satisfying (1) $A \subseteq H$,  $D \subseteq G$, (2) there exists $\mathcal{T}' \subseteq \mathcal{T}$ such that $(A,D)$ is a repair 
for $(G, \mathcal{C},\mathcal{T}',H)$, and (3) for every $\mathcal{T}'' \subseteq \mathcal{T}$   with $|\mathcal{T}''| > |\mathcal{T}'|$, there exists no repair 
for $(G, \mathcal{C},\mathcal{T}'', H)$.
We thus consider the max-variant \mLPrecQ of the CQA
problem, where max-repairs play the role of repairs for given query language $\mathcal{L}$, pre-order 
$\preceq$ on the repairs, and semantics $\mathcal{S}\in \{\exists, \forall, \cap\}$.

Note that, in our hardness proofs for the \cqa problem, 
we can always choose the target problem instance consisting of a graph $G$, shapes graph $(\mathcal{C},\mathcal{T})$, 
and set of hypotheses $H$ in such a way that a repair exists. 
In this case, the \cqa and \mcqa problems coincide. 
Hence, the hardness results of each variant of the \cqa problem carry over to the  \mcqa problem. 
Therefore, whenever we manage to show that the membership result for a variant of \cqa also holds for the corresponding variant of \mcqa, we may immediately conclude the completeness for this variant of \mcqa.

The algorithms sketched in Section~\ref{sect:brave} and~\ref{sect:iar} to illustrate our membership results 
for all variants of the \cqa problem now have to be extended by ensuring maximum cardinality of the targets covered by the repaired graphs that are guessed in step (1) of our algorithms.
In case of data complexity, we thus have to combine constantly many $\np$-problems deciding the question if there exists a repair that covers a subset $\mathcal{T}' \subseteq \mathcal{T}$ of cardinality $c$. 
This requires the power of some level in the Boolean Hierarchy \bh. 
In case of combined complexity, the maximally attainable cardinality of subsets $\mathcal{T}' \subseteq \mathcal{T}$ allowing for a repair has to be determined by binary search with logarithmically many 
$\np$-oracle calls of the form ``does there exist $\mathcal{T}' \subseteq \mathcal{T}$ with $|\mathcal{T}'| \geq c$
such that $\mathcal{T}'$ allows for a repair?'', which requires $\Theta_2$-power. 

Hence, for all variants of $\cqa$ whose data complexity is at least hard for every level of $\bh$ or
whose combined complexity is at least \thetatwo-hard, respectively, 
the membership and, therefore, also the completeness carries over to $\mcqa$.
Moreover, the data (resp.\ combined) complexity of the remaining cases is never lifted beyond $\bh$
(resp.\ \thetatwo).

In Theorem \ref{thm:maxData} below, 
we state that the \bh-membership results 
for the data complexity 
are actually tight by showing \bhk-hardness for every level $k>0$ of \bh.

\begin{theorem}
\label{thm:maxData}
\mbgpEquQ is \bhk-hard data complexity for every $k>0$ and $\mathcal{S}\in \{\exists, \forall, \cap\}$.
Furthermore, \mwdptProjEquEx, \mwdptProjEqu, and \mwdptEquIn are in \bh data complexity.
\end{theorem}

\begin{proof}[Proof sketch]
    We proceed similarly to the case \wdptProjEquIn.
    That is, we encode the semantics of a list of propositional formulas $(\phi_1,\psi_1),\dots,(\phi_k,\psi_k)$ into the constraints $\mathcal{C}$ and the formulas themselves into the data graph $G$.
    However, this time, the length of the list is fixed to $k$ and thus, it becomes \bhtwok-complete to decide if there is a pair $(\phi_i,\psi_i)$ such that $\phi_i$ is unsatisfiable while $\psi_i$ is satisfiable.

    We construct an instance of \mcqa with targets $\mathcal{T}_1$ that intuitively require repaired graphs to simulate the semantics of propositional formulas (similar to the proof of Theorem~\ref{thm:ARbgp}), and targets $\mathcal{T}_2$ that require the formulas to be satisfied.
    The targets $\mathcal{T}_1$ are chosen in such a way that (1) either all are satisfied or none
    and (2) there always exist repairs that validate $\mathcal{T}_1$.
    Thus, by choosing $|\mathcal{T}_1|>|\mathcal{T}_2|$, staying true to the semantics of propositional formulas is a ``hard'' constraints.
    On the other hand, $\mathcal{T}_2$ is such that, for each of the formulas $\phi_i$ and $\psi_i$, $\mathcal{T}_2$ 
    contains a target that is validated when the corresponding formula evaluates to true.
    These constitute ``soft'' constraints, and max-repairs try to satisfy as many formulas as possible.
    Therefore, a max-repair can ``decide'' whether there exists a pair $(\phi_i,\psi_i)$ as desired. In particular,
    a formula $\phi_i$ is definitely unsatisfiable, if the corresponding target in  $\mathcal{T}_2$ is not validated by a max-repair.
\end{proof}

In Theorem \ref{thm:maxCombined} below, we state that, for combined complexity, the \thetatwo-membership results are actually tight.

\begin{theorem}
\label{thm:maxCombined}
\mbgpEquQ is \thetatwo-hard combined complexity for $\mathcal{S}\in \{\exists, \forall, \cap\}$.
Furthermore, \mbgpProjEquEx, \mwdptEqu, and \mbgpEquIn are in \thetatwo combined complexity.
\end{theorem}

\begin{proof}[Proof sketch]
    We again reduce from the \thetatwo-complete problem \textsc{CardMinSat}.
    That is, given a propositional formula $\phi$ in 3-CNF, the question is whether the variable $\texttti{x}_1$ is true in some model of $\phi$ that minimizes the number of variables set to true.
        
    We construct an instance of \mcqa with multiple types of targets $\mathcal{T}_1, \mathcal{T}_2, \mathcal{T}_3, \mathcal{T}_4$.
    The targets $\mathcal{T}_1$ are essentially the same as before, ensuring that max-repaired graphs simulate the semantics of propositional logic.
    $\mathcal{T}_2$ asks for $\phi$ to be satisfied.
    $\mathcal{T}_3$ is such that, for every variable $\texttti{x}_i$, $\mathcal{T}_3$ contains two targets $t(\texttti{x}_i), t'(\texttti{x}_i)$ which are validated when $\texttti{x}_i$ is set to false.
    Lastly, $ \mathcal{T}_4$ consists of a single target that asks for $ \texttti{x}_1$ to be set to true.
    By setting $|\mathcal{T}_2| > |\mathcal{T}_3| + |\mathcal{T}_4|$ we ensure that, if $\phi$ is satisfiable, then all max-repairs will represent models of $\phi$.   
    Thus, max-repairs set both $\phi$ and $\texttti{x}_1$ to true if and only if there is a minimal model of $\phi$ in which $\texttti{x}_1$ is true.
\end{proof}

As we have already discussed at the beginning of this section, the complexity of all other cases remains the same.

\begin{theorem}
For all cases \mLPrecQ not mentioned in Theorem \ref{thm:maxData} or \ref{thm:maxCombined}, \mLPrecQ is $\mathsf{C}$-c if and only if \LPrecQ is $\mathsf{C}$-c.
This holds for data and combined complexity.
\end{theorem}

\section{Conclusion and Future Work}
\label{sect:conclusion}

In this work, we have carried out a thorough complexity analysis of the CQA problem for data graphs with SHACL constraints and we have pinpointed the complexity of this problem in a multitude of settings -- considering various query languages, inconsistency-tolerant semantics of CQA, and preference relations on repairs. 
Several new proof techniques had to be developed to obtain these results. For instance, this has allowed us -- in contrast to \cite{shqi-etal-2021kr} -- to prove all our hardness results without making use of recursion in shapes constraints. 
Hence, by carrying over our techniques to the problems studied in \cite{shqi-etal-2021kr},
we could extend the hardness to the non-recursive cases left open there.  

The targets we considered here are shape atoms of the form $\mathsf{s}(\texttti{c})$. The SHACL standard also allows for richer targets over class and property names. Specifically, it allows to state that a data graph must validate a shape name at each node of a certain class name, or domain (or range) of a property name. All the membership results in this paper can be immediately updated to support these richer targets. 
Some features of SHACL (such as disjointness and closed constraints) are not considered here,
but we strongly believe they do not change the complexity results.  

We considered the supported model semantics. Other validation semantics, which are inspired from the stable model semantics~\cite{DBLP:conf/www/AndreselCORSS20} or the well-founded semantics~\cite{DBLP:conf/datalog/ChmurovicS22} for logic programs with negation are left for future work. Again, we do not expect changes 
of the complexity: all our hardness proofs should hold since the three semantics coincide for non-recursive SHACL, and the membership  results can be immediately carried over  since they rely on the problem of validation, which is not harder than for the supported model semantics.

An immediate direction for future work is to investigate the notion of optimal repairs of prioritized data graphs over SHACL constraints and specifically, the notions of global, Pareto and completion optimality of repairs~\cite{DBLP:journals/amai/StaworkoCM12}. In particular, it would be of interest to study the computational properties of the main reasoning tasks such as repair checking and repair existence, and to analyze CQA for each of the three notions of optimal repairs and the various settings studied in this paper. Devising practical algorithms for CQA over SHACL constraints and in particular, identifying meaningful and relevant fragments of SHACL that admit better complexity results is also an important next step. 

\section*{Acknowledgements}

This work was supported by 
the Vienna Science and Technology Fund (WWTF) [10.47379/ICT2201,10.47379/VRG18013]. In addition, Ahmetaj was supported by the FWF and netidee SCIENCE project T1349-N.

\bibliographystyle{kr}
\bibliography{main}

\begin{thebibliography}{}

\bibitem[\protect\citeauthoryear{Afrati and Kolaitis}{2009}]{DBLP:conf/icdt/AfratiK09}
Afrati, F.~N., and Kolaitis, P.~G.
\newblock 2009.
\newblock Repair checking in inconsistent databases: algorithms and complexity.
\newblock In Fagin, R., ed., {\em Database Theory - {ICDT} 2009, 12th International Conference, St. Petersburg, Russia, March 23-25, 2009, Proceedings}, volume 361 of {\em {ACM} International Conference Proceeding Series},  31--41.
\newblock {ACM}.

\bibitem[\protect\citeauthoryear{Ahmetaj \bgroup et al\mbox.\egroup }{2021}]{shqi-etal-2021kr}
Ahmetaj, S.; David, R.; Ortiz, M.; Polleres, A.; Shehu, B.; and Simkus, M.
\newblock 2021.
\newblock Reasoning about explanations for non-validation in {SHACL}.
\newblock In Bienvenu, M.; Lakemeyer, G.; and Erdem, E., eds., {\em Proceedings of the 18th International Conference on Principles of Knowledge Representation and Reasoning, {KR} 2021, Online event, November 3-12, 2021},  12--21.

\bibitem[\protect\citeauthoryear{Ahmetaj \bgroup et al\mbox.\egroup }{2022}]{DBLP:conf/semweb/AhmetajDPS22}
Ahmetaj, S.; David, R.; Polleres, A.; and Simkus, M.
\newblock 2022.
\newblock Repairing {SHACL} constraint violations using answer set programming.
\newblock In Sattler, U.; Hogan, A.; Keet, C.~M.; Presutti, V.; Almeida, J. P.~A.; Takeda, H.; Monnin, P.; Pirr{\`{o}}, G.; and d'Amato, C., eds., {\em The Semantic Web - {ISWC} 2022 - 21st International Semantic Web Conference, Virtual Event, October 23-27, 2022, Proceedings}, volume 13489 of {\em Lecture Notes in Computer Science},  375--391.
\newblock Springer.

\bibitem[\protect\citeauthoryear{Ahmetaj \bgroup et al\mbox.\egroup }{2023}]{DBLP:conf/ecai/Ahmetaj0S23}
Ahmetaj, S.; Ortiz, M.; Oudshoorn, A.; and Simkus, M.
\newblock 2023.
\newblock Reconciling {SHACL} and ontologies: Semantics and validation via rewriting.
\newblock In {\em {ECAI} 2023 - 26th European Conference on Artificial Intelligence}, volume 372 of {\em Frontiers in Artificial Intelligence and Applications},  27--35.
\newblock {IOS} Press.

\bibitem[\protect\citeauthoryear{Ajtai, Fagin, and Stockmeyer}{2000}]{DBLP:journals/jcss/AjtaiFS00}
Ajtai, M.; Fagin, R.; and Stockmeyer, L.~J.
\newblock 2000.
\newblock The closure of monadic {NP}.
\newblock {\em J. Comput. Syst. Sci.} 60(3):660--716.

\bibitem[\protect\citeauthoryear{Andresel \bgroup et al\mbox.\egroup }{2020}]{DBLP:conf/www/AndreselCORSS20}
Andresel, M.; Corman, J.; Ortiz, M.; Reutter, J.~L.; Savkovic, O.; and Simkus, M.
\newblock 2020.
\newblock Stable model semantics for recursive {SHACL}.
\newblock In Huang, Y.; King, I.; Liu, T.; and van Steen, M., eds., {\em {WWW} '20: The Web Conference 2020, Taipei, Taiwan, April 20-24, 2020},  1570--1580.
\newblock {ACM} / {IW3C2}.

\bibitem[\protect\citeauthoryear{Arenas, Bertossi, and Chomicki}{1999}]{DBLP:conf/pods/ArenasBC99}
Arenas, M.; Bertossi, L.~E.; and Chomicki, J.
\newblock 1999.
\newblock Consistent query answers in inconsistent databases.
\newblock In Vianu, V., and Papadimitriou, C.~H., eds., {\em Proceedings of the Eighteenth {ACM} {SIGACT-SIGMOD-SIGART} Symposium on Principles of Database Systems, May 31 - June 2, 1999, Philadelphia, Pennsylvania, {USA}},  68--79.
\newblock {ACM} Press.

\bibitem[\protect\citeauthoryear{Arming, Pichler, and Sallinger}{2016}]{DBLP:conf/icdt/ArmingPS16}
Arming, S.; Pichler, R.; and Sallinger, E.
\newblock 2016.
\newblock Complexity of repair checking and consistent query answering.
\newblock In Martens, W., and Zeume, T., eds., {\em 19th International Conference on Database Theory, {ICDT} 2016, Bordeaux, France, March 15-18, 2016}, volume~48 of {\em LIPIcs},  21:1--21:18.
\newblock Schloss Dagstuhl - Leibniz-Zentrum f{\"{u}}r Informatik.

\bibitem[\protect\citeauthoryear{Bertossi}{2011}]{DBLP:series/synthesis/2011Bertossi}
Bertossi, L.~E.
\newblock 2011.
\newblock {\em Database Repairing and Consistent Query Answering}.
\newblock Synthesis Lectures on Data Management. Morgan {\&} Claypool Publishers.

\bibitem[\protect\citeauthoryear{Bienvenu and Bourgaux}{2016}]{DBLP:conf/rweb/BienvenuB16}
Bienvenu, M., and Bourgaux, C.
\newblock 2016.
\newblock Inconsistency-tolerant querying of description logic knowledge bases.
\newblock In Pan, J.~Z.; Calvanese, D.; Eiter, T.; Horrocks, I.; Kifer, M.; Lin, F.; and Zhao, Y., eds., {\em Reasoning Web, Tutorial Lectures}, volume 9885 of {\em Lecture Notes in Computer Science},  156--202.
\newblock Springer.

\bibitem[\protect\citeauthoryear{Bienvenu and Rosati}{2013}]{DBLP:conf/ijcai/BienvenuR13}
Bienvenu, M., and Rosati, R.
\newblock 2013.
\newblock Tractable approximations of consistent query answering for robust ontology-based data access.
\newblock In Rossi, F., ed., {\em {IJCAI} 2013, Proceedings of the 23rd International Joint Conference on Artificial Intelligence, Beijing, China, August 3-9, 2013},  775--781.
\newblock {IJCAI/AAAI}.

\bibitem[\protect\citeauthoryear{Bienvenu, Bourgaux, and Goasdou{\'{e}}}{2019}]{DBLP:journals/jair/BienvenuBG19}
Bienvenu, M.; Bourgaux, C.; and Goasdou{\'{e}}, F.
\newblock 2019.
\newblock Computing and explaining query answers over inconsistent dl-lite knowledge bases.
\newblock {\em J. Artif. Intell. Res.} 64:563--644.

\bibitem[\protect\citeauthoryear{Bienvenu}{2012}]{DBLP:conf/aaai/Bienvenu12}
Bienvenu, M.
\newblock 2012.
\newblock On the complexity of consistent query answering in the presence of simple ontologies.
\newblock In Hoffmann, J., and Selman, B., eds., {\em Proceedings of the Twenty-Sixth {AAAI} Conference on Artificial Intelligence, July 22-26, 2012, Toronto, Ontario, Canada},  705--711.
\newblock {AAAI} Press.

\bibitem[\protect\citeauthoryear{Bogaerts, Jakubowski, and den Bussche}{2022}]{DBLP:conf/lpnmr/BogaertsJB22}
Bogaerts, B.; Jakubowski, M.; and den Bussche, J.~V.
\newblock 2022.
\newblock {SHACL:} {A} description logic in disguise.
\newblock In Gottlob, G.; Inclezan, D.; and Maratea, M., eds., {\em Logic Programming and Nonmonotonic Reasoning - 16th International Conference, {LPNMR} 2022, Genova, Italy, September 5-9, 2022, Proceedings}, volume 13416 of {\em Lecture Notes in Computer Science},  75--88.
\newblock Springer.

\bibitem[\protect\citeauthoryear{Chmurovic and Simkus}{2022}]{DBLP:conf/datalog/ChmurovicS22}
Chmurovic, A., and Simkus, M.
\newblock 2022.
\newblock Well-founded semantics for recursive {SHACL}.
\newblock In Alviano, M., and Pieris, A., eds., {\em Proceedings of (Datalog-2.0 2022)}, volume 3203 of {\em {CEUR} Workshop Proceedings},  2--13.
\newblock CEUR-WS.org.

\bibitem[\protect\citeauthoryear{Corman, Reutter, and Savkovic}{2018}]{DBLP:conf/semweb/CormanRS18}
Corman, J.; Reutter, J.~L.; and Savkovic, O.
\newblock 2018.
\newblock Semantics and validation of recursive {SHACL}.
\newblock In Vrandecic, D.; Bontcheva, K.; Su{\'{a}}rez{-}Figueroa, M.~C.; Presutti, V.; Celino, I.; Sabou, M.; Kaffee, L.; and Simperl, E., eds., {\em The Semantic Web - {ISWC} 2018 - 17th International Semantic Web Conference, Monterey, CA, USA, October 8-12, 2018, Proceedings, Part {I}}, volume 11136 of {\em Lecture Notes in Computer Science},  318--336.
\newblock Springer.

\bibitem[\protect\citeauthoryear{Creignou, Pichler, and Woltran}{2018}]{DBLP:journals/lmcs/CreignouPW18}
Creignou, N.; Pichler, R.; and Woltran, S.
\newblock 2018.
\newblock Do hard sat-related reasoning tasks become easier in the {Krom} fragment?
\newblock {\em Log. Methods Comput. Sci.} 14(4).

\bibitem[\protect\citeauthoryear{Du, Qi, and Shen}{2013}]{DBLP:journals/kais/DuQS13}
Du, J.; Qi, G.; and Shen, Y.
\newblock 2013.
\newblock Weight-based consistent query answering over inconsistent {\textdollar}{\textdollar}\{{\textbackslash}mathcal \{SHIQ\}\}{\textdollar}{\textdollar} knowledge bases.
\newblock {\em Knowl. Inf. Syst.} 34(2):335--371.

\bibitem[\protect\citeauthoryear{Franconi, Ib{\'{a}}{\~{n}}ez{-}Garc{\'{\i}}a, and Seylan}{2011}]{DBLP:journals/entcs/FranconiIS11}
Franconi, E.; Ib{\'{a}}{\~{n}}ez{-}Garc{\'{\i}}a, Y.~A.; and Seylan, I.
\newblock 2011.
\newblock Query answering with {DBoxes} is hard.
\newblock {\em Electr. Notes Theor. Comput. Sci.}

\bibitem[\protect\citeauthoryear{Khalfioui and Wijsen}{2023}]{DBLP:conf/icdt/KhalfiouiW23}
Khalfioui, A. A.~E., and Wijsen, J.
\newblock 2023.
\newblock Consistent query answering for primary keys and conjunctive queries with counting.
\newblock In Geerts, F., and Vandevoort, B., eds., {\em 26th International Conference on Database Theory, {ICDT} 2023, March 28-31, 2023, Ioannina, Greece}, volume 255 of {\em LIPIcs},  23:1--23:19.
\newblock Schloss Dagstuhl - Leibniz-Zentrum f{\"{u}}r Informatik.

\bibitem[\protect\citeauthoryear{Knublauch and Kontokostas}{2017}]{shacl}
Knublauch, H., and Kontokostas, D.
\newblock 2017.
\newblock {Shapes constraint language (SHACL)}. {W3C Recommendation, W3C.}
\newblock \url{https://www.w3.org/TR/shacl/}.

\bibitem[\protect\citeauthoryear{Koutris and Wijsen}{2021}]{DBLP:journals/mst/KoutrisW21}
Koutris, P., and Wijsen, J.
\newblock 2021.
\newblock Consistent query answering for primary keys in datalog.
\newblock {\em Theory Comput. Syst.} 65(1):122--178.

\bibitem[\protect\citeauthoryear{Koutris, Ouyang, and Wijsen}{2021}]{DBLP:conf/pods/KoutrisOW21}
Koutris, P.; Ouyang, X.; and Wijsen, J.
\newblock 2021.
\newblock Consistent query answering for primary keys on path queries.
\newblock In Libkin, L.; Pichler, R.; and Guagliardo, P., eds., {\em PODS'21: Proceedings of the 40th {ACM} {SIGMOD-SIGACT-SIGAI} Symposium on Principles of Database Systems, Virtual Event, China, June 20-25, 2021},  215--232.
\newblock {ACM}.

\bibitem[\protect\citeauthoryear{Koutris, Ouyang, and Wijsen}{2024}]{DBLP:journals/pacmmod/KoutrisOW24}
Koutris, P.; Ouyang, X.; and Wijsen, J.
\newblock 2024.
\newblock Consistent query answering for primary keys on rooted tree queries.
\newblock {\em Proc. {ACM} Manag. Data} 2(2):76.

\bibitem[\protect\citeauthoryear{Leinberger \bgroup et al\mbox.\egroup }{2020}]{DBLP:conf/semweb/LeinbergerSRLS20}
Leinberger, M.; Seifer, P.; Rienstra, T.; L{\"{a}}mmel, R.; and Staab, S.
\newblock 2020.
\newblock Deciding {SHACL} shape containment through description logics reasoning.
\newblock In Pan, J.~Z.; Tamma, V. A.~M.; d'Amato, C.; Janowicz, K.; Fu, B.; Polleres, A.; Seneviratne, O.; and Kagal, L., eds., {\em The Semantic Web - {ISWC} 2020 - 19th International Semantic Web Conference, Athens, Greece, November 2-6, 2020, Proceedings, Part {I}}, volume 12506 of {\em Lecture Notes in Computer Science},  366--383.
\newblock Springer.

\bibitem[\protect\citeauthoryear{Lembo \bgroup et al\mbox.\egroup }{2010}]{DBLP:conf/rr/LemboLRRS10}
Lembo, D.; Lenzerini, M.; Rosati, R.; Ruzzi, M.; and Savo, D.~F.
\newblock 2010.
\newblock Inconsistency-tolerant semantics for description logics.
\newblock In Hitzler, P., and Lukasiewicz, T., eds., {\em Web Reasoning and Rule Systems - Fourth International Conference, {RR} 2010, Bressanone/Brixen, Italy, September 22-24, 2010. Proceedings}, volume 6333 of {\em Lecture Notes in Computer Science},  103--117.
\newblock Springer.

\bibitem[\protect\citeauthoryear{Lembo \bgroup et al\mbox.\egroup }{2015}]{DBLP:journals/ws/LemboLRRS15}
Lembo, D.; Lenzerini, M.; Rosati, R.; Ruzzi, M.; and Savo, D.~F.
\newblock 2015.
\newblock Inconsistency-tolerant query answering in ontology-based data access.
\newblock {\em J. Web Semant.} 33:3--29.

\bibitem[\protect\citeauthoryear{Letelier \bgroup et al\mbox.\egroup }{2013}]{DBLP:journals/tods/Letelier0PS13}
Letelier, A.; P{\'{e}}rez, J.; Pichler, R.; and Skritek, S.
\newblock 2013.
\newblock Static analysis and optimization of semantic web queries.
\newblock {\em {ACM} Trans. Database Syst.} 38(4):25.

\bibitem[\protect\citeauthoryear{Lutz, Seylan, and Wolter}{2013}]{DBLP:conf/ijcai/LutzSW13}
Lutz, C.; Seylan, I.; and Wolter, F.
\newblock 2013.
\newblock Ontology-based data access with closed predicates is inherently intractable(sometimes).
\newblock In Rossi, F., ed., {\em {IJCAI} 2013, Proceedings of the 23rd International Joint Conference on Artificial Intelligence, Beijing, China, August 3-9, 2013},  1024--1030.
\newblock {IJCAI/AAAI}.

\bibitem[\protect\citeauthoryear{Ortiz}{2023}]{DBLP:conf/wollic/Ortiz23}
Ortiz, M.
\newblock 2023.
\newblock A short introduction to {SHACL} for logicians.
\newblock In Hansen, H.~H.; Scedrov, A.; and de~Queiroz, R. J. G.~B., eds., {\em Logic, Language, Information, and Computation - 29th International Workshop, WoLLIC 2023, Halifax, NS, Canada, July 11-14, 2023, Proceedings}, volume 13923 of {\em Lecture Notes in Computer Science},  19--32.
\newblock Springer.

\bibitem[\protect\citeauthoryear{P{\'{e}}rez, Arenas, and Gutierrez}{2009}]{DBLP:journals/tods/PerezAG09}
P{\'{e}}rez, J.; Arenas, M.; and Gutierrez, C.
\newblock 2009.
\newblock Semantics and complexity of {SPARQL}.
\newblock {\em {ACM} Trans. Database Syst.} 34(3):16:1--16:45.

\bibitem[\protect\citeauthoryear{Staworko, Chomicki, and Marcinkowski}{2012}]{DBLP:journals/amai/StaworkoCM12}
Staworko, S.; Chomicki, J.; and Marcinkowski, J.
\newblock 2012.
\newblock Prioritized repairing and consistent query answering in relational databases.
\newblock {\em Ann. Math. Artif. Intell.} 64(2-3):209--246.

\bibitem[\protect\citeauthoryear{ten Cate, Fontaine, and Kolaitis}{2015}]{DBLP:journals/mst/CateFK15}
ten Cate, B.; Fontaine, G.; and Kolaitis, P.~G.
\newblock 2015.
\newblock On the data complexity of consistent query answering.
\newblock {\em Theory Comput. Syst.} 57(4):843--891.

\bibitem[\protect\citeauthoryear{W3C}{2013}]{W3CValidationWorkshopReport}
W3C.
\newblock 2013.
\newblock {RDF Validation Workshop Report: Practical Assurances for Quality {RDF} Data}.
\newblock \verb|http://www.w3.org/2012/12/rdf-val/report|.

\bibitem[\protect\citeauthoryear{Wijsen}{2019}]{DBLP:journals/sigmod/Wijsen19}
Wijsen, J.
\newblock 2019.
\newblock Foundations of query answering on inconsistent databases.
\newblock {\em {SIGMOD} Rec.} 48(3):6--16.

\end{thebibliography}

\clearpage\clearpage
\appendix
\section*{Appendix}

\section{More on Well-Designed SPARQL Queries}

We first want to revisit well-designed SPARQL and the OPT-normal form introduced in \cite{DBLP:journals/tods/PerezAG09} and recalled in Section~2. 
Recall that a SPARQL query $Q$ is \emph{well-designed}, if there is no subquery $Q' = (P_1 \,\text{OPT}\ P_2)$ of $Q$ and a variable $x$, such that $x$ occurs in $P_2$, outside of $Q'$, but not in $P_1$. 
For instance, the query 
$$(\text{Prof}(x) \,\text{ OPT}\ \text{knows}(x,y)) \,\text{OPT}\ \text{email}(y,z)$$
is not well-designed, because $y$ occurs in $\text{knows}(x,y)$ (right operand of OPT) and in $\text{email}(y,z)$ (``outside'') but not in $\text{Prof}(x)$ (left operand of OPT). 
In contrast, the query 
$$(\text{Prof}(x) \, \text{OPT}\ \text{teaches}(x,y)) \,\text{OPT}\ \text{email}(x,z)$$ 
is well-designed.

The essence of the OPT normal form is that the OPT operator must not occur in the scope of the $\wedge$-operator. 
Suppose that this condition is violated, e.g., an expression of the form $P \wedge (Q \OPT R)$. 
We only consider well-designed SPARQL here. 
That is, every variable occurring in the right operand of the OPT and outside the OPT-expression must also occur in the left operand of the OPT. 
This means that every variable occurring in both $P$ and $R$ must also occur in $Q$. 
It is easy to verify that then the following equivalence holds: 
$P \wedge (Q \OPT R) \equiv
(P \wedge Q) \OPT R$
By exhaustively applying this equivalence to transform an expression of the form shown on the left-hand side into the expression on the right-hand side, one can ultimately get an equivalent \wdpt in OPT-normal form in polynomial time~\cite{DBLP:journals/tods/PerezAG09}. 
In other words, when imposing the well-designedness restriction, we may, w.l.o.g., further restrict ourselves to OPT-normal form. 

In \cite{DBLP:journals/tods/Letelier0PS13}, 
it was shown that checking if a mapping $\mu$ is an answer to a \wdptProj $\pi_X Q$ over a graph $G$ is $\sigmatwo$-complete (combined complexity). 
To see the  $\sigmatwo$-membership, we briefly sketch a guess-and-check algorithm for this task, where the ``guessed'' object is polynomially bounded w.r.t.\ the size of the input and the check is feasible in \conp. 
It is convenient to consider the tree representation $T(Q)$ of the \wdpt $Q$ (referred to as ``well-designed pattern tree'', \textsc{wdPT} for short) introduced in \cite{DBLP:journals/tods/Letelier0PS13}, namely: 
a rooted, unordered tree, where the nodes are labeled by the \bgps and each parent-child relationship corresponds to an OPT. 

One may then proceed as follows: (1) first determine the minimal subtree $T'$ of $T(Q)$ whose set of free variables (i.e., the variables of $X$ occurring in the \bgps at the nodes of $T'$) is exactly $\idom(\mu)$. 
If no such $T'$ exists or if $T'$ does not include the root of $T(Q)$, then return ``no''.  
Otherwise, (2) let $Q'$ denote the conjunction of all \bgps labeling the nodes in $T'$ and let $Y$ denote the remaining variables in $Q'$; then check if $\mu$ can be extended to $\nu$ with $\idom(\nu) = \idom(\mu) \cup Y$, such that $\nu \in \semG{Q'}$.
If not, then return ``no''. 
Otherwise, (3) let $N$ denote the set of nodes $n$ in $T$, such that $n$ is a descendant of some leaf node $m_n$ in $\mathcal{T'}$ whose labeling \bgp contains a variable in $X \setminus var(Q')$ and there is no node between $n$ and $m_n$ with this property. 
Then, for every $n \in N$ with ancestor $m_n$ in $\mathcal{T'}$, let $Q_n$ denote the conjunction of $Q'$ and all \bgps on the path from $m_n$ to $n$ (in particular, including $n$). 
If there exists a $\nu \in \semG{Q'}$ which, for all $n \in N$, can {\em not} be extended to an answer of $\pi_X Q_n$, then return ``yes''. 
Otherwise return ``no''. 

In other words, for checking if some $\mu$ is an answer to a \wdptProj $Q$ over a graph $G$, one has to check if \emph{there exists} a $\nu$ such that \emph{there does not exist} an extension of $\nu$ to an answer of one of the queries $Q_n$. 
Or, phrased as a guess-and-check algorithm that proves the $\sigmatwo$-membership \cite{DBLP:journals/tods/Letelier0PS13}: guess $\nu$ and check that there does not exist such an extension of $\nu$.
It is now also clear, why we were allowed to assume a specific form of the \wdptProj $\pi_X Q$ when checking if some mapping $\mu$ is an answer to $Q$ over a graph $G$. 
More precisely, throughout this work, we are assuming that $Q$ is of the form $(( \dots ((P\OPT R_1)\OPT R_2) \dots)\OPT R_k) $, such that $\ivar(P) \cap X = \idom(\mu)$ and each of the subqueries $R_i$ contains at least one variable from $X \setminus \ivar(P)$. 
The corresponding \textsc{wdPT} $T$ has depth 1, the root is labeled with \bgp $P$, and the $k$ child nodes are labeled with the \bgps $R_1, \dots R_k$.
In the above sketched algorithm, $P$ corresponds to $Q'$ and the conjunction of $P$ with each of the \bgps $R_i$ corresponds to one of the queries $Q_n$.

\section{More on SHACL}

The complexity of validation for SHACL with recursive constraints under the supported model semantics is \np-complete in both data and combined complexity, and for the case without recursion is \ptime-complete in combined complexity and $\mathsf{NLogSpace}$-complete in data complexity~\cite{DBLP:conf/semweb/CormanRS18,shqi-etal-2021kr}.

We note that, as observed in~\cite{shqi-etal-2021kr,DBLP:conf/wollic/Ortiz23}, the supported models semantics for SHACL constraints is closely connected with the semantics of Description Logics terminologies extended with closed predicates~\cite{DBLP:journals/entcs/FranconiIS11,DBLP:conf/ijcai/LutzSW13}. Specifically, SHACL constraints of the form $\mathsf{s} \leftrightarrow \varphi$ can be viewed as \emph{concept definitions} of the form $\mathsf{s} \equiv \varphi$, where $s$ is viewed as a class name and $\varphi$ is a concept expression in the DL $\mathcal{ALCOIQ}$ extended with regular role expressions and equalities. Note that both class and shape names are viewed as class names.  Then, the set of constraints $\mathcal{C}$ can be naturally viewed as a TBox $\mathcal{C}_T$ with only concept definitions, i.e., a terminology. The data graph $G$ and the target set $\mathcal{T}$ can be viewed as an ABox. Moreover, we assume that all shape names are open predicates and (the rest of the) class and property names are closed predicates, in the sense that their interpretation in models of the TBox is exactly as given in the data (ABox). Validation of a data graph $G$ against a shapes graph $(\mathcal{C,T})$ can then be naturally viewed as checking satisfiability of the corresponding DL knowledge base $(\mathcal{C}_T, G \cup \mathcal{T})$.

The correspondence with DLs with closed predicates can be immediately observed when defining the semantics of SHACL in terms of interpretations~\cite{DBLP:conf/wollic/Ortiz23,DBLP:conf/lpnmr/BogaertsJB22}. More precisely, an interpretation consists of an non-empty domain $\Delta = N_N$ and an interpretation function $\cdot^I$ that maps each shape name or class name $W \in N_S \cup N_C$ to a set $W^I \subseteq \Delta$ and each property name $p \in N_P$ to a set of pairs $p^I \subseteq \Delta \times \Delta$. The evaluation of complex shape expressions w.r.t.\,an interpretation $I$ is given in terms of a function $\cdot^I$ that maps a shape expression $\phi$ to a set of nodes, and a path expression $E$ to a set of pairs of nodes as in Table~\ref{tab:evaluation2}.

\begin{table}
  \renewcommand{\arraystretch}{1.5}
  \begin{tabular}{l}
    $\top^{I}=  \Delta$ \hspace{.3cm} $\texttti{c}^{I} =  \{\texttti{c}\}$  \hspace{.3cm} ${B}^{I} =  \{\texttti{c} \mid \texttti{c} \in B^I\}$ \\%
      ${ \mathsf{s} }^{I} =  \{\texttti{c}\mid \texttti{c}\in \mathsf{s}^I\}$ \quad  ${ p  }^{I}=   \{ (\texttti{a},\texttti{b})\mid (\texttti{a},\texttti{b})\in p^I \}$ \\
     $({ p^-  })^{I}=  \{ (\texttti{a},\texttti{b})\mid (\texttti{b},\texttti{a})\in p^I \}$   \\
    $({ E\cup E'  })^{I}=   { E  }^{I} \cup  {  E'  }^{I}$ \\
    $({ E\cdot E'  })^{I}=  \{(\texttti{c},\texttti{d}) \in \Delta \times \Delta \mid \exists \texttti{d'} \text{ with} $ \\ 
    \hspace{2cm}$(\texttti{c},\texttti{d'}) \in E^I \text{ and } (\texttti{d'},\texttti{d}) \in E'^I\}$ \\
    $({E^{*}  })^{I}=  \{(\texttti{a},\texttti{a}) \mid \texttti{a} \in \Delta\} \cup { E}^{I} \cup ({E\cdot E})^{I} \cup \cdots$ \\
    $ ({ \neg \phi})^{I}=  \Delta\setminus {  \phi}^{I}$   \quad \quad  $({ \phi_1\land \phi_2 })^{I} =  { \phi_1 }^{I}\cap { \phi_2 }^{I}$ \\
    $(\atleast{n}{E}{\phi})^{I}=  \{ \texttti{c} \in \Delta \mid \exists \texttti{d}_1, \ldots, \texttti{d}_n \text{ with } \} $ \\ 
    \hspace{2cm} $(\texttti{c}, \texttti{d}_i)  \in E^I, \texttti{d}_i \in \phi^{I} \text{ for each } i \in [1,n] \}$  \\
    $({ E = E' })^{I}=  \{ \texttti{c} \in \Delta \mid \forall \texttti{d} \in \Delta:(\texttti{c},\texttti{d}) \in {E}^{I} \mbox{ iff } (\texttti{c}, \texttti{d}) \in {E'}^{I} \}$   
  \end{tabular}
  \caption{Evaluation of complex shapes}
  \label{tab:evaluation2}
\end{table}

Then, an interpretation $I$ satisfies a constraint $\mathsf{s} \leftrightarrow \varphi$ if $\mathsf{s}^I = \varphi^I$, and $I$ satisfies a shapes graph $(\mathcal{C,T})$, if $I$ satisfies all constraints in $\mathcal{C}$ and $a \in \mathsf{s}^I$ for each $\mathsf{s}(a) \in \mathcal{T}$. 

Now, an interpretation $I$ is a shape assignment for a data graph $G$, if $B^I = \{c \mid B(c) \in G\}$ for each class name $B$, and $p^I = \{(c,d) \mid p(c,d) \in G\}$ for each property name $p$. Intuitively, $I$  is a shape assignment for $G$ if class and property names are interpreted as specified by $G$. Finally, a data graph $G$ validates a shapes graph $(\mathcal{C,T})$ if there exists a shape assignment $I$ for $G$ that satisfies $(\mathcal{C,T})$. For more details on SHACL and the relation with DLs, we refer to~\cite{DBLP:conf/wollic/Ortiz23,DBLP:conf/lpnmr/BogaertsJB22}.

\section{Full Proofs for Brave and AR Semantics}

\begin{theoremAppendix}
\label{thm:app:first}
The following statements are true for data complexity:
\begin{itemize}
    \item \LEquEx is \np-c for $\mathcal{L}\in \{\bgp$, $\bgpProj$,  $\wdpt$, $\wdptProj\}$.
    \item \LCardEx and \LCard are \thetatwo-c for $\mathcal{L}\in \{\bgp$, $\bgpProj$,  $\wdpt, \wdptProj\}$.
    \item \LSubsEx is \sigmatwo-c for $\mathcal{L}\in \{\bgp$, $\bgpProj$,  $\wdpt$, $\wdptProj\}$.
    \item \LEqu is \conp-c for $\mathcal{L}\in \{\bgp$, $\bgpProj$,  $\wdpt$, $\wdptProj\}$.
    \item \LSubs is \pitwo-c for $\mathcal{L}\in \{\bgp$, $\bgpProj$,  $\wdpt$, $\wdptProj\}$.
\end{itemize}
The following statements are true for combined complexity:
\begin{itemize}
    \item \LEquEx is \np-c for $\mathcal{L}\in \{\bgp$, $\bgpProj\}$.
    \item \LCardEx is \thetatwo-c for $\mathcal{L}\in \{\bgp$, $\bgpProj\}$.
    \item \LSubsEx is \sigmatwo-c for $\mathcal{L}\in \{\bgp$, $\bgpProj$,  $\wdpt$, $\wdptProj\}$.
    \item \LEqu is \conp-c for $\mathcal{L}\in \{\bgp$, $\wdpt\}$.
    \item \LCard is \thetatwo-c for $\mathcal{L}\in \{\bgp$, $\wdpt\}$.
    \item \LSubs is \pitwo-c for $\mathcal{L}\in \{\bgp$, $\bgpProj$,  $\wdpt\}$.
\end{itemize}
\end{theoremAppendix}

\begin{proof}
    {\small$\blacktriangleright$} (Hardness for \bgpEquEx data complexity) 
    We reduce from 3-CNF \textsc{SAT}.
    I.e., instances are propositional formula $\phi(\texttti{X}) = \bigwedge_{j} (\texttti{l}_{1j} \lor \texttti{l}_{2j} \lor \texttti{l}_{3j})$, where $\texttti{l}_{ij}$ are literals.
    For the encoding, the nodes of the data graph $G$ consist of:
    \begin{itemize}
        \item All literals $\texttti{x}, \lnot \texttti{x}$ for variables $\texttti{x}\in\texttti{X}$.
        \item Clauses $\texttti{C}_{j}= (\texttti{l}_{1j} \lor \texttti{l}_{2j} \lor \texttti{l}_{3j})$.
        \item Two fresh auxiliary nodes $\texttti{s}, \texttti{e}$. 
    \end{itemize}
    Literals $\texttti{l}_{ij}$ are added to the class $Lit$, clauses $\texttti{C}_{j}$ to $Cl$, and $\texttti{s}$ to $Phi$.
    Thus, intuitively, $\texttti{s}$ represents $\phi$.
    
    We now connect the nodes as they are connected logically:
    \begin{itemize}
        \item Literals $\texttti{l}$ are connected to their duals $\texttti{l}_d$ via $dual(\texttti{l}, \texttti{l}_d)$.
        \item Literals $\texttti{l}$ are connected to the clauses $\texttti{C}$ they appear in via $or(\texttti{l}, \texttti{C})$.
        \item Each clause $\texttti{C}_{j}$ is connected to $\texttti{s}$ via $and(\texttti{C}_{j}, \texttti{s})$. 
    \end{itemize}
    
    Then, literals $\texttti{l}$, clauses $\texttti{C}$ and $\texttti{s}$ are all added to both the class $T$ and $F$, indicating that they can be either \textit{true} or \textit{false}.
    Furthermore, an arbitrary linear order $\preceq_{next}$ starting in $\texttti{s}$ and ending in $\texttti{e}$ is introduced on all the nodes of $G$, where $next$ indicates immediate predecessors and successors ($\preceq_{next}$ is the transitive closure of $next$ ). 
    
    Consider the following constraints $\mathcal{C}\colon$
    \begin{align}
    \begin{split}
        \mathsf{lit} \leftrightarrow {} & Lit \land ((T \land \lnot F \land \exists dual.F) \\
         & \quad {} \lor (F  \land \lnot T \land \exists dual.T)) 
    \end{split}\label{cnst:lit}\\
    \begin{split}
        \mathsf{cl} \leftrightarrow {} & Cl \land {=_3}or^- \land {=_1}and \land (\\
        & (F \land \lnot T \land \forall or^-.F) \\
        & \lor (T \land \lnot F \land \exists or^-.T) ) 
    \end{split}\label{cnst:cl}\\
    \begin{split}
        \mathsf{phi} \leftrightarrow {} & Phi \land ((\forall and^-.T \land T \land \lnot F) \\
        & \quad {}\lor (\exists and^-.F \land F \land \lnot T))
    \end{split}\label{cnst:phi}\\
    \begin{split}
        \mathsf{val} \leftrightarrow {} & \forall next^*.(\texttti{e} \lor \mathsf{phi} \lor \mathsf{lit} \lor \mathsf{cl}) \\
        & \land \exists next^*.\texttti{e}
    \end{split}\label{cnst:val}
    \end{align}
    The target of the shapes graph is $\mathcal{T}=\{\mathsf{val}(\texttti{s})\}$,
    $H=\emptyset$, $\mu = \{x\mapsto \texttti{s}\}$, and the query is $Q=T(x)$.

    First consider the Constraint \ref{cnst:val} of the form $\mathsf{val} \leftrightarrow \phi$.
    For $G_R$ to be a repaired graph, there has to be a supported model $I \supseteq G_R$ such that $\texttti{s}\in \llbracket \phi \rrbracket^I$, i.e., $\texttti{s}$ has to ``satisfy'' Constraint \ref{cnst:val}.
    Thus, $\texttti{s}$ has to be connected to $\texttti{e}$ via multiple jumps along $next$.
    But to reach $\texttti{e}$ from $\texttti{s}$ through $next$, all nodes of $\texttti{n}\in V(G_R)$ have to be traversed and, thus, $(\texttti{s},\texttti{n})\in \llbracket next^*\rrbracket^I$.
    Consequently, each node except $\texttti{e}$ has to satisfy one of the first 3 constraints.
    As $Lit$ can only contain literals, $Cl$ can only contain clauses, and $Phi$ can only contain $\texttti{s}$ ($H$ is empty), literals have to satisfy the Constraint~\ref{cnst:lit}, clauses Constraint~\ref{cnst:cl}, and $\texttti{s}$ Constraint~\ref{cnst:phi}.

    Intuitively, each repair has the (independent) choice of what truth values to assign to each variable $\texttti{x}$ by retaining either $T(\texttti{x})$ and not $F(\texttti{x})$ or the other way around.
    The truth values of the dual literals, the clauses, and the formula itself (represented by $\texttti{s}$) are then all functionally determined by this choice due to the constraints.
    Notice that only the classes $T$ and $F$ can differ between a repaired graph $G_R$ and $G$.
    Consequently, $T(\texttti{s})$ is in the repaired graph if and only if $\phi$ evaluates to true under the given choice of the truth values for the variables.
    Since each truth value assignment constitutes a repair, $\mu$ is an answer to $Q$ over some repaired graph if and only if $\phi$ is satisfiable.
    Thus, the reduction is correct. 

    {\small$\blacktriangleright$} (Hardness for \bgpEqu data complexity)
    Let us reconsider the previous reduction but with the query $Q=F(x)$.
    $\mu$ is an answer to $Q$ over all repaired graphs if and only if $\phi$ is unsatisfiable.
    We conclude, \bgpEqu is \conp-hard.
    
    {\small$\blacktriangleright$} (Hardness for \bgpCardEx data complexity) 
    We reduce from the \thetatwo-complete problem \textsc{CardMinSat} \cite{DBLP:journals/lmcs/CreignouPW18}.
    An instance is a propositional formula $\phi$ in 3-CNF including the variable $\texttti{x}_1$.
    $\phi$ is a yes-instance if the variable $\texttti{x}_1$ is true in some model minimizing the number of variables set to true.
    
    For the encoding, we take the graph $G$ as before and augment it.
    To that end, we add the following new nodes:
    \begin{itemize}
        \item For each variable $\texttti{x}\in \texttti{X}$ an additional copy $\texttti{x}_{copy}$ (copying \texttti{x}).
        \item For each variable $\texttti{x}\in \texttti{X}$ an additional copy $\texttti{s}_{\texttti{x}}$ (copying \texttti{s}).
        \item A copy $\texttti{s}_{copy}$ (copying \texttti{s}). 
    \end{itemize}
    
    The copies $\texttti{x}_{copy}$ are added to $F$ and the class $Copy_\texttti{x}$ while the copies $\texttti{s}_{\texttti{x}}, \texttti{s}_{copy}$ are added to $T$ and $Copy_\texttti{s}$.
    The class memberships to $F$ and $T$ are not intended to represent truth values for the copies but to add a ``cost'' to setting a variable to true and to not satisfy the formula.
    Furthermore, we add $\texttti{x}_1$ to $X_1$ and $\texttti{e}$ to $T,F$ and $E$.
    
    As for new properties, we connect copies via $copy_\texttti{x}(\texttti{x}_{copy}, \texttti{x})$, $copy_\texttti{s}(\texttti{s}_{\texttti{x}}, \texttti{s})$, $copy_\texttti{s}(\texttti{s}_{copy}, \texttti{s})$
    Furthermore, another arbitrary linear order $\preceq_{next'}$ starting in $\texttti{s}$ and ending in $\texttti{e}$ is introduced on all copies $\texttti{x}_{copy}, \texttti{s}_{\texttti{x}}, \texttti{s}_{copy}$ and additionally $\texttti{x}_1$ (together with $\texttti{s}, \texttti{e}$), where $next'$ indicates immediate predecessors and successors.
    
    The constraints then consist of the previous ones plus the following:
    \begin{align*}
        \mathsf{copy}_{\texttti{x}} \leftrightarrow {} & Copy_\texttti{x} \land \exists copy_\texttti{x} \land (\\
        &( F \land \exists copy_\texttti{x}.F) \\
        & \lor (\lnot F \land \exists copy_\texttti{x}.T))  \\
        \mathsf{copy}_{\texttti{s}} \leftrightarrow {} & Copy_\texttti{s} \land  \exists copy_\texttti{s} \land ( \\
        & (\lnot T \land \exists copy_\texttti{s}.F) \\
        & \lor ( T \land \exists copy_\texttti{s}.T))  \\
        \mathsf{x}_{1} \leftrightarrow {} & X_1\\
        \mathsf{e} \leftrightarrow {} & E \land  ( \\
        & (T \land {=}_2 next'^{-*}.((Phi \land T) \lor (X_1 \land T))) \\
        & \lor ( F \land \exists next'^{-*}.((Phi \land F) \lor (X_1 \land F)))  \\
        \mathsf{val}' \leftrightarrow {} & \forall next'^*.(\texttti{s} \lor \mathsf{e} \lor \mathsf{x}_1 \lor \mathsf{copy}_{\texttti{x}} \lor \mathsf{copy}_{\texttti{s}}) \\
        & \land \exists next'^*.\texttti{e}
    \end{align*}

    The targets of the shapes graph are $\mathcal{T}=\{\mathsf{val}(\texttti{s})$, $\mathsf{val}'(\texttti{s})\}$,
    $H=\emptyset$, $\mu = \{x\mapsto \texttti{e}\}$, and the query is $Q=T(x)$.
    
    Intuitively, repairs now have to remove copy variables $\texttti{x}_{copy}$ from $F$ if the corresponding variable $\texttti{x}$ is set to true.
    Thus, in term of the cardinality of that repair, it is punished per variable set to true.
    At this point, notice that previously all repairs had the same cardinality, no matter the truth assignment they represent.
    Continuing, if the formula $\phi$ evaluates to false under the variable assignment that the repair represents, it has to remove all copy variables $\texttti{s}_\texttti{x}$ and $\texttti{s}_{copy}$ from $T$.
    Hence, in such a case, it is punished per variable in $\phi$ plus 1.
    Consequently, the $\leq$-repairs contain the minimal models of $\phi$ (and no other models).    
    Moreover, $\mu$ is an answer to $Q$ over some $\leq$-repaired graph $G_R$, i.e., $\texttti{e}$ is set to true, if and only if $\texttti{x}_1$ is true in some minimal model of $\phi$.
    
    {\small$\blacktriangleright$} (Hardness for \bgpCard data complexity)
    Let us reconsider the previous reduction but with the query $Q=\land F(x)$.
    We have to distinguish two cases.
    If $\phi$ is satisfiable, $\mu$ is an answer to $Q$ over all $\leq$-repaired graphs if and only if $\texttti{x}_1$ is set to false in every minimal model of $\phi$.
    If $\phi$ is unsatisfiable, $\mu$ is an answer to $Q$ over all $\leq$-repaired graphs.
    Thus, the reduction is valid for co-\textsc{CardMinSat}.
    We conclude, \bgpCard is \thetatwo-hard.
    
    {\small$\blacktriangleright$} (Hardness for \bgpSubsEx data complexity) 
    To show \sigmatwo-hardness we have to extend the previous reductions, proceeding further up the polynomial hierarchy.
    To that end, we reduce from the $\sigmatwo$-complete problem $2$-\textsc{QBF} on formulas in 3-CNF.
    The instance consist of a propositional formula $\phi(X,Y) = \bigwedge_{j} (\texttti{l}_{1j} \lor \texttti{l}_{2j} \lor \texttti{l}_{3j})$.
    $\phi(X,Y)$ is a yes-instance if there exists a truth assignment for the $X$-variables that is not extendable to a model of $\phi$.
    
    For the encoding, we take the data graph $G$ from before, excluding the class $Lit$, and extend it further:
    \begin{itemize}
        \item For each literal $\texttti{l}$ of a $X$ (resp. $Y$) variable we add $Lit_X(\texttti{l})$ (resp. $Lit_Y(\texttti{l})$) to $G$.
        \item We add $NoExt(\texttti{s})$ to $H$ to indicate when \textit{not} instantiate the $Y$ variables.
    \end{itemize}
    
    Now, consider the following set of constraints $\mathcal{C}$:
    \begin{align*}
        \mathsf{lit}_X \leftrightarrow {} & Lit_X \land ((T \land \lnot F \land \exists dual.F) \\
         & \quad {}\lor (F  \land \lnot T \land \exists dual.T))  \\
         \mathsf{lit}_Y \leftrightarrow {} & Lit_Y \land ( \\
         & (\lnot \exists (next\cup next^-)^*.NoExt \land ( \\
         & \quad(T \land \lnot F \land \exists dual.F) \\
         & \quad {}\lor (F  \land \lnot T \land \exists dual.T))  \\
         & \lor ( \exists  (next\cup next^-)^*.NoExt \land \lnot F  \land \lnot T))  \\
        \mathsf{cl} \leftrightarrow {} & Cl \land {=_3}or^- \land {=_1}and \land (\\
        & (\lnot \exists  (next\cup next^-)^*.NoExt \land ( \\
        & \quad (F \land \lnot T \land \forall or^-.F) \\
        & \quad {} \lor (T \land \lnot F \land \exists or^-.T))  \\
         & \lor ( \exists  (next\cup next^-)^*.NoExt \land \lnot F  \land \lnot T))  \\
        \mathsf{phi} \leftrightarrow {} & Phi \land ((\lnot NoExt\land T \land \lnot F \land \forall and^-.T)  \\
        & \quad {}\lor (NoExt\land \lnot F \land \lnot T))  \\
        \mathsf{val} \leftrightarrow {} & \forall next^*.(\texttti{e} \lor \mathsf{phi} \lor \mathsf{lit}_X \lor \mathsf{lit}_Y \lor \mathsf{cl}) \\
        & \land \exists next^*.\texttti{e}
    \end{align*}
    The targets are again $\mathcal{T} = \{\mathsf{val}(\texttti{s})\}$.
    The query is now $Q = NoExt(x)$ and $\mu = \{x \mapsto \texttti{s}\}$.
    
    The idea is as follows.
    The repairs have the choice of whether to instantiate only the $X$ variables or both the $X$ and the $Y$ variables.
    In the first case, $\texttti{s}$ (representing the formula) has to be added to the class $NoExt$ and the $X$ variables instantiated but, at the same time, all literals of the $Y$ as well as the clauses and the formula itself have to be removed from both $T$ and $F$.
    In the second case, the formula remains outside the class $NoExt$, but both the $X$ and the $Y$ variables have to be instantiated, and the truth values of the dual literals, the clauses, and the formula itself are then all functionally determined by this choice.
    Furthermore, the formula itself has to evaluate to true.
    
    Thus, intuitively, repairs correspond 1-1 to instantiations of only the $X$ variables or a combination of the $X$ and the $Y$ that satisfy the formula.
    When the instantiation is the same on the $X$ variables, a repair of the second kind is a subset of the repair of the first kind.
    Thus, $\mu$ is an answer to $Q$ over some $\subseteq$-repaired graph $G_R$ if and only if the instantiation of the $X$ variables represented by $G_R$ cannot be extended to model of $\phi$. 
    
    {\small$\blacktriangleright$} (Hardness for \bgpSubs data complexity) 
    Let us reconsider the previous reduction but with the query $Q=T(x)$.
    Then, $\mu$ is an answer to $Q$ over all $\subseteq$-repaired graphs if and only if every truth assignment of the $X$ variables can be extended to a model of $\phi$.
    Thus, we can conclude that \bgpSubs is \pitwo-hard.
    
    {\small$\blacktriangleright$} (Membership for \wdptProjPrecEx data complexity)
    To check whether $\mu$ is an answer to $\pi_X Q$ with $ Q =(( \dots ((P\OPT P_1)\OPT P_2) \dots)\OPT P_k) $
    over some $\preceq$-repaired graph, we can do the following:
    \begin{enumerate}
        \item Guess a $\preceq$-repaired graph $G_R$.
        \item Go over all $\mu_P\colon (\ivar(P)\setminus \idom(\mu))\rightarrow V(G)$ and check that $(\mu \cup \mu_P)$ is an answer to $P$ over $G_R$.
        \item If that is the case, go over all $P_i$ and $\mu_{P_i}\colon (\ivar(P_i)\setminus \ivar(P)) \rightarrow V(G)$ and check that $(\mu \cup \mu_P\cup \mu_{P_i})|_{\ivar(P_i)}$ is no answer to $P_i$ over $G_R$.
    \end{enumerate}
    As we are considering data complexity, steps 2 and 3 together only take polynomial time.
    Guessing an arbitrary repaired graph can be done by a NTM as we can guess the repair $R$ together with a supporting model for $G_R$.
    Thus, \wdptProjEquEx is in \np.
    Likewise, guessing an $\subseteq$-repaired graph can be done by a NTM with access to an \np-oracle as we can guess the repair $R$ together with a supporting model for $G_R$ and then check that there is no smaller repair $R'$.
    Thus, \wdptProjSubsEx is in \sigmatwo.
    
    For $\leq$-repairs, we have to argue slightly differently.
    First, we compute the size of a minimal repair with logarithmically many \np-oracle calls using binary search.
    The oracle calls are of the form: ``Does there exist a repair of size $\leq c$''.
    Given the minimal size $k$, we can then ask whether there exists a repair $R$ of size $k$ such that $\mu$ is an answer to $\pi_XQ$ over $G_R$.
    This is then clearly an \np question and thus, \wdptProjCardEx is in \thetatwo.
    
    {\small$\blacktriangleright$} (Membership for \bgpProjEquEx and \bgpProjCardEx combined complexity)
    Here, we have to combine the guess for the repair with the guess for the bound variables.
    Concretely, for a query $\pi_X Q(X,Y)$, we have to guess a ($\leq$-minimal) repaired graph $G_R$ and a $\mu_Y\colon Y\rightarrow V(G)$ such that $\mu\cup \mu_Y$ is an answer to $Q$ over $G_R$.
    For \bgpProjEquEx, we can already conclude \np-membership.
    For $\leq$-repairs, we can again compute the size of a minimal repair with logarithmically many oracle calls beforehand.
    Thus, \bgpProjCardEx is in \thetatwo.
    
    {\small$\blacktriangleright$} (Membership for \wdptProjSubsEx combined complexity)
    As in the previous membership, we have to combine the guess for the bound variables with the guess for the repair.
    However, even more, we have to then check if there is an extension to the variables in the optional part with an oracle call, similar to how we check for $\subseteq$-minimality of the repair, i.e., we can do the following:
    \begin{enumerate}
        \item Guess a repair $R$, together with a $\mu_P \colon (\ivar(P) \setminus \idom(\mu))\rightarrow V(G)$.
        \item Check that $G_R$ is valid and $\mu\cup \mu_P$ is an answer to $P$ over $G_R$.
        \item Check that there is no repair $R'\subsetneq R$.
        \item For each $P_i$, check that there is no $\mu_{P_i} \colon (\ivar(P_i)\setminus \ivar(P)) \rightarrow V(G)$ such that $(\mu\cup\mu_P\cup\mu_{P_i})|_{\ivar(P_i)}$ is an answer to $P_i$ over $G_R$.
    \end{enumerate}
    Note that all checks only require access to an \np-oracle.
    Thus, \wdptProjSubsEx is in \sigmatwo.
    
    {\small$\blacktriangleright$} (Membership for \wdptProjPrec data complexity)
    We show membership for the co-Problem, i.e., whether $\mu\colon X \rightarrow V(G)$ is not an answer to $\pi_X Q$ over some $\preceq$-repaired graph.
    Notice that given a repaired graph $G_R$, in terms of data complexity, it is equally easy to check that $\mu$ is no answer to $Q$ as it is to check that it is an answer to $Q$.
    Thus, the methods developed for \wdptProjPrecEx can be reused.
    Hence, \wdptProjEqu is in \conp \wdptProjCard is in \thetatwo, and \wdptProjSubs is in \pitwo.

    {\small$\blacktriangleright$} (Membership for \wdptEqu and \wdptCard combined complexity)
    Again, we show membership for the co-Problem, i.e., whether $\mu\colon \ivar(P) \rightarrow V(G)$ is not an answer to $Q = (( \dots ((P\OPT P_1)\OPT P_2) \dots)\OPT P_k)$ over some ($\leq$-minimal) repaired graph.
    We can argue similarly to the cases \bgpProjEquEx and \bgpProjCardEx, but this time, we guess a $P_i$ and the values of $\ivar(P_i)$.
    I.e., we can proceed as follows:
    \begin{enumerate}
        \item If $\mu$ is an answer, guess a ($\leq$-)repair $R$.
        \item Guess a $P_i$ and a $\mu_{P_i}\colon \ivar(P_i) \setminus \ivar(P)\setminus \ivar \rightarrow V(G)$.
        \item Is $(\mu\cup \mu_{P_i})|_{\ivar(P_i)}$ an answer to $P_i$ or is $\mu$ \textit{not} an answer to $P$ over $G_R$?
    \end{enumerate}
    For arbitrary repairs, this process is in \np and, thus, \wdptEqu is in \conp.
    As before, for $\leq$-repairs, we can first compute the size of a minimal repair and then follow the process.
    Therefore, \wdptCard is in \thetatwo.
    
    {\small$\blacktriangleright$} (Membership for \wdptSubs and \bgpProjSubs combined complexity)
    One last time, let us consider the co-problem, i.e., whether $\mu$ is not an answer to $Q$ over some $\subseteq$-repaired graph.
    To achieve \sigmatwo membership, we can simply answer the query after guessing the repair with a single oracle call.
    I.e., we can proceed as follows:
    \begin{enumerate}
        \item Guess a repair $R$.
        \item Check that there is no repair $R' \subsetneq R$.
        \item Check that $\mu$ is no answer to $Q$ over $G_R$
    \end{enumerate}
    Clearly, all checks only require access to an \np-oracle.
    Thus, \wdptSubs and \bgpProjSubs are both in \pitwo.
\end{proof}

\begin{theoremAppendix}
The following statements are true for combined complexity:
\begin{itemize}
    \item \LEquEx and \LCardEx are \sigmatwo-c for $\mathcal{L}\in \{\wdpt$, $\wdptProj\}$.
    \item \bgpProjEqu and \bgpProjCard are \pitwo-c.
    \item \wdptProjPrec is \pithree-c for ${\preceq}{}\in{} \{=,\leq,\subseteq\}$.
\end{itemize}
\end{theoremAppendix}

\begin{proof}
    For the hardness proofs, we use the $k$-\textsc{Rounds-3-Colorability} problem as our reference problem.
    $k$-\textsc{Rounds-3-Colorability} problem is a game played by two Players on a graph for $k$ rounds.
    Player 1 starts round 1 by coloring the degree 1 vertices.
    Then Player 2 proceeds to color the degree 2 vertices.
    Then Player 1 again colors the degree 3 vertices, and so on \dots.
    In the last round $k$, all vertices are colored by the Turn Player, and this Player wins if the final coloring is a valid 3-coloring.
    A graph $G'$ is a yes-instance if Player 1 has a winning strategy. 
    $k$-\textsc{Round-3-Colorability} is know to be \sigmak-complete \cite{DBLP:journals/jcss/AjtaiFS00}.

    {\small$\blacktriangleright$} (Hardness for \wdptEquEx and \wdptCardEx)
    We reduce from $2$-\textsc{Rounds-3-Colorability}.
    To that end, let $G'$ be an arbitrary graph on the nodes $\texttti{v}_1,\dots,\texttti{v}_n$.
    To encode this problem we first construct the following data graph $G$.
    The nodes of $G$ are
    \begin{itemize}
        \item the vertices of $G'$ and
        \item 5 extra nodes $\texttti{r}, \texttti{b}, \texttti{g}, \texttti{s}, \texttti{e}$.
    \end{itemize}
    Then we add the following properties and classes:
    \begin{itemize}
        \item For every vertex $\texttti{v} \in V(G')$ we add $col(\texttti{v},\texttti{r})$, $col(\texttti{v},\texttti{g})$, $col(\texttti{v},\texttti{b})$ to indicate that every vertex can be colored by every color.
        \item We add $neq(\texttti{c}, \texttti{c}')$ for $\texttti{c}, \texttti{c}' \in \{\texttti{r}, \texttti{g}, \texttti{b}\}, \texttti{c}\neq \texttti{c}'$ to distinguish colors.
        \item We add $L(\texttti{u})$ to $G$ for every leaf (degree 1 node) $\texttti{u}$ of $G'$.
        \item We add $I(\texttti{v})$ for every non-leaf node $\texttti{v}$ of $G'$.
    \end{itemize}    
    Furthermore, we again encode an arbitrary linear order $\preceq_{next}$ on all the nodes of $G$ that starts in $\texttti{s}$ and ends in $\texttti{e}$.
    To that end, $next$ is added to $G'$ which indicates immediate predecessors and successors. 

    Next, we construct a shapes graph $(\mathcal{C},T)$ as follows.
    We add the following constraints to $\mathcal{C}$:
    \begin{align*}
        \mathsf{leaf}  \leftrightarrow {}& L \land  {=_1} col \\
        \mathsf{inner}  \leftrightarrow{} & I \land  {=_3} col\\
        \mathsf{valV}  \leftrightarrow {}& \forall next^*.(\texttti{s}\lor \texttti{e} \lor \mathsf{leaf} \lor \mathsf{inner})\\
        & \land \exists next^*.\texttti{e}\\
        \mathsf{valC}  \leftrightarrow {}& {=_2} neq
    \end{align*}
    The targets are $\{\mathsf{valV}(\texttti{s}), \mathsf{valC}(\texttti{r})$, $ \mathsf{valC}(\texttti{b})$, $ \mathsf{valC}(\texttti{g})\} = \mathcal{T}$.

    Furthermore, $H = \emptyset$, $\mu = \{\}$, and the query is the following:
    \begin{align*}
    Q = \top  \OPT \bigwedge_{i = 1}^n col(\texttti{v}_i,x_i) \land \bigwedge_{(\texttti{v}_i,\texttti{v}_j)\in G'} neq(x_i,x_j).
    \end{align*}
    The idea of the shapes graph is that every repaired graph $G_R$ ``selects'' a color for the leaf nodes but does not touch the inner nodes.
    This is achieved by the fact that $\mathsf{valV}(\texttti{s})$ can only be validated when every vertex is either in $\llbracket\mathsf{leaf}\rrbracket^I$ or in $\llbracket\mathsf{inner}\rrbracket^I$.
    The remaining targets guarantee that the $neq$ property is left untouched.
    The query $Q$ together with $\mu$ then asks whether there is no valid coloring using the colors $\{\texttti{c} \mid col(\texttti{v},\texttti{c})\in G_R\}$ for each $\texttti{v}\in G'$.
    All repairs are $\leq$-repairs and correspond 1-1 to possible colorings of the leaves.
    Hence, $\mu$ is an answer to $Q$ over some repaired graph if and only if there is a coloring of the leaf nodes that cannot be extended to a coloring of the whole graph, i.e., Player 1 has a winning strategy.

    {\small$\blacktriangleright$} (Hardness for \bgpProjEqu and \bgpProjCard)
    We again reduce from co-$2$-\textsc{Round-3-Colorability}, i.e., we ask whether there exists a coloring of the whole graph given any coloring of the leaves.
    For this, given a graph $G'$, we can copy $G,C, T,$ and $\mu$ from above, while only marginally modifying the query to the boolean query $\pi_\emptyset Q$ with
    \begin{align*}
        Q = \bigwedge_{i = 1}^n col(\texttti{v}_i,x_i) \land \bigwedge_{(\texttti{v}_i,\texttti{v}_j)\in G'} neq(x_i,x_j).
    \end{align*}
    $\mu$ is then an answer to $\pi_\emptyset Q$ over every repaired graph $G_R$ if and only if there exists a coloring of the whole graph given any coloring of the leaves, i.e., Player 1 has no winning strategy.

    {\small$\blacktriangleright$} (Hardness for \wdptProjPrec)
    We reduce from co-3-\textsc{Rounds-3-Coloring Extension}, i.e., we ask whether there exists a coloring of the degree 2 nodes for every coloring of the leaves that cannot be extended to all vertices.
    Again, given a graph $G'$, we can copy $G,C, T,$ and $\mu$ from above while only needing to slightly modify the query.
    To that end, let us assume $\texttti{v}_1,\dots,\texttti{v}_{m-1}$ to be the vertices of degree 2 or less of $G'$ and $\texttti{v}_{m},\dots,\texttti{v}_n$ to be the remainder.
    Then, consider the query $\pi_{X^n_m}Q$ with $X^n_m = \{x_m,\dots, x_n\}$ and
    \begin{align*}
    Q = &\bigwedge_{i = 1}^{m-1} c(\texttti{v}_i,x_i) \OPT{} \\
    &\bigwedge_{i = m}^n c(\texttti{v}_i,x_i) \land \bigwedge_{(\texttti{v}_i,\texttti{v}_j)\in G'} neq(x_i,x_j).
    \end{align*}
    $\mu = \{\}$ is then an answer to $Q$ over every repaired graph $G_R$ if and only if, given any coloring of the leaves (i.e., $G_R$), there exists a coloring of the degree 2 vertices (i.e., values for $x_1,\dots,x_{m-1}$) that cannot be extended to the whole graph (i.e., to $X_m^n$).
    Thus, this is exactly the case when Player 1 has no winning strategy.

    {\small$\blacktriangleright$} (Membership for \wdptProjEquEx, \wdptProjCardEx, \bgpProjEqu, and \bgpProjCard)
    We can copy the proof for the \wdptProjSubsEx (combined complexity) case (resp. \bgpProjSubs) almost one to one.
    The only change required is, that for \wdptProjEquEx (resp. \bgpProjEqu), the minimality of $R$ does not need to be checked, and for \wdptProjCardEx (resp. \bgpProjCard), the minimality of $R$ has to be checked against any $R'$ such that $|R'| < R$.
    
    {\small$\blacktriangleright$} (Membership for \wdptProjPrec)
    As always, we consider the co-problem, i.e., whether $\mu$ is no answer of $\pi_X Q$ over some $\preceq$-repaired graph $G_R$.
    Let us consider the following procedure:
    \begin{enumerate}
        \item Guess a repair $R$.
        \item Check that $R$ is $\preceq$-minimal.
        \item Check that $\mu$ is no answer of $Q$ over $G_R$.
    \end{enumerate}
    The checks clearly only require access to a \sigmatwo-oracle (in fact, only the second check), and, thus, \wdptProjPrec is in \pithree for ${\preceq}\in \{=,\leq,\subseteq\}$.
\end{proof}

\section{Full Proofs for IAR Semantics}

\begin{theoremAppendix}
The following statements are true for data complexity:
\begin{itemize}
    \item \LEquIn is \conp-c for $\mathcal{L}\in \{\bgp$, $\bgpProj\}$.
    \item \LCardIn is \thetatwo-c for $\mathcal{L}\in \{\bgp$, $\bgpProj$,  $\wdpt$, $\wdptProj\}$.
    \item \LSubsIn is \pitwo-c for $\mathcal{L}\in \{\bgp$, $\bgpProj\}$.
\end{itemize}
The following statements are true for combined complexity:
\begin{itemize}
    \item \bgpEquIn is \conp-c.
    \item \bgpProjEquIn is \thetatwo-c.
    \item \LCardIn is \thetatwo-c for $\mathcal{L}\in \{\bgp$, $\bgpProj$,  $\wdpt\}$.
    \item \LSubsIn is \pitwo-c for $\mathcal{L}\in \{\bgp$, $\bgpProj\}$.
\end{itemize}
\end{theoremAppendix}

\begin{proof}
    {\small$\blacktriangleright$} (Hardness for \bgpPrecIn data complexity)
    Notice that an atom $\mu(B(X))$ appears in every $\preceq$-repaired graph if and only if it appears in the intersection of all $\preceq$-repaired graphs.
    Thus, the hardness immediately carries over from the case for $\forall$-semantics.

    {\small$\blacktriangleright$} (Hardness for \bgpProjEquIn combined complexity)
    To show \thetatwo-hardness, we reduce from the problem \textsc{CardMin-Precoloring} on graphs with degree bounded by 5 (for hardness of the reference problem see Lemma \ref{lem:minPre} in Section \ref{sec:reference}) that asks whether for a graph $G'$ and pre-coloring $c:V(G') \rightarrow 2^{\{\texttti{r}, \texttti{g}, \texttti{b}\}}$, a given vertex $\texttti{v}_1$ is colored $\texttti{g}$ in some 3-coloring of $G'$ that minimizes the use of the color $\texttti{g}$.
    W.l.o.g., we assume $(G',c)$ to be 3-colorable.

    To encode this problem, we construct the following data graph $G$.
    The nodes of $G$ consist of
    \begin{itemize}
        \item the vertices $V(G') = \{\texttti{v}_1,\dots,\texttti{v}_n\}$ added to the class $Vtx$,
        \item the colors $\texttti{r}, \texttti{b}, \texttti{g}$, 
        \item and elements $I=\{\texttti{i}_1,\dots, \texttti{i}_n\}$ to count added to the class $Cnt$.
        \item two auxiliary fresh nodes $\texttti{s},\texttti{e}$ also added to the class $Cnt$, 
    \end{itemize}
    Furthermore, we add the following atoms: 
    \begin{itemize}
        \item For vertices $\texttti{v}\in V(G')$ we $col(\texttti{v}, \texttti{c})$ for $\texttti{c}\in c(\texttti{v})$ as possible color.
        \item We connect $\texttti{v}\in V(G')$ to its $i\leq 5$ neighbors by adding $e(\texttti{u}, \texttti{v})$ for $(\texttti{u}, \texttti{v})\in G'$ and $D_i(\texttti{v})$.
        \item For each subset $C\subseteq \{\texttti{r}, \texttti{g}, \texttti{b}\}, C\neq \emptyset$ we add classes $Pre_C$ with $Pre_C(\texttti{c})$ for every $\texttti{c}\in C$.
        \item We add $neqCol(\texttti{c}, \texttti{c}')$ for $\texttti{c}, \texttti{c}' \in \{\texttti{r}, \texttti{g}, \texttti{b}\}, \texttti{c}\neq \texttti{c}'$ to distinguish colors.
        \item We add $neqCnt(\texttti{i}_i,\texttti{i}_j)$, $neqCnt(\texttti{s},\texttti{i}_i)$, $ neqCnt(\texttti{i}_i,\texttti{e})$ for $i,j = 1,\dots, n$ and $i\neq j$ to distinguish counters.
        \item We connect each counter $\texttti{i}\in I$ to each color, i.e., via $allowed(\texttti{i}, \texttti{c})$ for  $\texttti{c} \in \{\texttti{r}, \texttti{g}, \texttti{b}\}$.
        \item We connect each counter $\texttti{i}\in  I$ to each vertex $\texttti{v}\in V(G')$ via $match(\texttti{i}, \texttti{v})$.
        \item We connect $\texttti{s},\texttti{e}$ to themselves, i.e., we add $match(\texttti{s}, \texttti{s})$ and $match(\texttti{e}, \texttti{e})$.
        \item We add a specific order on the counters by adding $nextCnt(\texttti{i}_i,\texttti{i}_{i+1})$ for $i=1,\dots,n-1$ and $nextCnt(\texttti{s},\texttti{i}_1)$ as well as $nextCnt(\texttti{i}_n,\texttti{e})$.
        \item We fully connect the vertices by adding $nextVtx(\texttti{v}_i,\texttti{v}_j),$ $nextVtx(\texttti{s},\texttti{v}_i),$ and $nextVtx(\texttti{v}_i,\texttti{e})$ for $\texttti{v}_i, \texttti{v}_j \in V(G'), i\neq j$.
        \item Lastly, as always, we again encode an arbitrary linear order $\preceq_{next}$ on all the nodes of $G$ that starts in $\texttti{s}$ and ends in $\texttti{e}$.
    To that end, $next$ is added to $G'$ which indicates immediate predecessors and successors. 
    \end{itemize}

    Next, we construct a shapes graph $(\mathcal{C,T})$.
    $\mathcal{C}$ is as follows:    
    {
    \begin{align}
    \begin{split}
        \mathsf{col} \leftrightarrow {} & {}{=_1} col \land {} \bigg(\bigvee_{i=1}^5 \Big(D_i \land \big(\\
        &(\exists col.\texttti{r} \land {=_i} e.(\lnot \exists col.\texttti{r})) \\
        &{} \lor  {}(\exists col.\texttti{g} \land {=_i} e.(\lnot \exists col.\texttti{g})) \\
        &{} \lor  {}(\exists col.\texttti{b} \land {=_i} e.(\lnot \exists col.\texttti{b}))\big)\Big)\bigg) 
    \end{split}\label{cnst:prec:col} \\
    \begin{split}
        \mathsf{lin} \leftrightarrow {} & {}{=_1} nextVtx \land {}{=_1} nextVtx^-
    \end{split}\label{cnst:prec:lin} \\
    \begin{split}
        \mathsf{vtx} \leftrightarrow {} & Vtx \land \mathsf{col} \land \mathsf{lin} 
    \end{split}\label{cnst:prec:vtx} \\
    \begin{split}
        \mathsf{nxt} \leftrightarrow {} &  {=_1} nextCnt \land {}{=_1} nextCnt^- \\
        & \land \big((nextCnt^*\cdot nextCnt) \\
        & \quad \cup (nextCnt^{-*} \cdot nextCnt^-)\big) = neqCnt
    \end{split}\label{cnst:prec:nxt} \\
    \begin{split}
        \mathsf{mat} \leftrightarrow {} & {} {=_1} match \land {}{=_1} (match\cdot match^-)
    \end{split}\label{cnst:prec:mat} \\
    \begin{split}
        \mathsf{con} \leftrightarrow {} & match\cdot nextVtx \cdot match^- = nextCnt
    \end{split}\label{cnst:prec:con} \\
    \begin{split}
        \mathsf{all} \leftrightarrow {} & \big(({=_3} allowed.(\texttti{r} \lor \texttti{g} \lor \texttti{b})\\
        &\quad{} \land \exists(match \cdot nxtVtx^*\cdot col).\texttti{g}) \\
        & {} \lor ({=_2} allowed.(\texttti{r} \lor  \texttti{b}) \\
        &\quad {}\land \lnot \exists(match \cdot nextVtx^*\cdot col).\texttti{g})\big)
    \end{split}\label{cnst:prec:all} \\
    \begin{split}
        \mathsf{cnt} \leftrightarrow {} & Cnt \land \mathsf{nxt} \land\mathsf{mat} \land \mathsf{con}
    \end{split}\label{cnst:prec:cnt} \\
    \begin{split}
        \mathsf{s} \leftrightarrow {} & \texttti{s} \land \mathsf{mat} \land \mathsf{con}
    \end{split}\label{cnst:prec:s} \\
    \begin{split}
        \mathsf{e} \leftrightarrow {} & \texttti{e} \land \mathsf{mat}
    \end{split}\label{cnst:prec:e} \\
    \begin{split}
        \mathsf{r} \leftrightarrow {} & \texttti{r} \land Pre_{\{\texttti{r}\}} \land Pre_{\{\texttti{r},\texttti{g}\}} \land Pre_{\{\texttti{r},\texttti{b}\}}  \\
        & \land Pre_{\{\texttti{r},\texttti{g},\texttti{b}\}}\land {=_2} neqCol
    \end{split}\label{cnst:prec:r} \\
    \begin{split}
        \mathsf{g} \leftrightarrow {} &  \texttti{g} \land Pre_{\{\texttti{g}\}} \land Pre_{\{\texttti{r},\texttti{g}\}} \land Pre_{\{\texttti{g},\texttti{b}\}} \\
        & \land Pre_{\{\texttti{r},\texttti{g},\texttti{b}\}}\land {=_2} neqCol   
    \end{split}\label{cnst:prec:g} \\
    \begin{split}
        \mathsf{b} \leftrightarrow {} &  \texttti{b} \land Pre_{\{\texttti{b}\}} \land Pre_{\{\texttti{r},\texttti{b}\}} \land Pre_{\{\texttti{g},\texttti{b}\}} \\
        & \land Pre_{\{\texttti{r},\texttti{g},\texttti{b}\}}\land {=_2} neqCol  
    \end{split}\label{cnst:prec:b} \\
    \begin{split}
        \mathsf{val} \leftrightarrow {} & \forall next^*.(\mathsf{se} \lor \mathsf{vtx} \lor \mathsf{cnt} \lor \mathsf{r} \lor \mathsf{g} \lor \mathsf{b})\\
        & \land \exists next^*.\texttti{e}
    \end{split} \label{cnst:prec:val}
    \end{align}
    }
    The target is $\mathcal{T} = \{\mathsf{val}(\texttti{s})\}$, $\mu = \{\}$, and the query is $\pi_\emptyset Q$ with
    \begin{align*}
        Q = {}& \bigwedge_{(\texttti{v}_i,\texttti{v}_j)\in G'} neqCol(y_i,y_j)\\
        & {}\land Pre_{\{\texttti{g}\}}(y_1) \land \bigwedge_{\texttti{v}_i\in V(G')} Pre_{c(\texttti{v}_i)}(y_i)  \\
        & {}\land \bigwedge_{\texttti{v}_i\in V(G')} allowed(x_i,y_i) \land \bigwedge_{i\neq j} neqCnt (x_i,x_j).
    \end{align*}

    As before, the target $\mathsf{val}(\texttti{s})$ together with the Constraint \ref{cnst:prec:val} ensures that every repaired graph $G_R$ is such that all nodes satisfy specific constraints.
    Concretely, vertices $\texttti{v}\in V(G')$ have to satisfy Constraint \ref{cnst:prec:vtx} and thus also \ref{cnst:prec:col} and \ref{cnst:prec:lin}; counters $\texttti{i}\in I$ have to satisfy Constraint \ref{cnst:prec:cnt} and thus also Constraints \ref{cnst:prec:nxt} through \ref{cnst:prec:all}; the auxiliary nodes $\texttti{s}$ has to satisfy the Constraint \ref{cnst:prec:s} and thus also \ref{cnst:prec:mat} and \ref{cnst:prec:con}; the auxiliary nodes $\texttti{e}$ has to satisfy the Constraint \ref{cnst:prec:e} and thus also \ref{cnst:prec:mat}; and the colors $\texttti{r}, \texttti{g}, \texttti{b}$ respectively satisfy the Constraints $\ref{cnst:prec:r}, \ref{cnst:prec:g}, \ref{cnst:prec:b}$.

    First, let us consider vertices.
    Constraint \ref{cnst:prec:col} ensures that every repair ``selects'' a color for each vertex and that the resulting coloring is a valid 3-coloring adhering to $c$.
    Furthermore, Constraint \ref{cnst:prec:lin} ensures that every vertex appears exactly twice in the property $nextVtx$, once in the first position and once in the second position.

    Next, let us consider counters and the auxiliary nodes $\texttti{s}, \texttti{e}$.
    Firstly, Constraint \ref{cnst:prec:nxt} ensures that the properties $nextCnt$ and $neqCnt$ are left unchanged.
    Then, Constraint \ref{cnst:prec:mat} ensures that counters $\texttti{i}_i\in I$ are matched 1-1 to vertices $\texttti{v}(\texttti{i}_i)\in V(G')$.
    Furthermore, this constraint also has to be satisfied by $\texttti{s}, \texttti{e}$, ensuring that they are match to themselves.
    Due to Constraint \ref{cnst:prec:con}, this matching has to consistent with the property $nextVtx$.
    This means that $nextVtx(\texttti{v}(\texttti{i}_i), \texttti{v}(\texttti{i}_{i+1}))\in G_R$ for $i=1,\dots,n-1$ as well as $nextVtx(\texttti{s}, \texttti{v}(\texttti{i}_{1}))\in G_R$ and $nextVtx(\texttti{v}(\texttti{i}_{n}), \texttti{e})\in G_R$.
    Thus, crucially, the property $nextVtx$ represents the immediate predecessors and successors relation of a linear order $\preceq_{nextVtx}$ on all the vertices $V(G')$ that additionally starts in $\texttti{s}$ and ends in $\texttti{e}$.
    Furthermore, by correctly ``selecting'' $nextVtx$ and $match$, any such linear order can be encoded in $G_R$.

    Now, consider Constraint \ref{cnst:prec:all} which all counters also have to satisfy.
    This constraint fixes the property $allowed$, which, thus, deterministically depends on the linear order $\preceq_{nextVtx}$ represented by the $G_R$.
    Concretely, $allowed(\texttti{i}_i, \texttti{r}), allowed(\texttti{i}_i, \texttti{b})$ are always in $G_R$ but $allowed(\texttti{i}_i, \texttti{g})$ is in $G_R$ if and only if there is a $j\geq i$ such that $\texttti{v}(\texttti{i}_j)$ is colored $\texttti{g}$, i.e., $col(\texttti{v}(\texttti{i}_j), \texttti{g})\in G_R$.

    Lastly, the constraints for the colors simply ensure that the property $neqCol$ and the classes $Pre_C$ are left unchanged.
    
    Thus, summarizing, the idea is the following: each repaired graph $G_R$ represents a valid 3-coloring of the graph adhering to the precoloring together with a specific linear order $\preceq_{nextVtx}$ on the vertices. 
    Then, there is a bijection between counters $\texttti{i}_i$ and vertices $\texttti{v}(\texttti{i}_i)$ with the property: $\texttti{v}(\texttti{i}_i) \preceq_{makeLin} \texttti{v}(\texttti{i}_j)$ if and only if $i\leq j$.
    Furthermore, the repaired graph allows the color $\texttti{g}$ to the counter $\texttti{i}_i$ if and only if there is a vertex $\texttti{v}(\texttti{i}_j)$ colored $\texttti{g}$ that comes after the vertex $\texttti{v}(\texttti{i}_i)$, i.e., $j\geq i$.
    Thus, the $allowed$ relation ``(over-) counts'' the uses of $\texttti{g}$.

    However, if the linear order is such the vertices colored $\texttti{g}$ come first, the repaired graph correctly counts the uses of $\texttti{g}$.
    Furthermore, for any two repaired graphs $G_R$ and $G_{R'}$, we have $\llbracket allowed\rrbracket ^{G_R} \subseteq \llbracket allowed\rrbracket ^{G_{R'}}$ or the other way around.
    Thus, there is a unique minimum $allowed$ relation which also appears in the intersection of all repaired graphs $G_\cap$.
    Concretely, if $c$ is the minimum number of uses of $\texttti{g}$ required, $allowed(\texttti{i}_i,\texttti{g})\in G_\cap$ if and only if $i\leq c$, i.e., $G_\cap$ correctly counts the minimal uses of $\texttti{g}.$
    Furthermore, as no repaired graph alters the classes $Pre_C$ and the properties $neqCol$, $allowed$, and $neqCnt$, they are the same in $G_\cap$ as they were in the construction.

    Now, consider the query.
    The query asks whether there is a 3-coloring of the graph that adheres to the precoloring, colors $\texttti{v}_1$ the color $\texttti{g}$, and only uses $\texttti{g}$ as little as possible.
    This is done by ``guessing'' a counter $\texttti{i}_j$ (variable $(x_i)$) and a color $\texttti{c}$ (variable $(y_i)$) for each vertex $\texttti{v}_i$.
    A vertex can be colored $\texttti{g}$ if and only if its counter $\texttti{i}_j$ is less than or equal to $c$, i.e., $j\leq c$ (facilitated by $allowed(x_i,y_i)$), assuring that $\texttti{g}$ is used at most $c$ times.
    Thus, values for the bound variables $x_i$ and $y_i$ exist exactly when there is a coloring that colors $\texttti{v}_1$ the color $\texttti{g}$ while still minimizing the use of $\texttti{g}$, completing the proof of correctness.
    
    {\small$\blacktriangleright$} (Membership for \bgpProjEquIn data complexity)
    The co-problem is to show that $\mu$ is no answer to $\pi_X Q(X,Y)$ over the intersection $G_\cap$ of all repaired graphs.
    This is the case if for all $\mu_Y \colon Y\rightarrow V(G)$, at least one atom $\alpha$ of $Q$ is not sent to $G_\cap$ by $\mu\cup \mu_Y$.
    Thus, there must also be a repair $R(\mu_Y)$ such that $\alpha$ is not sent to $G_{R(\mu_Y)}$ by $\mu\cup \mu_Y$.
    Hence, $R(\mu_Y)$ can be seen as a witness that $\mu\cup \mu_Y$ is no answer to $Q$ over $G_\cap$.
    Since there are only a polynomial number of $\mu_Y$, a NTM can guess a $R(\mu_Y)$ together with a supporting model for each $\mu_Y$.
    Thus, the co-problem is in \np and \bgpProjEquIn in \conp.
    
    {\small$\blacktriangleright$} (Membership for \bgpEquIn combined complexity)
    Looking at the previous argumentation, notice that the size of $Q$ being fixed was not used, but the fact that the number of $Y$ variables was fixed.
    Of course, the same holds for \bgp\ in combined complexity and thus, also \bgpEquIn is in \conp.

    {\small$\blacktriangleright$} (Membership for \bgpProjEquIn combined complexity)
    \bgpProjEquIn is in \thetatwo as we can compute the size $K$ of the intersection of the repaired graphs with logarithmically many calls to an \np-oracle.
    We can do this by making oracle calls of the form: ``is the size of intersection of all repaired graphs less than $c$?''
    This question can be answered by an NTM following the subsequent procedure:
    \begin{enumerate}
        \item Let $S:=|G|+|H|$.
        \item Guess repairs $R_1,\dots,R_{S-c+1}$.
        \item Check that $|\bigcap_i G_{R_i}| < c$.
    \end{enumerate}
    Intuitively, $S$ is the number of potential atoms in the intersection, and each $R_i$ witnesses that at least 1 potential atom does not appear in the intersection.
    Of course, if the size of the intersection is less than $c$, at least $S-c+1$ such repairs must exist (some may appear duplicately).
    
    Then, given $K$, we can answer the query with a single \np call.
    For this, we proceed as follows (the query is $\pi_X Q(X,Y)$):
    \begin{enumerate}
        \item Let $S:=|G|+|H|$.
        \item Guess repairs $R_1,\dots,R_{S-K+1}$, and a $\mu_Y\colon Y\rightarrow V(G)$.
        \item Check that $|\bigcap_i G_{R_i}| = K$.
        \item Check that $\mu \cup \mu_Y$ is an answer to $Q$ over $\bigcap_i G_{R_i}$.
    \end{enumerate}
    
    {\small$\blacktriangleright$} (Membership for \wdptProjCardIn data complexity)
    We start by showing that we can compute the size of the intersection of the $\leq$-repaired graphs in \thetatwo.
    To do this, we first have to compute the minimal size $k$ of a repair with logarithmically many \np-oracle calls (as we have already done many times).
    Then, using $k$, we make oracle calls of the form: ``is the size of intersection of all $\leq$-repaired graphs less than $c$?''
    Given $k$, this question can be answered by an NTM similar to the previous case but where we additionally have to check that $|R_i| = k$.    
    Consequently, we can again use binary search to compute the exact size $K$ of the intersection with only logarithmically many \np-oracle calls.
    Finally, given $k$ and $K$, the question of whether $\mu$ is an answer to $Q$ becomes an \np-problem, and, thus, only requires a single additional \np-oracle call.
    
    The reason for that is that given $k$ and $K$, a NTM can guess repairs $R_1,\dots,R_{S-K+1}$ and be sure that $G_\cap = \bigcap_i G_{R_i}$ if $|\bigcap_i G_{R_i}| = K$.
    The remainder is then only a matter of going through all possible (polynomially many) extensions of $\mu$ to the variables projected away and, then, variables in the OPT-part.
    This proves the \thetatwo-membership.
    
    {\small$\blacktriangleright$} (Membership for \bgpProjCardIn and \wdptCardIn combined complexity)
    Using the same method as before, we can determine $k$ and $K$.
    Then, given $k$ and $K$, answering a \bgpProj\ only requires \np-power, also in combined complexity.
    To accomplish this, we simply have to guess the values of the variables projected away at the same time we guess the repairs $R_1,\dots,R_{S-K+1}$.
    
    For a \wdpt $Q=(( \dots ((P\OPT P_1)\OPT P_2) \dots) \OPT P_k)$, the argument is slightly different as, given $k,$ $K$, answering the query requires \conp-power in combined complexity.
    I.e., given $k$ and $K$, to identify no-instances we can guess a $P_i$ and a $\mu_{P_i}\colon \ivar(P_i) \setminus \ivar(P) \rightarrow V(G)$ together with $\leq$-repairs $R_1,\dots,R_{S-K+1}$ such that $|\bigcap_iG_{R_i}| = K$.
    Then, if $(\mu\cup \mu_{P_i})_{\ivar(P_i)}$ is an answer to $P_i$ over $\bigcap_iG_{R_i}$ or $\mu$ is no answer to $P$ over $\bigcap_iG_{R_i}$, we can assert that $\mu$ is also no answer to $Q$ over the intersection of all repairs.
        
    {\small$\blacktriangleright$} (Membership for \bgpProjSubsIn combined complexity)
    We consider the co-problem, i.e., whether $\mu$ is no answer of $Q$ over the intersection of all $\subseteq$-repaired graphs.
    For this, we proceed as follows:
    \begin{enumerate}
        \item Guess (a superset of the intersection of all $\subseteq$-repaired graphs) $G_\cap = A_\cap \cup (G\setminus D_\cap)$ together with repairs $(A_a,D_a)$ for each $a\in H\setminus A_\cap$ and $(A_d,D_d)$ for each $d\in D_\cap$.
        \item Check that $\mu$ is no solution of $Q$ over $G_\cap$.
        \item Check that $a\not\in A_a$ and that $(A_a,D_a)$ is a $\subseteq$-repair.
        \item Check that $d\in D_d$ and that $(A_d,D_d)$ is a $\subseteq$-repair.
    \end{enumerate}
    All of these checks only require access to an \np-oracle.
    Furthermore, if all of these checks are positive, $G_\cap$ is indeed a superset of the intersection of all $\subseteq$-repairs, and, since $\pi$-BGP are monotone, $\mu$ is no answer of $Q$ over the intersection.
    Hence, we can correctly assert that we are dealing with a no-instance.
    On the other hand, if for all guesses at least one check is negative, there also has to be a check that is negative when $G_\cap$ is the intersection of all $\subseteq$-repairs and $(A_a, D_a), (A_d, D_d)$ are $\subseteq$-repairs with $a\not\in A_a$ and $d\in D_d$, respectively (The existence of $(A_a, D_a), (A_d, D_d)$ is guaranteed by $G_\cap$ being the intersection).
    Thus, the first check has to be negative.
    We can, therefore, correctly conclude that we are dealing with a yes-instance.
    Thus, in total, the co-problem is in $\sigmatwo$ and \bgpProjSubsIn is in \pitwo.
\end{proof}

\begin{theoremAppendix}
The following statements are true for data complexity:
\begin{itemize}
    \item \wdptEquIn is \DP-c.
    \item \wdptSubsIn is \DPtwo-c.
\end{itemize}
The following statements are true for combined complexity:
\begin{itemize}
    \item \wdptEquIn is \thetatwo-c.
    \item \wdptSubsIn is \DPtwo-c.
\end{itemize}
\end{theoremAppendix}

\begin{proof}
    We show the following for data and combined complexity, and ${\preceq}\in \{=,\subseteq\}$ (it also holds for $\leq$): If \bgpPrecIn is $\mathsf{C}$-complete and \bgpProjPrecIn is $\mathsf{D}$-complete, then \wdptPrecIn is complete for $\mathsf{C} \land \mathsf{coD} = \{C\cap D \mid C\in \mathsf{C} \text{ and } D\in \mathsf{co}\mathsf{D}\}$.
    Note that this implies the statements in the theorem as $\DP = \conp \land \np$, $\DPtwo = \pitwo \land \sigmatwo$, and $\thetatwo = \conp \land \thetatwo$.
    
    {\small$\blacktriangleright$} (Membership)
    we may assume that a \wdpt $Q$ is of the form 
    $(( \dots ((P\OPT P_1)\OPT P_2) \dots)\OPT P_k) $
    with $\ivar(P) = \idom(\mu)$.    
    Then $\mu$ is an answer to $Q$ if and only if $\mu$ is an answer to $P$ 
    but, for all $i$, $\mu$ \textit{cannot} be extended to an answer of $P_i$.
    Checking if $\mu$ is an answer to $P$ corresponds to identifying yes-instances of \bgpEquIn; 
    checking that $\mu$ cannot be extended to an answer of any $P_i$ corresponds to 
    identifying no-instances of \PBEquIn (note that the variables in $\ivar(P_i) \setminus \ivar(P)$
    behave like bound variables in this case, since we are not interested in a particular 
    extension of $\mu$ to these variables but to {\em any} extension).
    The problem \wdptPrecIn can, therefore, be seen as the intersection of \bgpPrecIn and (multiple) \bgpProjPrecIn.
    Notice that identifying multiple no-instances of \PBEquIn does not increase the complexity.

    {\small$\blacktriangleright$} (Hardness)
    For the hardness part, we reduce from the complete problem \big( \bgpPrecIn, co-\bgpProjPrecIn\big), whose instances are pairs $(\mathcal{I},\mathcal{J})$, where $\mathcal{I}$ and $\mathcal{J}$ are instances of \bgpPrecIn and \bgpProjPrecIn, respectively. 
    That is, they have the form 
    $\mathcal{I} = (G_\mathcal{I}, \mathcal{C}_\mathcal{I}, \mathcal{T}_\mathcal{I}, H_\mathcal{I}, 
    Q_\mathcal{I}, \mu_\mathcal{I})$ and 
    $\mathcal{J} = (G_\mathcal{J}, \mathcal{C}_\mathcal{J}, \mathcal{T}_\mathcal{J}, H_\mathcal{J}, 
    Q_\mathcal{J}, \mu_\mathcal{J})$, respectively, with
    $Q_\mathcal{I}(X) = P_1(X)$ and $Q_\mathcal{J}(Y) = \pi_Y P_2(Y,Z)$.
    
    A pair $(\mathcal{I},\mathcal{J})$ is a yes-instance instance if and only if both $\mathcal{I}$ and $\mathcal{J}$ are yes-instances, that is, if $\mu_\mathcal{I}$ is an answer to $Q_\mathcal{I}$ over the intersection of all $\preceq$-repaired graphs of $G_\mathcal{I}$ w.r.t.\ $(( \mathcal{C}_\mathcal{I}, \mathcal{T}_\mathcal{I}), H_\mathcal{I})$ and     $\mu_\mathcal{J}$ is {\em not} an answer to $Q_\mathcal{J}$ over the intersection of all $\preceq$-repaired graphs of $G_\mathcal{J}$ w.r.t.\ $(( \mathcal{C}_\mathcal{J}, \mathcal{T}_\mathcal{J}), H_\mathcal{J})$.

    We may assume that the instances $\mathcal{I}$ and $\mathcal{J}$ share no symbols; in other words, they use distinct variables, constants, classes, properties, \dots. 
    Then, to construct an instance of \wdptPrecIn, we can combine the shapes graphs, the data graphs, the hypotheses, and the variable assignments $\mu = \mu_\mathcal{I} \cup \mu_\mathcal{J}$.
    Additionally, we add a new class $V_\mathcal{J}$ that consists of all nodes $V(G_\mathcal{J})\cup V(H_\mathcal{J})$.
    Furthermore, we add auxiliary nodes $\texttti{s}, \texttti{e}$ and an arbitrary linear order on $V(G_\mathcal{J})\cup V(H_\mathcal{J})$, where $next_{\mathcal{J}}$ indicates immediate predecessors and successors, and which starts in $\texttti{s}$ and ends in $\texttti{e}$.
    Let $G^*$ be these new atoms.
    Then, we also add the constraint 
    \begin{align*}
        \mathsf{val} \leftrightarrow \forall next_\mathcal{J}^*.(\texttti{s} \lor \texttti{e} \lor V_\mathcal{J}) \land \exists next_\mathcal{J}^*.\texttti{e}
    \end{align*}
    together with the target $\mathsf{val}(\texttti{s})$.

    Moreover, we define the query of the combined instance as $Q(X, Y, Z) = P_1(X) \land \bigwedge_{y\in Y} V_{\mathcal{J}}(y) \OPT P_2(Y, Z)$.

    Let $(A_i,D_i)_i, (A_j,D_j)_j$ be the $\preceq$-repairs of $\mathcal{I}$ and $\mathcal{J}$, respectively.
    Then the $\preceq$-repairs of the constructed instance are exactly $(A_i \cup A_j, D_i \cup D_j)_{i,j}$ due to the instances sharing no symbols.
    Furthermore, due to disjointedness, we have the intersection of the $\preceq$-repaired graphs
    \begin{align*}
        G_\cap & = \bigcap_{i,j} (A_i \cup A_j) \cup \bigcap_{i,j}((G_\mathcal{I} \cup  G_\mathcal{J}) \setminus (D_i \cup D_j)) \cup G^*\\
        & = \bigcap_{i} A_i \cup \bigcap_{i}(G_\mathcal{I} \setminus D_i) \cup \bigcap_{j} A_j \cup \bigcap_{j}(G_\mathcal{J} \setminus D_j) \cup G^* \\
        & = G_{\mathcal{I},\cap} \cup G_{\mathcal{J},\cap} \cup G^*,
    \end{align*}
    where $G_{\mathcal{I},\cap}, G_{\mathcal{J},\cap}$ are the intersections of the $\preceq$-repaired graphs of $\mathcal{I}$ and $\mathcal{J}$, respectively.
    Thus, $\mu$ is an answer to $Q$ over the intersection $G_\cap$ if and only if $\mu_\mathcal{I}$ is an answer to $Q_\mathcal{I}$ over $G_{\mathcal{I},\cap}$ and $\mu_\mathcal{J}$ is not an answer to $Q_\mathcal{J}$ over $G_{\mathcal{J},\cap}$.
\end{proof}

\begin{theoremAppendix}
\label{thmApp:wdptProj}
The following statements are true for data complexity:
\begin{itemize}
    \item \wdptProjEquIn is \thetatwo-c.
    \item \wdptProjSubsIn is \thetathree-c.
\end{itemize}
The following statements are true for combined complexity:
\begin{itemize}
    \item \wdptProjEquIn and \wdptProjCardIn are \sigmatwo-c.
    \item \wdptProjSubsIn is \thetathree-c.
\end{itemize}
\end{theoremAppendix}

\begin{proof}

{\small$\blacktriangleright$} (Hardness for \wdptProjEquIn data complexity)
    We reduce from the \thetatwo-complete problem \textsc{List-Pair-SAT} (for hardness see Lemma \ref{lem:listPair} in Section~\ref{sec:reference}).
    The instances consist of a list $(\phi_1, \psi_1),\dots, (\phi_n, \psi_n)$ of pairs of propositional formula in 3-CNF, i.e., $\phi_i = \bigwedge^{m(\phi_i)}_{j=1} (\texttti{l}_{1j}(\phi_i) \lor \texttti{l}_{2j}(\phi_i) \lor \texttti{l}_{3j}(\phi_i))$ and $\psi_i = \bigwedge^{m(\psi_i)}_{j=1} (\texttti{l}_{1j}(\psi_i) \lor \texttti{l}_{2j}(\psi_i) \lor \texttti{l}_{3j}(\psi_i))$.
    $(\phi_1, \psi_1),\dots, (\phi_n, \psi_n)$ is a yes instance if some pair $(\phi_i, \psi_i)$ is such that $\phi_i$ is unsatisfiable while $\psi_i$ is satisfiable.
    W.l.o.g, we assume the variables of all formulas to be pairwise disjoint.
    
    The encoding is similar to the one used in Theorem \ref{thm:app:first}.
    The nodes of the data graph $G$ consist of:
    \begin{itemize}
        \item All literals $\texttti{l}$.
        \item Clauses $\texttti{C}_{j}(\phi_i)= (\texttti{l}_{1j}(\phi_i) \lor \texttti{l}_{2j}(\phi_i) \lor \texttti{l}_{3j}(\phi_i))$ and $\texttti{C}_{j}(\psi_i)= (\texttti{l}_{1j}(\psi_i) \lor \texttti{l}_{2j}(\psi_i) \lor \texttti{l}_{3j}(\psi_i))$.
        \item Formulas $\phi_i,\psi_i$.
        \item Counters $\texttti{i}_i$ for $i=1,\dots,n$.
        \item Auxiliary elements $\texttti{s}, \texttti{e}$ to $G$.
    \end{itemize}
    
    Literals $\texttti{l}$ are added to the class $Lit$, clauses $\texttti{C}$ to $Cl$, formulas $\phi_i,\psi_i$ to $Phi$, and counters $\texttti{i}$ to $Cnt$.

    We now connect the nodes as they are connected logically.
    \begin{itemize}
        \item Dual literals $\texttti{l}$ and $\texttti{l}_d$ are connected via $dual(\texttti{l}, \texttti{l}_d)$.
        \item Literals $\texttti{l}_{kj}(\phi_i)$ (resp. $\texttti{l}_{kj}(\psi_i)$) are connected to clauses $\texttti{C}_{j}(\phi_i)$ (resp. $\texttti{C}_{j}(\psi_i)$) via $or(\texttti{l}_{kj}(\phi_i), \texttti{C}_{j}(\phi_i))$ (resp. $or(\texttti{l}_{kj}(\phi_i), \texttti{C}_{j}(\psi_i))$).
        \item Each clause $\texttti{C}_{j}(\phi_i)$ (resp. $\texttti{C}_{j}(\psi_i)$) is connected $\phi_{i}$ (resp. $\psi_i$) via $and(\texttti{C}_{j}(\phi_i), \phi_{i})$ (resp. $and(\texttti{C}_{j}(\psi_i), \psi_{i})$).
        \item Counters $\texttti{i}_i$ are connected to the formulas $\phi_{i}$ and $\psi_{i}$ via $is_{\phi}(\texttti{i}_i, \phi_{i})$ and $is_{\psi}(\texttti{i}_i, \psi_{i})$, respectively.
    \end{itemize}

    Each literal $\texttti{l}$, clause $\texttti{C}_{ik}$, and formula $\phi_{i},\psi_{i}$ is added to both the class $T$ and $F$, indicating that they can be either \textit{true} or \textit{false}.
    All nodes of $G$ are put into a single arbitrary linear order $\preceq_{next}$, encoded by the property $next$, indicating immediate predecessors and successors, and which starts in $\texttti{s}$ and ends in $\texttti{e}$.

    We use the following constraints $\mathcal{C}$:
    \begin{align*}
        \mathsf{lit} \leftrightarrow {} & Lit \land ((T \land \lnot F \land \exists dual.F) \\
         & \quad {} \lor (F  \land \lnot T \land \exists dual.T))  \\
        \mathsf{cl} \leftrightarrow {} & Cl \land {=_3}or^- \land {=_1}and \land (\\
        & (F \land \lnot T \land \forall or^-.F) \\
        & \lor (T \land \lnot F \land \exists or^-.T)  \\
        \mathsf{phi} \leftrightarrow {} & Phi \land ((F \land \lnot T \land \exists and^-.F) \\
        & \quad {} \lor (T \land \lnot F \land \forall and^-.T)  \\
        \mathsf{cnt} \leftrightarrow {} & Cnt \land {=_1}is_{\phi} \land {=_1}is_{\psi}\\
        \mathsf{val} \leftrightarrow {} & \forall next^*.(\texttti{s}\lor \texttti{e} \lor \mathsf{lit}\lor \mathsf{cl}\lor \mathsf{phi}\lor \mathsf{cnt}) \\
        & {}\land \exists next^*.\texttti{e}
    \end{align*}
    The target is $\{\mathsf{val}(\texttti{s})\} = \mathcal{T}$.
    $\mu$ is $\{\}$ and the query is  $\pi_z$ with
        \begin{align*}
            Q = is_{\phi}(x,y)\land  F(y) \OPT is_{\psi}(x,z)\land  F(z).
        \end{align*}

    For each formula $\phi_i$ and $\psi_i$, the repairs have the (independent) choice of what truth values to assign to each literal $\texttti{x}$ by retaining either $T(\texttti{x})$ and not $F(\texttti{x})$ or the other way around.
    The truth values of the dual literals, the clauses, and the formula itself are then all functionally determined by this choice.
    Thus, $F(\phi_i)$ (resp. $F(\psi_i)$) is in the repaired graph if and only if $\phi_i$ (resp. $\psi_i$) evaluates to false under the given choice of the truth values for the variables that appear in $\phi_i$ (resp. $\psi_i$).
    Since each truth value assignment constitutes a repair, $F(\phi_i)$ (resp. $F(\psi_i)$) appears in the intersection of all repaired graphs if and only if $\phi_i$ (resp. $\psi_i$) is unsatisfiable.
    Hence, $\mu = \{\}$ is an answer to $Q$ if and only if there is an $\texttti{i}_i$ (variable $x$) such that $\phi_i$ (variable $y$) is unsatisfiable but $\psi_i$ (variable $z$) is not unsatisfiable, i.e., $\psi_i$ is satisfiable.
    This completes the reduction.

    {\small$\blacktriangleright$} (Hardness for \wdptProjSubsIn data complexity)
    To show \thetathree-hardness we have to extend the reduction used in the \wdptProjEquIn case, proceeding further up the polynomial hierarchy.
    We reduce from the following \thetathree-complete problem \textsc{List-Pair-}$2$\textsc{-QBF} (for hardness see Lemma \ref{lem:listPair} in Section~\ref{sec:reference}).
    The instances consist of a list $(\phi_1(X(\phi_1),Y(\phi_1))$, $\psi_1(X(\psi_1),Y(\psi_1)))$, $\dots$, $(\phi_n(X(\phi_n),Y(\phi_n))$, $\psi_n(X(\psi_n),Y(\psi_n)))$ of pairs of propositional formula in 3-CNF, i.e., $\phi_i = \bigwedge^{m(\phi_i)}_{j=1} (\texttti{l}_{1j}(\phi_i) \lor \texttti{l}_{2j}(\phi_i) \lor \texttti{l}_{3j}(\phi_i))$ and $\psi_i = \bigwedge^{m(\psi_i)}_{j=1} (\texttti{l}_{1j}(\psi_i) \lor \texttti{l}_{2j}(\psi_i) \lor \texttti{l}_{3j}(\psi_i))$.
    An instance is a yes instance if some pair $(\phi_i(X(\phi_i),Y(\phi_i)), \psi_i(X(\psi_i),Y(\psi_i)))$ is such that every truth assignment of the $X(\phi_i)$ can be extended to a model of $\phi_i$ while this is not the case for the $X(\psi_i)$ and $\psi_i$.
    W.l.o.g, we assume the variables of all formulas to be pairwise disjoint.
    
    For the encoding, we take the data graph $G$ from before, excluding the class $Lit$, and extend it further:
    \begin{itemize}
        \item for each literals $\texttti{l}$ of a $X$ (resp. $Y$) variable we add $Lit_X(\texttti{l})$ (resp. $Lit_Y(\texttti{l})$).
        \item Each literals $\texttti{l}$ of $Y(\phi_i)$ (resp. $Y(\psi_i)$) and clause $\texttti{C}_j(\phi_i)$ (resp. $\texttti{C}_j(\psi_i)$) is connected to $\phi_i$ (resp. $\psi_i$) via $isIn(\texttti{l}, \phi_i)$ and $isIn(\texttti{C}_j(\phi_i), \phi_i)$ (resp. $isIn(\texttti{l}, \psi_i)$ and $isIn(\texttti{C}_j(\psi_i), \psi_i)$).
        \item For each $\phi_i,\psi_i$ we add $Ext(\phi_i), Ext(\psi_i)$ to $G$ which shall indicate whether the corresponding $Y$ variables are instantiated.
    \end{itemize}

    We use the following constraints $\mathcal{C}$:
    \begin{align*}
        \mathsf{lit}_X \leftrightarrow {} & Lit_X \land ((T \land \lnot F \land \exists dual.F) \\
         & \quad {}\lor (F  \land \lnot T \land \exists dual.T))  \\
         \mathsf{lit}_Y \leftrightarrow {} & Lit_Y \land \exists isIn \land ( \\
         & (\exists isIn.Ext \land ( \\
         & \quad(T \land \lnot F \land \exists dual.F) \\
         & \quad {}\lor (F  \land \lnot T \land \exists dual.T))  \\
         & \lor ( \lnot \exists isIn.Ext \land \lnot F  \land \lnot T))  \\
        \mathsf{cl} \leftrightarrow {} & Cl \land {=_3}or^- \land {=_1}and \land \exists isIn \land (\\
        & (\exists isIn.Ext \land ( \\
        & \quad (F \land \lnot T \land \forall or^-.F) \\
        & \quad {} \lor (T \land \lnot F \land \exists or^-.T))  \\
         & \lor ( \lnot \exists isIn.Ext \land \lnot F  \land \lnot T))  \\
        \mathsf{phi} \leftrightarrow {} & Phi \land ((Ext\land T \land \lnot F \land \forall and^-.T)  \\
        & \quad {} \lor (\lnot Ext\land \lnot F \land \lnot T))  \\
        \mathsf{cnt} \leftrightarrow {} &Cnt \land {=_1}is_{\phi} \land {=_1}is_{\psi}\\
        \mathsf{val} \leftrightarrow {} & \forall next^*.(\texttti{s}\lor \texttti{e} \lor \mathsf{lit}_X \lor \mathsf{lit}_Y\lor \mathsf{cl}\lor \mathsf{phi}\lor \mathsf{cnt}) \\
        & {}\land \exists next^*.\texttti{e}
    \end{align*}
    We leave the target and $\mu$ unchanged, while the query now is $\pi_zQ$ with
    \begin{align*}
        Q(z) = is_\phi(x,y) \land Ext(y) \OPT is_\psi(x,z)\land  Ext(z).
    \end{align*}
    
    The idea is as follows.
    For each formula, the repairs have the choice of whether to instantiate only the $X$ variables or both the $X$ and the $Y$ variables.
    In the first case, the formula has to be removed from the class $Ext$ and the $X$ variables instantiated, but, at the same time, all literals of the $Y$ as well as the clauses and the formula itself have to be removed from both $T$ and $F$.
    In the second case, the formula remains in the class $Ext$, but both the $X$ and the $Y$ variables have to be instantiated, and the truth values of the dual literals, the clauses, and the formula itself are then all functionally determined by this choice.
    Furthermore, the formula itself has to evaluate to true.

    Thus, intuitively, when concentrating on a single formula, repairs correspond 1-1 to instantiations of only the $X$ variables or a combination the $X$ and the $Y$ that satisfy the formula.
    When the instantiation is the same on the $X$ variables, a repair of the second kind is a subset of the repair of the first kind.
    Thus, a formula $\phi_i$ (resp. $\psi_i$) is not in the class $Ext$ in some $\subseteq$-repaired graph $G_R$ if and only if the instantiation of the $X$ variables represented by $G_R$ cannot be extended to the $Y$ variables while satisfying the $\phi_i$ (resp. $\psi_i$). 

    Concluding, $\mu = \{\}$ is an answer to $Q$ if and only if there is an $\texttti{i}_i$ (variable $x$) such that for $\phi_i$ (variable $y$), every instantiation of the variables $X(\phi_i)$ can be extended to a model of $\phi_i$, while this is not true for $\psi_i$ (variable $z$) and the variables $X(\psi_i)$.

    {\small$\blacktriangleright$} (Hardness for \wdptProjEquIn and \wdptProjCardIn combined complexity)
    Recall that answering \wdptProj\ queries over fixed data graphs $G$ is already \sigmatwo-complete combined complexity \cite{DBLP:journals/tods/Letelier0PS13}.
    Thus, we can explicitly name every atom from $G$ in the shapes graph such that the only repaired graph is $G$ itself.

    {\small$\blacktriangleright$} (Membership for \wdptProjEquIn and \wdptProjCardIn combined complexity)
    To check whether $\mu\colon X \rightarrow V(G)$ is an answer to $\pi_X Q$ with $ Q =(( \dots ((P\OPT P_1)\OPT P_2) \dots)\OPT P_k) $, we can do the following:
    \begin{enumerate}
        \item Let $S:=|G|+|H|$.
        \item Compute the size $K$ of the intersection of all repaired graphs.
        \item Guess repairs $R_1,\dots,R_{S-K+1}$ and a $\mu_{P}\colon \ivar(P)\setminus X\rightarrow V(G)$.
        \item Check that $R_i$ is valid (and no smaller repair exists).
        \item Check that $|\bigcap_i G_{R_i}| = K$.
        \item Check that $\mu \cup \mu_P$ is an answer to $Q$ over $\bigcap_i G_{R_i}$.
    \end{enumerate}
    Clearly, all checks only require access to an \np-oracle and we have already seen how to compute $K$ deterministically with access to an \np-oracle.

    {\small$\blacktriangleright$} (Membership for \wdptProjEquIn data complexity)
    We can proceed similarly to the case \wdptProjCardIn, but we skip the computation of $k$.
    That is, we simply determine the size $K$ of the intersection of all repaired graphs.
    Then, we ask the question of whether $\mu$ is an answer to the query $\pi_XQ$.
    Given $K$, this only requires \np-power as a NTM can guess the intersection of all repaired graphs $G_\cap$ (size $K$), and a witnessing repair (of size $k$) together with a model supporting it for each atom not in the intersection.
    Then, the NTW can go through all (polynomially many) instantiations of the bound variables and, if needed, also through all (polynomially many) extensions to the remaining variables in the OPT-part, always checking whether this variables assignment satisfies the corresponding query over $G_\cap$.
    Thus, the NTW can return yes if $\mu$ is an answer.

    {\small$\blacktriangleright$} (Membership for \wdptProjSubsIn combined complexity)
    The argument for \thetathree-membership combined complexity of \wdptProjSubsIn works similar to the case of \wdptProjEquIn data complexity.
    However, we have to check that the repairs are minimal and checking that $\mu$ is, in fact, an answer to the query $\pi_XQ$ requires more power.

    Concretely, we start by determining the size of the intersection of all $\subseteq$-repaired graphs with logarithmically many \sigmatwo-oracle calls.
    The calls are of the form: ``Is size of the intersection of all $\subseteq$-repaired graphs less than $c$?''
    This is a \sigmatwo-question as we can guess the atoms not in the intersection and a repair per atom as a witness.
    With access to an \np-oracle, we then have to check that there are no smaller repairs.

    Given the size $K$ of the intersection, the question of whether $\mu$ is an answer to $\pi_XQ$ over the intersection only requires a single further \sigmatwo-oracle call.
    A NTM with access to an \np-oracle simply has to guess the values $\mu_X$ for the variables $X$, and the intersection $G_\cap$ together with a witnessing repair per atom not in $G_\cap$,
    Then, it can check that each repair is minimal and whether $\mu\cup \mu_X$ is an answer to the query over the intersection.
    For the last check, a single \np-oracle call suffices.
\end{proof}

\section{Full Proofs for CQA for Maximal Repairs}

First, we recall the Boolean hierarchy \bh, which is less common than the other complexity classes encountered so far: on the bottom levels, we have \bhone = \np and  \bhtwo = the class of languages obtained as the intersection of a language in \bhone and a language in \conp. 
In other words, we have $\bhtwo = \np \land \conp =\DP$. 
For the general case, we have $\bhtwok = \bigvee_{i=0}^k \DP $, i.e., the class of languages that are the union of $k$ \DP languages, and $\bhtwokplus = \np \lor \bigvee_{i=0}^k \DP$, i.e., the class of languages that are the union of $k$ \DP languages and an \np language.
In order to establish membership of a problem in \bh, it suffices to show that the problem can be solved in deterministic polynomial time with constantly many calls to an \np-oracle. 
Likewise, to show that a problem is hard for every level of \bh, it suffices to exhibit a reduction from the combinations of $k$ pairs of instances of an \np-hard problem for arbitrary $k > 0$ asking if there is a pair $(\mathcal{I}_i,\mathcal{J}_i)$ such that $\mathcal{I}_i$ is a yes-instance and $\mathcal{J}_i$ is a no-instance.

\begin{theoremAppendix}
\label{thmApp:maxData}
\mbgpEquQ is \bhk-hard data complexity for every $k>0$ and $\mathcal{S}\in \{\exists, \forall, \cap\}$.
Furthermore, \mwdptProjEquEx, \mwdptProjEqu, and \mwdptEquIn are in \bh data complexity.
\end{theoremAppendix}

\begin{proof}
    {\small$\blacktriangleright$} (Hardness of \mbgpEquQ)
    We reduce from the \bhtwok-complete problem $k$-\textsc{Pair-SAT}.
    An instance $\mathcal{I}$ consist of $k$ pairs $\texttti{p}_1,\dots, \texttti{p}_k=(\phi_1, \psi_1),\dots, (\phi_k, \psi_k)$ of propositional formula in 3-CNF, i.e., $\phi_i = \bigwedge^{m(\phi_i)}_{j=1} (\texttti{l}_{1j}(\phi_i) \lor \texttti{l}_{2j}(\phi_i) \lor \texttti{l}_{3j}(\phi_i))$ and $\psi_i = \bigwedge^{m(\psi_i)}_{j=1} (\texttti{l}_{1j}(\psi_i) \lor \texttti{l}_{2j}(\psi_i) \lor \texttti{l}_{3j}(\psi_i))$.
     $\mathcal{I}$ is a yes-instance if some pair $(\phi_i, \psi_i)$ is such that $\phi_i$ is unsatisfiable while $\psi_i$ is satisfiable.
    W.l.o.g, we assume the variables of all formulas to be pairwise disjoint.
    
    We extend our standard encoding for propositional formulas.
    Concretely, the nodes of the data graph $G$ consist of:
    \begin{itemize}
        \item All literals $\texttti{l}$.
        \item Clauses $\texttti{C}_{j}(\phi_i)= (\texttti{l}_{1j}(\phi_i) \lor \texttti{l}_{2j}(\phi_i) \lor \texttti{l}_{3j}(\phi_i))$ and $\texttti{C}_{j}(\psi_i)= (\texttti{l}_{1j}(\psi_i) \lor \texttti{l}_{2j}(\psi_i) \lor \texttti{l}_{3j}(\psi_i))$.
        \item Formulas $\phi_i,\psi_i$.
        \item The pairs $\texttti{p}_1, \dots, \texttti{p}_k$.
        \item The instance $\mathcal{I}$ itself.
        \item Auxiliary elements $\texttti{s}, \texttti{s}_{\phi_1}, \texttti{s}_{\psi_1}, \dots, \texttti{s}_{\phi_k}, \texttti{s}_{\psi_k}, \texttti{e}$.
    \end{itemize}
    Intuitively, we need the additional $\texttti{s}_{\phi_i}, \texttti{s}_{\psi_i}$ to increase the ``weight'' of some targets such that they then always have to be satisfied.    
    Literals $\texttti{l}$ are added to the class $Lit$, clauses $\texttti{C}$ to $Cl$, formulas $\phi_i,\psi_i$ to $Phi$, pairs $\texttti{p}$ to $Pair$, the auxiliary elements $\texttti{s}, \texttti{s}_{\phi_i}, \texttti{s}_{\psi_i}$ to $Start$, and the instance $\mathcal{I}$ to $Inst$.
    
    As many times before, we connect the nodes as they are connected logically:
    \begin{itemize}
        \item Dual literals $\texttti{l}$ and $\texttti{l}_d$ are connected via $dual(\texttti{l}, \texttti{l}_d)$.
        \item Literals $\texttti{l}_{kj}(\phi_i)$ (resp. $\texttti{l}_{kj}(\psi_i)$) are connected to clauses $\texttti{C}_{j}(\phi_i)$ (resp. $\texttti{C}_{j}(\psi_i)$) via $or(\texttti{l}_{kj}(\phi_i), \texttti{C}_{j}(\phi_i))$ (resp. $or(\texttti{l}_{kj}(\phi_i), \texttti{C}_{j}(\psi_i))$).
        \item Each clause $\texttti{C}_{j}(\phi_i)$ (resp. $\texttti{C}_{j}(\psi_i)$) is connected $\phi_{i}$ (resp. $\psi_i$) via $and(\texttti{C}_{j}(\phi_i), \phi_{i})$ (resp. $and(\texttti{C}_{j}(\psi_i), \psi_{i})$).
        \item Each pair $\texttti{p}_i$ is connected to its components $\phi_{i}$ and $\psi_{i}$ via $and_{\phi}(\phi_{i}, \texttti{p}_i)$ and $and_{\psi}(\psi_{i}, \texttti{p}_i)$, respectively.
        \item Furthermore, each pair $\texttti{p}_i$ is connected to $\mathcal{I}$ via $or(\texttti{p}_i, \mathcal{I})$.
    \end{itemize}
    
    Each literal $\texttti{l}$, clause $\texttti{C}$, formula $\phi_{i},\psi_{i}$, pair $\texttti{p}$, and the instance $\mathcal{I}$ is added to both the class $T$ and $F$, indicating that they can be either \textit{true} or \textit{false}.
    All nodes of $G$ except the $2k$ nodes $\texttti{s}_{\phi_i}, \texttti{s}_{\psi_i}$ are put into a single arbitrary linear order, where $next$ indicates immediate predecessors and successors, and which starts in $\texttti{s}$ and ends in $\texttti{e}$.
    The auxiliary nodes $\texttti{s}_{\phi_i}, \texttti{s}_{\psi_i}$ are added as predecessors of $\texttti{s}$, i.e., $next(\texttti{s}_{\phi_i}, \texttti{s}), next(\texttti{s}_{\psi_i}, \texttti{s})$.
    
    The set of constraints $\mathcal{C}$ then consists of:
    \begin{align*}
        \mathsf{lit} \leftrightarrow {} & Lit \land ((T \land \lnot F \land \exists dual.F) \\
         & \quad {} \lor (F  \land \lnot T \land \exists dual.T))  \\
        \mathsf{cl} \leftrightarrow {} & Cl \land {=_3}or^- \land {=_1}and \land (\\
        & (F \land \lnot T \land \forall or^-.F) \\
        & \lor (T \land \lnot F \land \exists or^-.T)  \\
        \mathsf{phi} \leftrightarrow {} & Phi \land ((F \land \lnot T \land \exists and^-.F) \\
        & \quad {} \lor (T \land \lnot F \land \forall and^-.T))  \\
        \mathsf{pair} \leftrightarrow {} & Pair \land {=_1}and^-_{\phi} \land {=_1}and^-_{\psi} \land {=_1}or \land (\\
        & ((F \land \lnot T \land ({=_1}and^-_{\phi}.T \lor {=_1}and^-_{\psi}.F ))\\
        & \lor (T \land \lnot F \land {=_1}and^-_{\phi}.F \land {=_1}and^-_{\psi}.T )) \\
        \mathsf{inst} \leftrightarrow {} & Inst \land ((F \land \lnot T \land \forall or^-.F) \\
            & \quad {} \lor (T \land \lnot F \land \exists or^-.T))  \\
        \mathsf{phiT} \leftrightarrow {} & Phi \land T  \\
        \mathsf{val} \leftrightarrow {} & \forall next^*.(\texttti{e} \lor Start \lor \mathsf{lit}\lor \mathsf{cl}\lor \mathsf{phi}\lor \mathsf{pair} \lor \mathsf{list}) \\
        & {}\land \exists next^*.\texttti{e}
    \end{align*}
    The targets are $\mathcal{T}=\mathcal{T}_1 \cup \mathcal{T}_2$ with 
    \begin{align*}
        \mathcal{T}_1 &=\{\mathsf{val}(\texttti{s})\} \cup \{\mathsf{val}(\texttti{s}_{\phi_i}),  \mathsf{val}(\texttti{s}_{\psi_i}) \mid i=1,\dots,k\},\\
        \mathcal{T}_2 & = \{\mathsf{phiT}(\phi_i), \mathsf{phiT}(\psi_i) \mid i =1,\dots,k\}.
    \end{align*}
    Furthermore, $H=\emptyset, \mu = \{x\rightarrow \mathcal{I}\}$, and the query is $Q(x)=T(x)$.

    Notice that there are $2k+1$ targets in $\mathcal{T}_1$ while there are only $2k$ targets in $\mathcal{T}_2$.
    Intuitively, the first $2k+1$ only asks the repair to remain true to the semantics of propositional logic, and thus, there are repairs that validate those targets.
    Moreover, if one target of $\mathcal{T}_1$ is validated by a repair, it may as well validate all targets of $\mathcal{T}_1$.
    Therefore, max-repairs must validate all of $\mathcal{T}_1$ and remain true to the semantics of proportional logic.
  
    Furthermore, repairs validating the targets $\mathcal{T}_1$, i.e., max-repairs, can choose the truth assignments for every variable freely, but the remaining truth values (of the literals, clauses, $\dots$) deterministically depend on this choice.
    A max-repair then validates $\mathsf{phiT}$ at $\phi_i$ (resp. $\psi_i$) if the choice of the truth assignment leads to the $\phi_i$ (resp. $\psi_i$) being evaluated to true under this assignment, i.e., the corresponding formula is satisfiable.
    Consequently, a repair maximizing the number of targets validated now must not only remain true to the semantics of propositional logic but must pick truth assignments that maximizes the number of formulas that evaluate to true.
    As the variables of all formulas are pairwise disjoint, every maximal repair sets all satisfiable formulas to true and all unsatisfiable formulas to false.
    Thus, in these repairs, $\mathcal{I}$ is set to true iff there exists a pair $(\phi_i, \psi_i)$ such that $\phi_i$ is unsatisfiable and $\psi_i$ is satisfiable. 
    Therefore, the reduction is correct no matter the semantics used for CQA.

    {\small$\blacktriangleright$} (Membership for \mwdptProjEquEx, \mwdptProjEqu)
    For data complexity, the targets $\mathcal{T}$ are considered fixed. 
    Hence, there are only constantly many subsets $\mathcal{T}'$ of $\mathcal{T}$ and with constantly many \np-oracle calls we can determine the maximal cardinality of a $\mathcal{T}'$ such that there is a repair that validates all targets in $\mathcal{T}'$.
    Then, given $\mathcal{K}=|\mathcal{T}'|$, answering the question whether $\mu$ is an answer of $Q$ of a / all repaired graphs validating $|\mathcal{T}'|$ targets only requires a single call to an \np-oracle.
    To see this, we can simply revisit the \np and \conp membership Full Proofs for \wdptProjEquEx and \wdptProjEqu.
    The only change required is that instead of guessing repairs $R$ together with supported models $I$ such that $\mathcal{T}\subseteq I$, we guess $R, I$ such that $\mathcal{T}'\subseteq I$ for some $\mathcal{T}' \subseteq \mathcal{T}, |\mathcal{T}'|=\mathcal{K},$. 
    
    {\small$\blacktriangleright$} (Membership for \mwdptEquIn)
    Let us first consider the easier case \mbgpProjEquIn.
    For this, we can proceed as before and first compute $\mathcal{K}=|\mathcal{T}'|$.
    Then, we can follow the membership proof of \bgpProjEquIn and insist that repairs need only validate $\mathcal{K}$ target.
    Thus, \mbgpProjEquIn is in \bh.
    
    For the case of \mwdptEquIn, recall the membership proof of \wdptEquIn.
    We can copy this proof for the present case, and as \mbgpProjEquIn is in \bh and $\mathsf{BH} = \mathsf{BH} \land \mathsf{co}$-$\mathsf{BH}$, \mwdptEquIn is also in \bh.
\end{proof}

\begin{theoremAppendix}
\label{thmApp:maxCombined}
\mbgpEquQ is \thetatwo-hard combined complexity for $\mathcal{S}\in \{\exists, \forall, \cap\}$.
Furthermore, \mbgpProjEquEx, \mwdptEqu, and \mbgpEquIn are in \thetatwo combined complexity.
\end{theoremAppendix}

\begin{proof}
    {\small$\blacktriangleright$} (Hardness of \mbgpEquQ)
    We reduce from the \thetatwo-complete problem \textsc{CardMinSat} \cite{DBLP:journals/lmcs/CreignouPW18}.
    Recall, a propositional formula $\phi(\texttti{x}_1,\dots, \texttti{x}_n)$ in 3-CNF is a yes-instance if the variable $\texttti{x}_1$ is true in some model minimizing the number of variables set to true.
    
    For the encoding, the nodes of the data graph $G$ consist of:
    \begin{itemize}
        \item All literals $\texttti{l}$.
        \item An additional copy $\texttti{x}'_i$ for each variable $\texttti{x}_i$.
        \item The clauses $\texttti{C}$ of $\phi$.
        \item The formula $\phi$ itself.
        \item $2n+1$ copies $\phi_1, \dots, \phi_{2n+1}$ of $\phi$.
        \item Auxiliary nodes $\texttti{s}, \texttti{s}_{1}, \dots \texttti{s}_{4n+3}, \texttti{e}$.
    \end{itemize} 
    Literals $\texttti{l}$ are added to the class $Lit$, clauses $\texttti{C}$ to $Cl$, the formula $\phi$ to $Phi$, copies $\texttti{x}'_i$ of variables to $CVar$, copies $\phi_i$ of $\phi$ to $CPhi$, and the auxiliary elements $\texttti{s}, \texttti{s}_{i}$ to $Start$.

    We connect the nodes as they are connected logically.
    \begin{itemize}
        \item Literals $\texttti{l}$ and there dual $ \texttti{l}^d$ are connected via $dual(\texttti{l}, \texttti{l}^d)$
        \item Literals are connected to clauses $\texttti{C}$ they appear in via $or(\texttti{l}, \texttti{C}$.
        \item Each clause $\texttti{C}$ is connected to $\phi$ via $and(\texttti{C}, \phi$.
        \item Copies $\texttti{x}'_i$ of variables and copies $\phi_i$ of $\phi$ are connected to their mirrors via $copy(\texttti{x}'_i, \texttti{x}_i)$ and $copy(\phi_i, \phi)$.
    \end{itemize}
    Each literal $\texttti{l}$, clause $\texttti{C}$, formula $\phi$, and copies $\texttti{x}'_i, \phi_i$ are added to both the class $T$ and $F$, indicating that they can be either \textit{true} or \textit{false}.
    All nodes of $G$ except the $4n+3$ nodes $\texttti{s}_{i}$ are put into a single arbitrary linear order, where $next$ indicates immediate predecessors and successors, and which starts in $\texttti{s}$ and ends in $\texttti{e}$.
    The auxiliary nodes $\texttti{s}_{i}$ are added as predecessors of $\texttti{s}$, i.e., $next(\texttti{s}_{i}, \texttti{s})$.

        The set of constraints $\mathcal{C}$ then consists of:
    \begin{align*}
        \mathsf{lit} \leftrightarrow {} & Lit \land ((T \land \lnot F \land \exists dual.F) \\
         & \quad {} \lor (F  \land \lnot T \land \exists dual.T))  \\
         \mathsf{cVar} \leftrightarrow {} & CVar \land ((T \land \lnot F \land \exists copy.T) \\
         & \quad {} \lor (F  \land \lnot T \land \exists copy.F))  \\
        \mathsf{cl} \leftrightarrow {} & Cl \land {=_3}or^- \land {=_1}and \land (\\
        & (F \land \lnot T \land \forall or^-.F) \\
        & \lor (T \land \lnot F \land \exists or^-.T) ) \\
        \mathsf{phi} \leftrightarrow {} & Phi \land ((\forall and^-.T \land T \land \lnot F) \\
        & \quad {}\lor (\exists and^-.F \land F \land \lnot T))\\
        \mathsf{cPhi} \leftrightarrow {} & CPhi \land ((T \land \lnot F \land \exists copy.T) \\
         & \quad {} \lor (F  \land \lnot T \land \exists copy.F))  \\
        \mathsf{phiT} \leftrightarrow {} & (Phi \lor CPhi) \land T  \\
        \mathsf{litT} \leftrightarrow {} & (Lit \lor CVar) \land T  \\
        \mathsf{litF} \leftrightarrow {} & (Lit \lor CVar) \land F  \\
        \mathsf{val} \leftrightarrow {} & \forall next^*.(\texttti{e} \lor Start \lor \mathsf{phi}\lor \mathsf{cPhi} \\
        & \quad \quad \quad {}\lor \mathsf{lit}\lor \mathsf{cVar} \lor \mathsf{cl}) \\
        & \land \exists next^*.\texttti{e}
    \end{align*}
    
    The targets are $\mathcal{T}=\mathcal{T}_1 \cup \mathcal{T}_2 \cup \mathcal{T}_3\cup \mathcal{T}_4$ with 
    \begin{align*}
        \mathcal{T}_1 & = \{\mathsf{val}(\texttti{s})\} \cup \{\mathsf{val}(\texttti{s}_{i}) \mid i=1,\dots, 4n+3\}\\
        \mathcal{T}_2 &= \{\mathsf{phiT}(\phi)\} \cup \{\mathsf{phiT}(\phi_i) \mid i=1,\dots,2n+1\} \\
        \mathcal{T}_3 &= \{\mathsf{litF}(\texttti{x}_i), \mathsf{litF}(\texttti{x}'_i) \mid i=1,\dots, n\} \\
        \mathcal{T}_4 &= \{\mathsf{litT}(\texttti{x}_1)\}
    \end{align*}
    Furthermore, $H=\emptyset, \mu = \{x\rightarrow \phi, y\rightarrow \texttti{x}_1\}$, and the query is $Q(x,y)=T(x)\land T(y)$.

    Notice that there are $4n+4$ targets in $\mathcal{T}_1$ and only $4n+3$ other targets.
    Again, these targets ask the repair to remain true to the semantics of propositional logic, and thus, there are repairs that validate those targets.
    Moreover, if one target of $\mathcal{T}_1$ is validated by a repair, it may as well validate all targets of $\mathcal{T}_1$.
    Therefore, max-repairs must validate all of $\mathcal{T}_1$ and remain true to the semantics of proportional logic.
    This means that repairs validating the targets $\mathcal{T}_1$ can choose the truth assignments for every variable freely, but the remaining truth values (of the literals, copies, $\dots$) deterministically depend on this choice.

    Next, let us consider $\mathcal{T}_2$.
    Notice that there are $2n+2$ targets in $\mathcal{T}_3$ and only $2n+1$ other targets (in $\mathcal{T}_3,\mathcal{T}_4$).
    These ask for $\phi$ to evaluate to true.
    Thus, by arguing similar as for the targets $\mathcal{T}_1$, if $\phi$ is satisfiable, all max-repairs must validate $\mathcal{T}_2$ and, consequently, represent models of $\phi$.

    We proceed with $\mathcal{T}_3$.
    Targets $\mathsf{litF}(\texttti{x}_i), \mathsf{litF}(\texttti{x}'_i)$ ask for $\texttti{x}_i$ to be set to false.
    Thus, if $\phi$ is satisfiable, max-repairs for the targets $\mathcal{T}_1 \cup \mathcal{T}_2 \cup \mathcal{T}_3$ represent the minimal models of $\phi$.

    Lastly, consider the remaining target $\mathcal{T}_4 = \{\mathsf{litT}(\texttti{x}_1)\}$.
    This target asks for $\texttti{x}_1$ to be set to true.
    However, this is a single target and thus, models of $\phi$ that set $\texttti{x}_1$ to true are only slightly prioritized over models that set $\texttti{x}_1$ to false. 
    
    The targets for the shape name $\mathsf{phiT}$ are then validated if the choice of the truth assignment leads to the formula $\phi$ being evaluated to true under this assignment, $\mathsf{litF}$ if the corresponding variable is set to false, and $\mathsf{litT}$ is $\texttti{x}_1$ is set to true.
    Thus, in total, if $\phi$ is satisfiable, max-repairs are minimal models of $\phi$, and, if there is a minimal model with $\texttti{x}_1$ being set to true, max-repairs are exactly those minimal models where $\texttti{x}_1$ is set to true.
    Consequently, this proves the correctness of reduction no matter the semantics used for CQA.

    {\small$\blacktriangleright$} (Membership for \mbgpProjEquEx, \mwdptEqu, and  \mbgpEquIn)
    We only need to slightly modify the decision procedure and its complexity analysis from the proof of the previous theorem.
    First, in combined complexity, the set $\mathcal{T}$ is no longer considered as fixed. 
    Hence, we now determine the maximum cardinality $\mathcal{K}$ of subsets $\mathcal{T}'$ of $\mathcal{T}$ such that there exists a repair of $G$ w.r.t., $((\mathcal{C},\mathcal{T}'), H)$.
    That is, by asking logarithmically many \np-questions of the form ``does there exists a subset $\mathcal{T}' \subseteq \mathcal{T}$ with $|\mathcal{T}'| \geq c$ such that there exists a repair of $G$ w.r.t. $((\mathcal{C},\mathcal{T}'), H)$?''
    
    After that, we only need one more oracle call to decide if $\mu$ is an answer to query $Q$ over some / all / the intersection of all repaired graphs of $G$ w.r.t.\ $((\mathcal{C},\mathcal{T}'), H)$, where $\mathcal{T}'\subseteq \mathcal{T}, |\mathcal{T}'| = \mathcal{K}$.
    For \mbgpProjEquEx, a NTW answering that oracle call has to guess such a repaired graph together with an instantiation of the quantified variables for the yes-instances.
    For \mwdptEqu, the NTW has to guess a repaired graph together with an instantiation of the variables in the OPT-part for the no-instances.
    For \mbgpEquIn, the NTW has to simply guess a repaired graph for the no-instances.
\end{proof}

\section{Reference Problems}
\label{sec:reference}

\problemdef{\textsc{CardMin-Precoloring}}
{A graph $G = (V,E)$ with $v_1\in V$, a precoloring $c\colon V\rightarrow 2^{\{\texttti{r},\texttti{g},\texttti{b}\}}$}
{Does there exist a 3-coloring of $G$ adhering to $c$ that colors $v_1$ the color $\texttti{g}$ while at the same time minimizes the use of $\texttti{g}$}
\begin{lemma}
\label{lem:minPre}
    \textsc{CardMin-Precoloring} is \thetatwo-hard for 3-colorable graphs with degree bounded by 5.
\end{lemma}

\begin{proof}[Proof sketch]
    We reduce from the \thetatwo-complete problem \textsc{CardMinSat} \cite{DBLP:journals/lmcs/CreignouPW18}.
    An instance is a propositional formula $\phi(x_1,\dots, x_n) = \bigwedge_iC_i = \bigwedge_i (l_{1i} \lor l_{2i} \lor l_{3i})$ in 3-CNF including the variable $x_1$.
    $\phi$ is a yes-instance if $x_1$ is true in some model minimizing the variables set to true.
    Now let us define $\psi = \phi \lor (\lnot x_1 \land x_2 \land \cdots \land x_{n+1})$.
    Then, $\psi$ has a minimal model setting $x_1$ to true if and only if $\phi$ has a minimal model setting $x_1$ to true.
    
    Now, applying the classical reduction from propositional formulas to \textsc{3-Colorability} to the formula $\psi(x_1,\dots,x_{n+1})$, we can get a graph $G$ such that $x_1,\dots,x_{n+1}, \phi\in V(G)$.
    Furthermore, interpreting $\texttti{g}$ as true and $\texttti{r}$ as false, we can ask for a graph $G$ such that a $col\colon \{x_1,\dots,x_{n+1}, \psi\} \rightarrow \{\texttti{r}, \texttti{g}\}$ is extendable to a coloring of the whole graph if and only if $col(\psi) = \psi(col(x_1,\dots,x_{n+1}))$.
    Thus, let us consider the precoloring $c\colon x_i\mapsto\{\texttti{r}, \texttti{g}\}, \phi \mapsto\{\texttti{g}\}$, and $v\mapsto\{\texttti{r}, \texttti{g}, \texttti{b}\}$ for all other $v\in V(G)$.
    Then, each coloring $col$ adhering to $c$ corresponds to exactly one model of $\psi$, and every model of $\psi$ corresponds to at least one coloring adhering to $c$.
    Notice that $(G,c)$ is 3-colorable.
    
    Two steps are left.
    First, we need to correspond minimal models of $\psi$ to colorings of $(G,c)$ minimizing the use of $\texttti{g}$.
    I.e., the colors of all but the variables $x_i$ have to be irrelevant.
    To achieve this, let us introduce a triangle $v,v',v''$ ($v',v''$ are new vertices) for each $v\in V(G),v\neq x_i$ and define $c(v')=c(v'')=\{\texttti{r}, \texttti{g}, \texttti{b}\}$.
    Thus, now, each coloring $col$ adhering to $c$ that minimizes the use of $\texttti{g}$ corresponds to exactly one minimal model of $\phi$ and every minimal model of $\phi$ corresponds to at least one coloring adhering to $c$ that minimizes the use of $\texttti{g}$.
    
    Second, we need to ensure that the degree of every vertex is at most 5.
    To achieve this, iteratively pick a vertex $v\in V(G)$ with neighbors $v_1,\dots,v_k,k\geq 6$.
    Then, introduce two new triangles $v,v',v''$ and $v',v'',v'''$, and connect $v_1,\dots,v_3$ to $v$ and $v_4,\dots,v_k$ to $v'''$.
    Furthermore, define $c(v')=c(v'')=\{\texttti{r}, \texttti{g}, \texttti{b}\}$ and $c(v''')=c(v)$.
    Thus, the relative ``sizes'' of the colorings remain the same as every coloring $col$ has to color exactly one of the new vertices $v',v'',v'''$ the color $\texttti{g}$.
    Furthermore, this does not change the colorability and $col(v)=col(v''')$.
    This completes the sketch of the reduction.
\end{proof}

\problemdef{\textsc{List-Pair-}$k$\textsc{-QBF}}
{A list $(\phi_1(X_1(\phi_1)\cdots X_k(\phi_1))$, $\psi_1(X_1(\psi_1)\cdots X_k(\psi_1)))$, $\dots$, $(\phi_n(X_1(\phi_n)\cdots X_k(\phi_n))$, $\psi_n(X_1(\psi_n)\cdots X_k(\psi_n)))$ of pairs of propositional formulas in 3-CNF}
{Does there exist a $i$ such that $QX_1(\phi_i)\cdots \forall X_{k-1}(\phi_i)\exists X_k(\phi_i) \phi_i$ and $\lnot QX_1(\psi_i)\cdots \forall X_{k-1}(\psi_i)\exists X_k(\psi_i) \psi_i$ are both valid, where $Q=\forall$ if $k$ is even and $Q=\exists$ otherwise}
We refer to \textsc{List-Pair-}$1$\textsc{-QBF} also by \textsc{List-Pair-SAT} (in the proof of Theorem \ref{thmApp:wdptProj} we swap the roles of $\phi$ and $\psi$ for the problem \textsc{List-Pair-SAT} as the reduction is then more natural).
\begin{lemma}
\label{lem:listPair}
    \textsc{List-Pair-}$k$\textsc{-QBF} is \thetak-c.
\end{lemma}
\begin{proof}[Proof sketch]
    We reduce from the canonical \thetak-complete problem \textsc{LogLexMax-}$k$\textsc{-QBF}.
    An instance is a quantified propositional formula $\Phi = \exists X_1\forall X_2\dots QX_k \phi(Y,X_1\cdots X_k)$ with free variables $Y=y_1,\dots,y_{m}$ with $m\leq \log|\Phi|$.
    This is a yes-instance if $y_m$ is true in the lexicographically (the order is given by the bit vectors $(\alpha(y_1), \dots,\alpha(y_m))$) maximal model of $\Phi(Y)$.
    Let us first consider the case where $k$ is odd.
    
    Let us define for each bitvector $B = (b_1,\dots,b_{m-1})$)
    \begin{align*}
    \Phi_{B} =&  \exists Y,X_1 \forall X_2\dots \exists X_k \bigvee_{B'\geq (B, 1)} Y=B \land \phi, \\
    \Psi_{B} =&  \exists Y,X_1 \forall X_2\dots \exists X_k \bigvee_{B' > (B, 1)} Y=B \land\phi. 
    \end{align*}   
    Preserving validity, these can clearly be transformed into formulas $\exists X_1(\phi_B)\cdots \forall X_{k-1}(\phi_B)\exists X_k(\phi_B) \phi_B$ and $\exists X_1(\psi_B)\cdots \forall X_{k-1}(\psi_B)\exists X_k(\psi_B) \psi_B,$ respectively, where $\phi_B$ and $\psi_B$ are in 3-CNF.
    (In fact, $X_i=X_i(\psi_B) = X_i(\phi_B)$ for $i\neq 1,k$.)
    Thus, we reduce $\Phi$ to the list $(\phi_B(X_1(\phi_B)\cdots X_k(\phi_B)), \psi_B(X_1(\psi_B)\cdots X_k(\psi_B)))_B$ of length $2^{m-1}\leq |\Phi|$.
    Note that $y_m$ is true in the lexicographically maximal model of $\Phi$ if and only if there exists a $B$ such that both $\Phi_B$ and $\lnot \Psi_B$ are valid, thus, entailing the correctness of the reduction.

    In the case where $k$ is odd, we proceed as follows.
    For each bitvector $B = (b_1,\dots,b_{m-1})$), we again define
    \begin{align*}
    \Phi_{B} =&  \exists Y,X_1 \forall X_2\dots \forall X_k \bigvee_{B'\geq (B, 1)} Y=B \land \phi, \\
    \Psi_{B} =&  \exists Y,X_1 \forall X_2\dots \forall X_k \bigvee_{B' > (B, 1)} Y=B \land\phi. 
    \end{align*}   
    Preserving validity, these can clearly be transformed into formulas $\lnot \forall X_1(\phi_B)\cdots \forall X_{k-1}(\phi_B)\exists X_k(\phi_B) \phi_B$ and $\lnot \forall X_1(\psi_B)\cdots \forall X_{k-1}(\psi_B)\exists X_k(\psi_B) \psi_B,$ respectively, where $\phi_B$ and $\psi_B$ are in 3-CNF.
    Thus, we reduce $\Phi$ to the list $(\psi_B(X_1(\psi_B)\cdots X_k(\psi_B)), \phi_B(X_1(\phi_B)\cdots X_k(\phi_B)))_B$ of length $2^{m-1}\leq |\Phi|$ (the positions of $\phi_B$ and $\psi_B$ have changed).
    Note that $y_m$ is true in the lexicographically maximal model of $\Phi$ if and only if there exists a $B$ such that both $\lnot \Psi_B$ and $\Phi_B$ are valid, thus, entailing the correctness of the reduction.
\end{proof}

\end{document}